\documentclass[runningheads,envcountsect,envcountsame]{llncs}
\usepackage[T1]{fontenc}
\usepackage[english]{babel}
\usepackage{graphicx}
\usepackage{tikz}
\usepackage{bussproofs}
\usepackage{amssymb}
\usepackage{amsmath}
\usepackage{relsize}
\usepackage{hyperref}
\usepackage{cleveref}
\usepackage{mathrsfs}
\usepackage{stmaryrd}
\usepackage[shortlabels]{enumitem}
\usepackage{multirow}
\usepackage{xspace}
\usepackage{virginialake}
\usepackage{adjustbox}
\usepackage{thm-restate}

\newcommand{\nb}[1]{\textcolor{red}{$|$}\mbox{}\marginpar{\scriptsize\raggedright\textcolor{red}{#1}}}
\newcommand{\nbm}[1]{\textcolor{orange}{$|$}\mbox{}\marginpar{\scriptsize\raggedright\textcolor{orange}{M: #1}}}

\newcommand{\G}{\Gamma}
\newcommand{\D}{\Delta}
\newcommand{\seq}{\Rightarrow}
\newcommand{\sseq}{\Leftrightarrow}
\newcommand{\seqder}[1]{\overset{\raise.3em\hbox{$\varDer_{#1}$}}{\seq}}
\newcommand{\sseqder}[2]{\overset{\raise.3em\hbox{$\varDer_{#1}, \varDer_{#2}$}}{\sseq}}
\newcommand{\bigAND}{\bigwedge}

\newcommand{\vd}{\vdash}
\newcommand{\Vd}{\Vdash}

\newcommand{\mf}{\mathsf}

\newcommand{\ih}{i.h.}
\newcommand{\varS}{\mathcal{S}}

\newcommand{\varDer}{\mathcal{D}}

\newcommand{\wij}{Wijesekera}
\newcommand{\olk}{Olkhovikov\xspace}

\newcommand{\intu}{intuitionistic}
\newcommand{\const}{constructive}

\newcommand*{\IMP}{\mathbin{\to}}
\newcommand*{\IMPP}{\mathbin{\leftrightarrow}}
\newcommand*{\BOX}{\mathord{\Box}}
\newcommand*{\DIA}{\mathord{\Diamond}}
\newcommand*{\AND}{\mathbin{\wedge}}
\newcommand*{\OR}{\mathbin{\lor}}

\newcommand{\condDiam}{\ensuremath{%
  \mathbin{\mathsmaller{\Diamond}\kern-1.5pt\raise0.5pt\hbox{$\mathord{\mathsmaller\rightarrow}$}}}}
  
\newcommand{\condBox}{\ensuremath{%
  \mathbin{\mathsmaller{\Box}\kern-1.5pt\raise1pt\hbox{$\mathord{\mathsmaller\rightarrow}$}}}}
  
\newcommand{\condCirc}{\ensuremath{%
  \mathbin{\circ\kern-1.5pt\raise1pt\hbox{$\mathord{\mathsmaller\rightarrow}$}}}}

\newcommand{\condBlackDiam}{\ensuremath{%
  \mathbin{\mathsmaller{\blacklozenge}\kern-1.5pt\raise0.5pt\hbox{$\mathord{\mathsmaller\rightarrow}$}}}}

\newcommand{\cb}{\condBox}
\newcommand{\cd}{\condDiam}
\newcommand{\cc}{\condCirc}
\newcommand{\cbd}{\cd_\bot}

\newcommand{\imp}{\to}

\newcommand{\biimp}{\leftrightarrow}
\newcommand{\diam}{\Diamond}
\newcommand{\lan}{\mathcal{L}}
\newcommand{\atm}{Atm}
\newcommand{\logic}{\logicnamestyle L}
\newcommand{\fint}{i}
\newcommand{\fin}[1]{i(#1)}
\newcommand{\finseq}{\iota}

\newcommand{\lanB}{\lan^{\cb}}
\newcommand{\lanBD}{\lan}
\newcommand{\lanm}{\lan_m}
\newcommand{\lancm}{\lan_{cm}}

\newcommand{\logicnamestyle}[1]{\ensuremath{\mathsf{#1}}}
\newcommand{\K}{\logicnamestyle{K}}

\newcommand{\CPL}{\logicnamestyle{CPL}}
\newcommand{\IPL}{\logicnamestyle{IPL}}
\newcommand{\WK}{\logicnamestyle{WK}}
\newcommand{\sfour}{\logicnamestyle{S4}}

\newcommand{\CK}{\logicnamestyle{CK}}
\newcommand{\intCK}{\logicnamestyle{IntCK}}
\newcommand{\constCK}{\logicnamestyle{ConstCK}}
\newcommand{\IK}{\logicnamestyle{IK}}
\newcommand{\iK}{\logicnamestyle{iK}}
\newcommand{\ICK}{\intCK}
\newcommand{\CCK}{\constCK}
\newcommand{\IntCK}{\intCK}
\newcommand{\ConstCK}{\constCK}
\newcommand{\CKbox}{\logicnamestyle{CK}^{\cb}}
\newcommand{\constCKbox}{\logicnamestyle{ConstCK}^{\cb}}
\newcommand{\CCKbox}{\constCKbox}
\newcommand{\ConstCKbox}{\constCKbox}
\newcommand{\CKsfour}{\CK \oplus \sfour}
\newcommand{\CCKID}{\logicnamestyle{ConstCKID}}
\newcommand{\CCKMP}{\logicnamestyle{ConstCKMP}}
\newcommand{\CCKMPID}{\logicnamestyle{ConstCKMPID}}
\newcommand{\CCKCEM}{\logicnamestyle{ConstCKCEM}}
\newcommand{\CCKstar}{\logicnamestyle{ConstCK^*}}

\newcommand{\NIntCK}{\logicnamestyle{N}.\IntCK}

\newcommand{\SCK}{\logicnamestyle{S}.\CK}
\newcommand{\SConstCK}{\logicnamestyle{S}.\ConstCK}
\newcommand{\SConstCKbox}{\logicnamestyle{S}.\ConstCKbox}
\newcommand{\SWK}{\logicnamestyle{S}.\WK}
\newcommand{\SCCK}{\SConstCK}
\newcommand{\SCCKbox}{\SConstCKbox}
\newcommand{\SCCKID}{\logicnamestyle{S}.\CCKID}
\newcommand{\SCCKMP}{\logicnamestyle{S}.\CCKMP}
\newcommand{\SCCKMPID}{\logicnamestyle{S}.\CCKMPID}
\newcommand{\SCCKCEM}{\logicnamestyle{S}.\CCKCEM}
\newcommand{\SCCKstar}{\logicnamestyle{S}.\CCKstar}

\newcommand{\ccm}{\textsf{ccm}}
\newcommand{\W}{W}
\newcommand{\V}{V}
\newcommand{\M}{M}
\newcommand{\ltrset}{\llbracket}
\newcommand{\rtrset}{\rrbracket}
\newcommand{\trset}[1]{\ltrset{#1}\rtrset}

\newcommand{\less}{\futs}
\newcommand{\more}{\geq}

\newcommand{\pow}{\mathcal P}

\newcommand{\futs}{\leq}
\newcommand{\R}{R}
\newcommand{\Rof}[1]{R_{#1}}
\newcommand{\sat}{\Vdash}
\newcommand{\truthset}{\trset}

\newcommand{\ax}{\AxiomC}
\newcommand{\uinf}{\UnaryInfC}
\newcommand{\binf}{\BinaryInfC}
\newcommand{\tinf}{\TrinaryInfC}
\newcommand{\qinf}{\QuaternaryInfC}
\newcommand{\llab}{\LeftLabel}
\newcommand{\rlab}{\RightLabel}
\newcommand{\disp}{\DisplayProof}

\newcommand{\bl}[1]{{#1}^\bullet}
\newcommand{\wh}[1]{{#1}^\circ}

\newcommand{\NS}[2]{[{#1} : {#2}]} 
\newcommand{\GC}[1]{\G\{#1\}} 
\newcommand{\GCp}[1]{\G'\{#1\}} 
\newcommand{\LC}[1]{\Lambda\{#1\}} 
\newcommand{\GCD}[1]{\G^{\downarrow}\{#1\}} 
\newcommand{\depth}[1]{\mathit{depth}(#1)}

\newcommand{\leftrule}[1]{${#1}_\mathsf{L}$}
\newcommand{\rightrule}[1]{${#1}_\mathsf{R}$}
\newcommand{\init}{$\mf{init}$}

\newcommand{\limp}{\leftrule{\imp}}
\newcommand{\rimp}{\rightrule{\imp}}
\newcommand{\lland}{\leftrule{\land}}
\newcommand{\rland}{\rightrule{\land}}
\newcommand{\llor}{\leftrule{\lor}}

\newcommand{\rlorone}{${\lor}_\mathsf{R}^1$}
\newcommand{\rlortwo}{${\lor}_\mathsf{R}^2$}
\newcommand{\lwk}{\rulesfont{w_L}}
\newcommand{\rwk}{\rulesfont{w_R}}

\newcommand{\rulecb}{$\cb$}
\newcommand{\rulecd}{$\cd$}
\newcommand{\rulecbd}{$\cbd$}

\newcommand{\rulempbox}{$\mathsf{mp}_{\Box}$}
\newcommand{\rulempboxstar}{$\mathsf{mp_{\Box}^*}$}
\newcommand{\rulempdiam}{$\mathsf{mp}_{\diam}$}

\newcommand{\rulecbid}{$\cb^\mathsf{id}$}
\newcommand{\rulecdid}{$\cd^\mathsf{id}$}
\newcommand{\rulecbdid}{$\cbd^\mathsf{id}$}
\newcommand{\rulecbdcem}{$\cbd^\mathsf{cem}$}
\newcommand{\rulecdcem}{$\cd^\mathsf{cem}$}

\newcommand{\hilbertaxiomstyle}[1]{{#1}}

\newcommand{\MP}{\hilbertaxiomstyle{MP}}
\newcommand{\RAbox}{\hilbertaxiomstyle{RA}$_\Box$}
\newcommand{\RCbox}{\hilbertaxiomstyle{RC}$_\Box$}
\newcommand{\RAdiam}{\hilbertaxiomstyle{RA}$_\diam$}
\newcommand{\RCdiam}{\hilbertaxiomstyle{RC}$_\diam$}
\newcommand{\RMbox}{\hilbertaxiomstyle{RM}$_\Box$}
\newcommand{\RMdiam}{\hilbertaxiomstyle{RM}$_\diam$}

\newcommand{\axCMbox}{\hilbertaxiomstyle{CM}$_\Box$}
\newcommand{\axCCbox}{\hilbertaxiomstyle{CC}$_\Box$}
\newcommand{\axCNbox}{\hilbertaxiomstyle{CN}$_\Box$}

\newcommand{\axCMdiam}{\hilbertaxiomstyle{CM}$_\diam$}
\newcommand{\axCCdiam}{\hilbertaxiomstyle{CC}$_\diam$}
\newcommand{\axCNdiam}{\hilbertaxiomstyle{CN}$_\diam$}
\newcommand{\axCKdiam}{\hilbertaxiomstyle{CK}$_\diam$}
\newcommand{\axCFS}{\hilbertaxiomstyle{CFS}}
\newcommand{\axCW}{\hilbertaxiomstyle{CW}}
\newcommand{\axdefdiam}{\hilbertaxiomstyle{def}$_\diam$}

\newcommand{\axIDbox}{\hilbertaxiomstyle{ID}$_\Box$}
\newcommand{\axMPbox}{\hilbertaxiomstyle{MP}$_\Box$}
\newcommand{\axMPdiam}{\hilbertaxiomstyle{MP}$_\diam$}
\newcommand{\axCEMdiam}{\hilbertaxiomstyle{CEM}$_\diam$}
\newcommand{\axCEMbox}{\hilbertaxiomstyle{CEM}$_\Box$}

\newcommand{\pset}[1]{[#1]}

\newcommand{\ceset}[1]{\langle{#1}\rangle}
\newcommand{\ce}{\emph{ce}}
\newcommand{\A}{A}
\newcommand{\B}{B}

\newcommand{\ceR}{\mathscr{R}}

\newcommand{\ceU}{\mathscr{U}}

\newcommand{\Wc}{W}
\newcommand{\Vc}{V}
\newcommand{\Mc}{M}
\newcommand{\Rc}{R}
\newcommand{\Rcof}{\Rof}
\newcommand{\lessc}{\leq}
\newcommand{\fcbm}[1]{\varphi \cb ^-(#1)}

\newcommand{\rr}{\mathsf{r}}
\newcommand{\derivesIntCk}{\vd_{\IntCK}}

\newcommand{\bcont}{$\bullet$-context\xspace}
\newcommand{\wcont}{$\circ$-context\xspace}

\newcommand{\rulesfont}[1]{\mathsf{#1}}
\newcommand{\wk}{\rulesfont{w}}
\newcommand{\ctr}{\rulesfont{c}}
\newcommand{\cut}{\rulesfont{cut}}
\newcommand{\necs}{[\rulesfont{nec}]}

\newcommand{\rep}{\rulesfont{rep}}
\newcommand{\meds}{[\rulesfont{m}]}

\newcommand{\brl}[1]{\bl{#1}}
\newcommand{\wrl}[1]{\wh{#1}}
\newcommand{\initv}{\mathsf{init}}

\newcommand{\w}[1]{{w}(#1)}
\newcommand{\rk}[1]{r(#1)}

\newcommand{\nseq}{\Gamma}
\newcommand{\nder}{\mathcal{D}}

\newcommand{\NIntCKcut}{\logicnamestyle{N}.\IntCK\cup\{\cut\}}
\newcommand{\NIntCKcutrep}{\logicnamestyle{N}.\IntCK\cup\{\cut,\rep\}}

\newcommand{\ps}{\mathcal{P}}

\newcommand{\rg}{\mathsf{R}}
\newcommand{\height}[1]{h(#1)}

\newcommand{\seqdern}[2]{\overset{\raise.4em\hbox{$\varDer_{#1}$}}{$#2$}}

\newcommand{\marianna}[1]{\textcolor{orange}{[\![M: #1]\!]}}
\newcommand{\tiz}[1]{\textcolor{purple}{[\![T: #1]\!]}}

\newcommand\T{\rule{0pt}{2.6ex}}       
\newcommand\TT{\rule{0pt}{3ex}}       
 
\begin{document}
\title{A 
	Proof-Theoretic 
	View of Basic Intuitionistic Conditional Logic (Extended Version)}
%
%
\author{Tiziano Dalmonte\inst{1}
	\and
Marianna Girlando\inst{2}
}
\authorrunning{T. Dalmonte and M. Girlando}
%
\institute{Free University of Bozen-Bolzano, Bolzano, Italy \and
University of Amsterdam, Netherlands\\
}
\maketitle              
\begin{abstract}
	Intuitionistic conditional logic, studied by Weiss, Ciardelli and Liu, and Olkhovikov, aims at providing a constructive analysis of conditional reasoning. 
	In this framework, the \emph{would}  and the \emph{might} conditional operators are no longer interdefinable. 
	The intuitionistic conditional logics considered in the literature are defined by setting Chellas' conditional logic $\CK$, whose semantics is defined using selection functions, within the constructive and intuitionistic framework introduced for intuitionistic modal logics. 
	This operation gives rise to a constructive and an intuitionistic variant of (might-free-) $\CK$, which we call $\CCKbox$ and $\IntCK$ respectively. 
	Building on the proof systems defined for $\CK$ and for intuitionistic modal logics, in this paper we introduce a nested calculus for $\IntCK$ and a sequent calculus for $\CCKbox$.  
	Based on the sequent calculus, we define $\CCK$, a conservative extension of Weiss' logic $\CCKbox$ with the might operator. We introduce a class of models and an axiomatization for $\CCK$, and extend these result to several extensions of $\CCK$. 
%

\end{abstract}

\section{Introduction}

Intuitionistic%
\footnote{A short version of this paper was accepted for presentation at TABLEAUX 2025.} conditional logic
aims at providing a constructive analysis of conditional reasoning
by combining 
the \intu\ implication ($\IMP$) with the  \emph{would}  $\cb$ and the \emph{might}  $\cd$ conditional operators. 
These modalities have been extensively studied from the 1970s, and several possible-world semantics have been proposed for them, including \emph{selection function}~\cite{chellas:1975}, \emph{sphere models}~\cite{lewis:1973} and \emph{preferential semantics}~\cite{burgess:1981}. 
Conditional operators have mostly been studied in a classical setting; the goal is now to interpret them constructively, providing a formal and modular framework for their analysis. 
This line of research shares several insights with the definition of \emph{intuitionistic modal logics}, where the intuitionistic implication is combined with the modalities $\Box$ and $\Diamond$. In this direction,   \emph{constructive} variants of modal logic $\K$ have been defined, inspired by a computational interpretation~\cite{alechina:etal:01,wij:1990}, while \emph{intuitionistic} versions of $\K$ rely on the correspondence with first-order intuitionistic logic~\cite{fischerservi:1984,simpson:1994}. 

The study of intuitionistic variants of conditional logics is a relatively new field of research. 
This analysis started with the study of
\intu\ variants of
Chellas' `basic' conditional logic $\CK$~\cite{chellas:1975}. 
This logic\footnote{In the literature, $\CK$  also denotes the constructive version of $\K$. We here use $\CK$ for Chellas' logic, and use the prefix $\mathsf{Const}$ for the `constructive' version of a logic.}, one of the simplest normal conditional logics, constitutes the `core' of several other systems, including preferential conditionals~\cite{burgess:1981} and Lewis' counterfactuals~\cite{lewis:1973}. 
The first \intu\ counterpart of $\CK$ was the system $\mathsf{ICK}$ (which we call $\CCKbox$)
proposed by Weiss in~\cite{weiss:2019}
and extended by Ciardelli and Liu in~\cite{ciardelli:2020} with additional conditional axioms.
These logics are defined in a language 
with $\cb$
as the only conditional modality.
While in classical conditional logics
the {might} operator $\cd$ is definable via the relation
$(\varphi \cd \psi) \biimp \neg (\varphi \cb \neg\psi)$ (\cite{lewis:1973}),
the same 
is generally assumed not to hold intuitionistically. 
Weiss, Ciardelli and Liu leave open the question of how to extend  $\ConstCKbox$ with the might operator $\cd$. 




An answer was given by \olk\ in~\cite{olk:2024}, with the definition of logic $\IntCK$, inspired by the Fisher Servi-style \emph{bi-relational semantics} for intuitionistic modal logic. 
While $\IntCK$ supports a meaningful interpretation of $\cd$,
it validates $\cb$-principles that are not valid in $\ConstCKbox$,
so that 
$\IntCK$ is \emph{not} a conservative extensions of Weiss' system. 
This is analogous to what happens for intuitionistic modal logics where, as pointed out in~\cite{das:2023},  logic $\iK$, the $\Box$-fragment of constructive $\K$~\cite{bozic:1984,wolter:1999} differs from $\IK$, the intuitionistic variant of modal logic comprising both $\Box$ and $\Diamond$~\cite{fischerservi:1984,simpson:1994}.  
And indeed, as observed by Weiss~\cite{weiss:2019} and \olk~\cite{olk:2024},
by defining  $\Box$ and $\diam$ in terms of conditionals (through $\top \cb \varphi$ and $\top \cd \varphi$, respectively)
one can reconstruct  $\iK$ and $\IK$ within $\CCKbox$ and $\ICK$, respectively.

Very few proof-theoretical accounts of \intu\ conditionals exist in the literature.
A labelled calculus for a constructive conditional logic representing access control policies is introduced in~\cite{genovese:2014}, while  a natural deduction system supporting  a BHK interpretation of counterfactual inferences is defined in~\cite{wik:2025}. 


In this paper we explore the proof theory of  intuitionistic conditional logic, and we propose an alternative answer to the question posed by Weiss, Ciardelli and Liu. 
We first present a nested sequent calculus $\NIntCK$ for $\IntCK$.
This calculus is defined by suitably combining Stra\ss burger's nested calculus for $\IK$ \cite{strassburger:2013}
and Alenda, Olivetti and Pozzato's nested calculus for $\CK$ \cite{alenda:2012,alenda:2016}.
In particular, $\NIntCK$ inherits from \cite{strassburger:2013} the treatment of the intuitionistic base,
and from \cite{alenda:2012,alenda:2016} the indexing of nested components to
 treat antecedents of conditional sentences. 

We  then introduce $\SCCKbox$, a sequent calculus for $\CCKbox$ based on Pattinson and Schröder's  calculus for classical $\CK$ \cite{pattinson:2011}. 
While this calculus does not require to enrich the structure of Gentzen-style sequents, the number of premisses of the conditional rules is bounded by the number of conditional operators appearing in the conclusion. $\SCCKbox$ is defined by restricting Pattinson and Schröder's calculus to  single-succedent sequents. 

Proof-theory serves as a guide to integrate the might operator into $\CCKbox$. We first define the classical rules for $\cd$, which  are duals of the rules for $\cb$ in Pattinson and Schröder's calculus. By  adding the single-succedent restriction, we obtain a proof-theoretic characterization of a new system, which we call $\CCK$, comprising both $\cb$ and $\cd$. We propose an axiomatisation and a class of models for $\CCK$. Since the $\cd$-free fragment of $\CCK$ corresponds to $\CCKbox$, this logic represent a natural answer to the question of how to integrate the $\cd$ modality in $\CCKbox$. 
We conclude by discussing extensions of $\CCK$ with 
\emph{identity}, \emph{conditional modus ponens} and \emph{conditional excluded middle}.

The paper is structured as follows. In \Cref{sec:preliminaries} we introduce axiomatisations and semantics for intuitionistic conditional logics. Then, in \Cref{sec:nested} we present the nested calculus $\NIntCK$ for $\IntCK$, and in \Cref{sec:basic} we introduce the sequent calculus $\SCCKbox$ for $\CCKbox$, and the new logic $\CCK$. Finally, in \Cref{sec:ext} we discuss extensions of $\CCK$, and we conclude with directions for future works in~\Cref{sec:con}. Full proofs of our results are reported in the Appendix.  


%

%
%
%
%
%

\section{Preliminaries}
\label{sec:preliminaries}
Fixed a countable set of propositional atoms $\atm$, the formulas of $\lanBD$ are generated as follows, for $p \in \atm$:
$
\varphi ::= p \mid \bot \mid \varphi \AND \varphi  \mid \varphi \OR \varphi \mid \varphi \imp \varphi \mid \varphi \cb \varphi \mid \varphi \cd \varphi
$. 
We assume the convention that
$\land, \lor$ bind more strongly than $\imp, \cb, \cd$.
We define $\neg \varphi := \varphi\imp\bot$, $\top := \lnot \bot$ and $\varphi \biimp \psi := (\varphi \imp \psi) \land (\psi \imp \varphi)$.
The formulas of $\lanB$ are defined as those of $\lanBD$, but by removing the operator $\cd$.
For $\cc \in \{\cb, \cd\}$, we call $\cc$-formula (or \emph{conditional} formula) any formula whose main connective is $\cc$. 
We use capital Latin letters $A, B, C, ...$ to denote \emph{sets} of formulas,
and capital Greek letters $\G, \D, \Sigma, ...$ to denote \emph{multisets} of formulas.


\begin{figure}[t!]
\centering

\begin{multicols}{2}
\begin{tabular}{ll}
\axCMbox\ \ & $(\varphi\cb \psi\land\chi ) \imp (\varphi \cb \psi)\land(\varphi\cb\chi)$ \\
\axCCbox & $(\varphi \cb \psi)\land(\varphi\cb\chi) \imp (\varphi\cb \psi\land\chi)$ \\
\axCNbox\ & $\varphi \cb \top$ \\
\axCMdiam\ \ & $(\varphi \cd \psi)\lor(\varphi\cd\chi) \imp (\varphi\cd \psi\lor\chi)$ \\
\axCCdiam & $(\varphi\cd \psi\lor\chi) \imp (\varphi \cd \psi)\lor(\varphi\cd\chi)$ \\
\axCNdiam & $\neg(\varphi \cd \bot)$ \\
\axCW & $(\varphi \cd \psi)\land(\varphi\cb\chi) \imp (\varphi\cd \psi\land\chi)$ \\
\axCFS & $((\varphi \cd \psi)\imp(\varphi\cb\chi)) \imp (\varphi\cb (\psi\imp\chi))$ \\
\axCKdiam & $(\varphi \cb (\psi \imp \chi)) \imp ((\varphi \cd \psi) \imp (\varphi \cd \chi))$ \\
\axdefdiam\ & $(\varphi \cd \psi) \biimp \neg (\varphi \cb \neg\psi)$ \\
\end{tabular}

\columnbreak

\begin{tabular}{l}
\ax{$\varphi \biimp \rho$}
\llab{\RAbox}
\uinf{$(\varphi\cb\psi) \biimp (\rho\cb\psi)$}
\disp \\ 

\\[-0.5em]

\ax{$\psi \biimp \chi$}
\llab{\RCbox}
\uinf{$(\varphi\cb\psi) \biimp (\varphi\cb\chi)$}
\disp \\

\\[-0.5em]

\ax{$\varphi \biimp \rho$}
\llab{\RAdiam}
\uinf{$(\varphi\cd\psi) \biimp (\rho\cd\psi)$}
\disp \\

\\[-0.5em]

\ax{$\psi \biimp \chi$}
\llab{\RCdiam}
\uinf{$(\varphi\cd\psi) \biimp (\varphi\cd\chi)$}
\disp \\
\end{tabular}
\end{multicols}

\begin{tabular}{lllll}
\emph{System name} \ \ & \emph{Language} \ \ & \emph{Prop. base} \ \ & \ \emph{Conditional axioms and rules} 
 \\
\hline
\ $\CKbox$ & \ \ \ $\lanB$ & \ \ \ $\CPL$ & \RAbox, \RCbox, \axCMbox, \axCCbox, \axCNbox. 
\TT \\
\ $\CK$ & \ \ \ $\lan$ & \ \ \ $\CPL$ & \RAbox, \RCbox, \axCMbox, \axCCbox, \axCNbox, \axdefdiam. \T \\
\ $\ConstCKbox$ & \ \ \ $\lanB$ & \ \ \ $\IPL$ & \RAbox, \RCbox, \axCMbox, \axCCbox, \axCNbox. \T \\
\ $\IntCK$ & \ \ \ $\lan$ & \ \ \ $\IPL$ & \RAbox, \RCbox, \RAdiam, \RCdiam, \axCMbox, \axCMdiam, \axCCbox, \T \\
&&& \axCCdiam, \axCNbox, \axCNdiam, \axCW, \axCFS. \\
\end{tabular}
\caption{\label{fig:axioms}Conditional axioms and rules, and axiomatic systems.}
\end{figure}

The classical (resp. intuitionistic) conditional logics we consider are defined extending 
any axiomatisation of classical  (resp. intuitionistic) propositional logic $\CPL$ (resp. $\IPL$) including the \emph{modus ponens} rule (\MP), with suitable axioms and rules, as summarised in Fig. \ref{fig:axioms}. 


As usual, we say that a rule 
$A / \varphi$
 is \emph{derivable} in a logic $\logic$
if there is a finite sequence of formulas $\psi_1, ..., \psi_n$ such that $\psi_n = \varphi$
and for all $1 \leq i \leq n$, $\psi_i \in \A$, or it is an (instance of an) axiom of $\logic$, or it is obtained from previous formulas by the application of a rule of $\logic$.
Moreover, a formula $\varphi$ is \emph{derivable} in $\logic$ (written $\vd_{\logic} \varphi$), if
$\emptyset / \varphi$
is derivable.
Finally, we say that $\psi$ is derivable in $\logic$ from a set of formulas $\A$  (written $A \vd_{\logic} \varphi$)
if there are $\psi_1, ..., \psi_n \in \A$ such that $\vd_{\logic} \psi_1 \land ... \land \psi_n \imp \varphi$. 


Logic $\ConstCKbox$ was introduced in~\cite{weiss:2019,ciardelli:2020} (under the name $\mathsf{ICK}$), while $\IntCK$ is from~\cite{olk:2024}. In~\cite{olk:2024}, \olk shows that if $\varphi$ is derivable in $\ConstCKbox$, then it is derivable in $\IntCK$. 
However, similarly to what happens in the intuitionistic modal  case, $\IntCK$ is a non-conservative extension of $\CCKbox$:  formula $ \lnot \lnot (\top \cb \bot) \IMP (\top \cb \bot) $ is derivable in the  $\cd$-free fragment of $\IntCK$, but it is \emph{not} derivable in $\ConstCKbox$~\cite{olk:2024}. 

We introduce the bi-relational Kripke models for $\ConstCKbox$ and $\IntCK$ 
from \cite{weiss:2019} and  \cite{olk:2024} respectively. 
We first present some general definitions. 

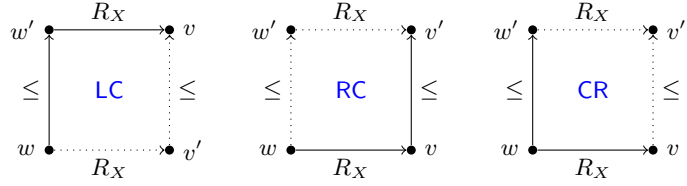
\begin{figure}[t!]\label{fig:fcbc}
	\begin{center}
			\begin{tikzpicture}[scale=0.8]
			\tikzstyle{node}=[circle,fill=black,inner sep=1.2pt]
			
			\node[label= left :{\small{$w$}}] (1) at (0,0) [node] {};
			\node[label= right :{\small{$v'$}}] (2) at (2,0) [node]
			{};
			\node[label=  right :{\small{$v$}}] (3) at (2,2) [node]
			{};
			
			\node[label= left :{\small{$w'$}}] (4) at (0, 2) [node]
			{};
			\node[][] at (1,1){\textsf{\textcolor{blue}{LC}}};

			\draw[->,dotted] (1) edge[below] node {\footnotesize{$R_X$}} (2);
			\draw[->, dotted] (2) edge[right] node{\footnotesize{$\leq$}} (3);
			\draw[-> ] (1) edge[left] node{\footnotesize{$\leq$}}  (4);
			\draw[->] (4) edge[above] node {\footnotesize{$R_X$}} (3);
			
		\end{tikzpicture}
		\quad 
		\begin{tikzpicture}[scale=0.8]
			\tikzstyle{node}=[circle,fill=black,inner sep=1.2pt]
			
			\node[label= left :{\small{$w$}}] (1) at (0,0) [node] {};
			\node[label= right :{\small{$v$}}] (2) at (2,0) [node]
			{};
			\node[label=  right :{\small{$v'$}}] (3) at (2,2) [node]
			{};
			
			\node[label= left :{\small{$w'$}}] (4) at (0, 2) [node]
			{};
			\node[][] at (1,1){\textsf{\textcolor{blue}{RC}}};

			\draw[->] (1) edge[below] node {\footnotesize{$R_X$}} (2);
			\draw[->] (2) edge[right] node{\footnotesize{$\leq$}} (3);
			\draw[->, dotted ] (1) edge[left] node{\footnotesize{$\leq$}}  (4);
			\draw[->,dotted,] (4) edge[above] node {\footnotesize{$R_X$}} (3);
			
		\end{tikzpicture}
		\quad 
			\begin{tikzpicture}[scale=0.8]
			\tikzstyle{node}=[circle,fill=black,inner sep=1.2pt]
			
			\node[label= left :{\small{$w$}}] (1) at (0,0) [node] {};
			\node[label= right :{\small{$v$}}] (2) at (2,0) [node]
			{};
			\node[label=  right :{\small{$v'$}}] (3) at (2,2) [node]
			{};
			
			\node[label= left :{\small{$w'$}}] (4) at (0, 2) [node]
			{};
			\node[][] at (1,1){\textsf{\textcolor{blue}{CR}}};

			\draw[->] (1) edge[below] node {\footnotesize{$R_X$}} (2);
			\draw[->, dotted] (2) edge[right] node{\footnotesize{$\leq$}} (3);
			\draw[-> ] (1) edge[left] node{\footnotesize{$\leq$}}  (4);
			\draw[->,dotted,] (4) edge[above] node {\footnotesize{$R_X$}} (3);
			
		\end{tikzpicture}
		\vspace{-0.5cm}
	\end{center}
	\caption{Left commutativity, right commutativity and Church-Rosser
	}
	\label{fig:frames}
\end{figure}

\begin{definition}\label{def:basic models}
A \emph{basic birelational Chellas model} (\emph{basic model} for short) is a tuple 
$\M = \langle \W, \futs, \R, \V \rangle $ consisting of $\W$, a non-empty set of elements called \emph{worlds}, a binary relation $\futs $ over $\W$ which is reflexive and transitive, and a \emph{monotone} valuation function $\V: \W \longrightarrow 2^{\atm}$, that is, for $w, w'\in \W$, if $w \futs w'$ then $\V(w) \subseteq \V(w')$. Finally, $\R \subseteq \W \times \pow(\W) \times \W$ is a relation associating to every $X \subseteq \W$ a binary relation $\Rof{X}$ on $W$. 
\end{definition}

\begin{definition}
	For $\M$ basic model, $X \subseteq \W$, we consider the properties (Fig.~\ref{fig:frames}):
	\begin{center}
	\begin{tabular}{l @{\quad} l}
		\emph{Left Commutativity:} & $\forall w, w', v  (w \futs w' \Rof{X} v \Rightarrow \exists v' ( w\Rof{X} v'\futs v))$. \\
		\emph{Right Commutativity:} & $\forall w, w', v (w \Rof{X} v \futs v' \Rightarrow \exists w' (w \futs w' \ \& \ w'\Rof{X} v'))$. \\
		\emph{Church-Rosser:
		} & $\forall w, w', v (w \futs w' \ \& \ w \Rof{X} v \Rightarrow \exists v' (w'\Rof{X} v' \ \& \ v \futs v'))$. \\
	\end{tabular}
	\end{center}
\end{definition}

\begin{definition}\label{def:sat}
For $\M$ basic model, $w$ world of $\M$ and $\varphi$ formula, the satisfaction relation $\M, w \sat \varphi$ is defined inductively as follows. 
For $\truthset{\psi} = \{w \in \W \mid \M, w \sat \psi\}$, 
conditionals have a   \emph{local} $(L)$ or \emph{global} $(G)$ interpretation:
	\begin{itemize}
		\item $\M, w \sat p $ \emph{iff} $p \in \V(w)$;
		\item $\M, w \not \sat \bot$;
		\item $\M, w \sat \varphi \AND \psi$  \emph{iff} $\M, w \sat \varphi $ and $\M, w \sat \psi$;
		\item $\M, w \sat \varphi \OR \psi$  \emph{iff} $\M, w \sat \varphi $ or $\M, w \sat \psi$;
		\item $\M, w \sat \varphi \IMP \psi$  \emph{iff} for all $w' \in \W $ s.t. $w \futs w'$, if $\M, w' \sat \varphi $ then $\M, w' \sat \psi$;
		\item$(L\cb)$ $\M, w \sat \varphi \cb \psi$  \emph{iff} for all $v \in \W $ s.t. $w \Rof{\truthset{\varphi}}v$, $\M, v \sat \psi$;
		\item$(G\cb)$ $\M, w \sat \varphi \cb \psi$  \emph{iff} for all $w',v \in \W $ s.t. $w \futs w'\Rof{\truthset{\varphi}}v$, $\M, v \sat \psi$; 
		\item$(L\cd)$ $\M, w \sat \varphi \cd \psi$  \emph{iff}  there is $v \in W$ s.t. $w\Rof{\truthset{\varphi}}v$ and $\M, v \sat \psi$; 
		\item$(G\cd)$ $\M, w \sat \varphi \cd \psi$  \emph{iff} for all $w'\in\W $ s.t. $w \futs w'$ there is $v \in W$ s.t. $w'\Rof{\truthset{\varphi}}v$ and $\M, v \sat \psi$.
\end{itemize}
\end{definition}

Models for $\CCKbox$ and $\IntCK$ are respectively defined in \cite{weiss:2019,olk:2024} as follows.

\begin{definition}
A \emph{Weiss model} 
is a basic model satisfying  \emph{left commutativity}. 
For $\varphi \in \lanB$, the satisfaction relation $\M, w \sat \varphi$ is inductively defined by the propositional clauses in Def.~\ref{def:sat}
and $(L\cb)$. 
\end{definition}

\begin{definition}
An \emph{Olkhovikov model}, or \emph{\intu\ Chellas model}, is a basic model satisfying \emph{right commutativity} and \emph{Church-Rosser}.
For  $\varphi \in \lanBD$, the satisfaction relation $\M, w \sat \varphi$ is inductively defined by the propositional clauses in Def.~\ref{def:sat}, $(G\cb)$ and $(L\cd)$. 
\end{definition}

For $\M$ any kind of model above, a formula $\varphi$ is \emph{valid in} $\M$ if it is satisfied by all its worlds.
For  $\varphi\in\lanB$, $\varphi$ is derivable in $\CCKbox$ iff it is valid in all Weiss models \cite{weiss:2019},
and for $\psi\in\lan$, $\psi$ 
is derivable in $\IntCK$ iff  it is valid in all \olk\ models \cite{olk:2024}.

%

\section{Nested sequent calculus for $\IntCK$}
\label{sec:nested}

In this section we present a nested sequent calculus $\NIntCK$ for logic $\IntCK$. The proof system `combines' the nested proof system for classical conditional  logic $\CK$ from Alenda et al. in~\cite{alenda:2012,alenda:2016} with  the nested sequent calculus for intuitionitstic modal logic $\IK$ introduced by Stra\ss burger in~\cite{strassburger:2013}. In particular, the treatment of the intuitionistic base is from~\cite{strassburger:2013}, while the idea of indexing nested components with formulas to evaluate conditional operators is from~\cite{alenda:2012,alenda:2016}. 

For convenience, we employ \emph{polarised} formulas, as in~\cite{strassburger:2013}. We associate to formulas $\varphi$ of $\lan$ a \emph{input polarity}, $\bl{-}$, or a \emph{output polarity}, $\wh{-}$. 
Polarities indicate the position of the formulas within two-sided sequents: input formulas $\bl{\varphi}$ can be thought as occurring in the antecedent of a sequent, while output formulas $\wh{\varphi}$ as formulas occurring in the consequent. 
Our nested sequents are `single-succedent', meaning that we allow \emph{exactly} output formula in a nested sequent. 
As it is usual for nested sequents, we enrich the structure of Gentzen-style sequent with an additional structural connective $[\,]$, which is now indexed with a non-polarised formula acting as index, or placeholder. A colon separates the index formulas from the rest of the nested sequent. The formal definition is given below. 

\begin{definition}	
	\emph{Intuitionistic conditional nested sequents} $\Gamma$ (\emph{nested sequents} for short)  
	are generated as follows, for $\varphi, \psi$ (polarized) formulas of the language: 
	$$
	\Lambda ::= \emptyset \mid \Lambda, \bl\varphi \mid \Lambda ,\NS{\psi}{\Lambda}
	\qquad 
	\Gamma ::= \Lambda, \wh\varphi \mid \Lambda, \NS{\psi}{\Gamma}
	$$
	Objects $\Lambda$ containing only input/$\bullet$-formulas are \emph{input sequents}, while objects $\Gamma$ containing exactly one output/$\circ$-formula are \emph{nested sequents}. 
	The formula interpretation of both objects  is defined as:
	\begin{equation*}
		\begin{array}{r c l @{\hspace{0.4cm}} r c l}
			\fint(\emptyset) & := &  \top  & \fint(\Lambda, \wh\varphi) & := & \fint(\Lambda) \IMP \varphi \\
			\fint(\Lambda, \bl \varphi) & := &  \fint(\Lambda) \AND \varphi  & 		\fint(\Lambda, \NS{\psi}{\Gamma}) & := & \fint(\Lambda) \IMP  \big( \psi \cb \fint(\Gamma) \big) \\
			\fint(\Lambda, \NS{\psi}{\Lambda'}) & := &  \fint(\Lambda) \AND  \big( \psi \cd \fint(\Lambda ') \big) & &  \\
		\end{array}
	\end{equation*}
\end{definition}
Observe that, according to our definition, nested sequents cannot be empty. 
A \emph{context} is a nested sequent  containing a \emph{hole}, $\{ \,\}$, to be filled with a nested sequent.  This notion is needed to apply rules `deep' within nested sequents. 

\begin{definition}
	We inductively define the notion of \emph{context}, denoted $\GC{\,}$, and \emph{depth} of a context, denoted $\depth{\GC{\,}}$, for $\Delta$ nested sequent and $\varphi$ formula:
	\begin{itemize}
		\item $\{\,\}$ is a context with $\depth{\{\,\}} = 0$;
		\item $\Delta, \GC{\,}$ is a context with $\depth{\Delta, \GC{\,}} = \depth{\GC{\,}}$;
		\item$\NS{\varphi}{ \GC{\,}}$ is a context with $\depth{\NS{\varphi}{ \GC{\,}} }= \depth{\GC{\,}}+1$.
	\end{itemize}
\end{definition}

A context containing only ${\bullet}$-formulas is  a  \emph{\bcont}, and a context containing exactly one  $\circ$-formula is a  \emph{\wcont}. 
A  {\bcont} filled with an input sequent is an input sequent, while a {\bcont} filled with a nested sequent is a nested sequent. A \wcont  can only be filled with an input sequent.



\begin{example}
	$\Lambda \{ \,\} = \bl{\varphi}, \bl{\psi}, \NS{\eta}{\bl{\zeta}, \{\, \}}$ is a \bcont, $\bl{\chi}$ is an input sequent, and  $\Delta= \bl{\chi}, \NS{\gamma}{\wh{\xi}}$ is a nested sequent. $\Lambda \{ \bl{\chi} \} = \bl{\varphi}, \bl{\psi}, \NS{\eta}{\bl{\zeta}, \bl{\chi}}$ is an input sequent, and $\Lambda \{\Delta \} = \bl{\varphi}, \bl{\psi}, \NS{\eta}{\bl{\zeta}, \bl{\chi}, \NS{\gamma}{\wh{\xi}} }$ is a nested sequent.  $\Gamma \{ \,\} = \bl{\varphi}, \bl{\psi}, \NS{\eta}{\wh{\zeta}, \{\, \}}$ is a \wcont; $\Gamma \{ \bl{\chi}\} $ is a nested sequent, while $\Gamma\{ \Delta \} $ is \emph{not} a nested sequent, as it contains two output formulas.  
\end{example}

\begin{figure}[t!]
\begin{center}
	\begin{tabular}{c}
		$\vlinf{\initv}{}{\GC{\bl{p}, \wh{p}}}{}$ 
		\qquad 
		$\vlinf{\brl{\bot}}{}{\GC{\bl{\bot}} }{}$ 
		\qquad 
		$\vlinf{\brl{\AND}}{}{ \GC{\bl{ \varphi \AND \psi }} }{ \GC{\bl{\varphi}, \bl{\psi}} }$
		\qquad
		$\vliinf{\wrl{\AND}}{}{ \GC{\wh{\varphi \AND \psi}} }{ \GC{\wh{\varphi}} }{ \GC{\wh{\psi}} }$\\[0.5cm]
		$\vliinf{\brl{\OR}}{}{ \GC{\bl{\varphi \OR \psi}} }{ \GC{\bl{\varphi}} }{ \GC{\bl{\psi}} }$
		\quad 
		$\vlinf{\wrl{\OR}}{}{ \GC{\wh{ \varphi \OR \psi }} }{ \GC{\wh{\varphi}}}$
		\quad 
		$\vlinf{\wrl{\OR}}{}{ \GC{\wh{ \varphi \OR \psi }} }{ \GC{\wh{\psi}}}$
		\quad 
		$\vliinf{\brl{\IMP}}{}{ \GC{\bl{\varphi \IMP \psi}} }{ \GCD{\bl{\varphi \IMP \psi}, \wh{\varphi}} }{\GC{\bl{\psi}} }$\\[0.5cm]
		$\vlinf{\wrl{\IMP}}{}{\GC{\wh{\varphi \IMP \psi}}}{\GC{\bl{\varphi}, \wh{\psi}}}$
		\qquad 
		$\vliiinf{\bl{\cb}}{}{ \GC{ \bl{\varphi \cb \psi}, \NS{\eta}{\Delta} } }{\bl{\varphi}, \wh{\eta}}{\bl{\eta}, \wh{\varphi}}{ \GC{ \bl{\varphi \cb \psi}, \NS{\eta}{\bl{\psi}, \Delta} } }$
		\qquad 
		$\vlinf{\wh{\cb}}{}{ \GC{\wh{\varphi \cb \psi}} }{ \GC{\NS{\varphi}{\wh{\psi}}} }$\\[0.5cm]
		$\vlinf{\bl{\cd}}{}{ \GC{\bl{\varphi \cd \psi}} }{ \GC{\NS{\varphi}{\bl{\psi}}} }$
		\qquad 
		$\vliiinf{\wh{\cd}}{}{ \GC{ \wh{\varphi \cd \psi}, \NS{\eta}{\Delta} } }{\bl{\varphi}, \wh{\eta}}{\bl{\eta}, \wh{\varphi}}{ \GC{  \NS{\eta}{\wh{\psi}, \Delta} } }$
	\end{tabular}
\end{center}	
%
%
%
%
%
\caption{The rules of $\NIntCK$}
\label{fig:nested:IntCK}
\end{figure}

The rules of $\NIntCK$ are defined in \Cref{fig:nested:IntCK}. 
Propositional rules are standard (refer, e.g., to~\cite{troelstra:2000}). For $\GC{\,}$ context, $\GCD{\,}$ denotes the result of removing the $\circ$-formula from $\GC{\,}$. 
The conditional rules, inspired from~\cite{alenda:2016}, `exemplify' the semantic conditions for $\cb$ and $\cd$. The rule $\wh{\cb}$, read bottom-up, creates a new nested component indexed with $\varphi$, representing the set of worlds accessible through $\Rof{\truthset{\varphi}}$. The rule $\bl{\cb}$ `moves' formula $\bl{\psi}$ to a nested component indexed with $\eta$, under the condition that $\eta$ and $\varphi$ are equivalent, i.e., that derivations of both $\bl{\varphi}, \wh{\eta}$ and $\bl{\eta} ,\wh{\varphi}$ can be constructed. 
A \emph{proof}, or \emph{derivation}, of a nested sequent $\Gamma$ in $\NIntCK$ is a tree whose root is labeled with $\Gamma$, whose leaves are occupied by instances of $\initv$ or $\bl{\bot}$, and such that the nested sequents occurring at intermediate nodes are obtained through application of $\NIntCK$-rules. 
If there is a proof of $\Gamma$ in $\NIntCK$, then $\Gamma$ is \emph{derivable} in $\NIntCK$. 

The proof of soundness of $\NIntCK$ is in Appendix~\ref{sec:app:soundness-nes}. The proof follows the structure of the soundness proof for intuitionistic nested sequents in~\cite{strassburger:2013}. 
\begin{restatable}{theorem}{soundnessnested}
	If a nested sequent $\nseq$ is derivable in $\NIntCK$, then $\derivesIntCk \fin{\nseq}$. 
\end{restatable}

Next, we shall give a \emph{syntactic} proof of completeness of $\NIntCK$,  inspired from~\cite{alenda:2012,alenda:2016}.  
We wish to prove that $\cut$ is eliminable (\Cref{fig:structural-rules});  however, to do so, we need to prove eliminaibility of the \emph{index-replacement} rule $\rep$. So,  after  proving admissibility the structural rules (\Cref{lemma:str-prop}), we show how to `simulate' instances of $\cut$  and of $\rep$, in~\Cref{lemma:adm:rp,lemma:cut-adm} .  In \Cref{thm:cut-el-nes}, we show how to transform derivations in $\NIntCKcut$  into  derivations in $\NIntCK$. 

For  $\varphi$ formula of $\lan$,  the \emph{weight} of $\varphi$ is defined as the number of binary connectives occurring in $\varphi$. 
The \emph{height} of a $\NIntCK$ derivation $\nder$, denoted by $\height{\nder}$, is the length of its longest branch, minus 1. 
In the $\cut$ rule as displayed in \Cref{fig:structural-rules}, we call $\xi$ the \emph{cut formula}. 
The \emph{rank} of an application of the $\cut$ rule is set to be the weight of the cut formula, plus 1. The \emph{rank} of a derivation $\nder$ in  $\NIntCKcut$, denoted by $\rk{\nder}$, is the maximum of all the ranks of rules $\cut$ occurring in $\nder$.
In the $\rep$ rule as displayed in \Cref{fig:structural-rules}, we call $\varphi$ the \emph{rep-formula}. 

Given a rule $\rg = \psi_1, .., \psi_n / \chi$ and a proof system $\ps$,  $\rg$ is \emph{admissible} in $\ps$ if, whenever there are $\ps$-proofs $\nder_1, .., \nder_n$ of $\psi_1, .., \psi_n$, then there is a $\ps$-proof $\nder$ of $\chi$.  
$\rg$ is \emph{height-preserving (hp-) admissible}  in $\ps$ if $\height{\nder} \leq \max(\height{\nder_1}, ..,\height{\nder_n})$, and it is \emph{rank-preserving (rk-) admissible} in $\ps$ if $\rk{\nder} \leq \max(\rk{\nder_1}, ..,\rk{\nder_n})$. 
Moreover, $\rg$ is \emph{invertible} in $\ps$ if, whenever there is a $\ps$-proof $\nder$ of $\chi$, then there are $\ps$-proofs $\nder_1, .., \nder_n$ of $\psi_1, .., \psi_n$. If, for every $i \leq n$,  $\height{\nder_i}\leq \height{\nder}$, the rule is \emph{hp-invertible}, and if $\rk{\nder_i}\leq \rk{\nder}$ the rule is \emph{rp-invertible}. 

We omit the proof of the following result, which easily follows	by induction on $\w{\varphi}$. Full proofs for the subsequent lemmas can be found in \Cref{sec:app:cut-el-nested}. 

\begin{proposition}
	\label{prop:gen:init}
	Every sequent of the form $\GC{\bl{\varphi}, \wh{\varphi}}$ is derivable in $\NIntCK$. 
\end{proposition}

\begin{figure}[t]
	\[
		\vlinf{\wk}{}{\GC{\Lambda}}{\GC{\emptyset}}
	\qquad 
	\vlinf{\necs}{}{\NS{\eta}{\Gamma}}{\Gamma}
	\qquad 
		\vliinf{\meds}{}{\GC{\NS{\eta}{\Delta_1, \Delta_2}}}{ \GC{\NS{\eta}{\Delta_1}} }{\GC{\NS{\eta}{\Delta_2}}}
		\qquad 
	\vlinf{\ctr}{}{\GC{\bl{\varphi}}}{\GC{\bl{\varphi}, \bl{\varphi}}}
	\]
	\[
	\vliinf{\cut}{}{\GC{\emptyset}}{\GCD{\wh{\xi}}}{\GC{\bl{\xi}}}
	\qquad 
	\vliiinf{\rep}{}{\GC{\NS{\eta}{\Delta}}}{\bl{\varphi}, \wh{\eta}}{\bl{\eta},\wh{\varphi}}{\GC{\NS{\varphi}{\Delta}}}
	\]
	\caption{
		The structural rules of \emph{weakening}, \emph{necessitation}, \emph{box-medial}, \emph{contraction} and \emph{cut}, and the \emph{index-replacement} rule. In the $\wk$ rule, $\Lambda$ needs to be an input sequent.}
	\label{fig:structural-rules}
\end{figure}

\begin{restatable}{lemma}{structuralnes} The following hold in $\NIntCKcutrep$:
	\label{lemma:str-prop}
	\begin{enumerate}[$a)$]
		\item 
		\label{it:adm:rules}
		The rules of  $\wk$, $\necs$, $\meds$ and $\ctr$ are hp- and rp-admissible; 
		\item 
		\label{it:inv} Rules $\bl{\AND}$, $\wh{\AND}$, $\bl{\OR}$, $\wh{\IMP}$, $\bl{\cb}$, $\wh{\cb}$ and $\bl{\cd}$ are hp- and rp-invertible.
	\end{enumerate}
\end{restatable}


\begin{restatable}{lemma}{replacement}
	\label{lemma:adm:rp}
	If $\nder^1$, $\nder^2$, $\nder^3$ are $\NIntCKcut$ derivations of 
	 $\bl{\varphi}, \wh{\eta}$ and $ \bl{\eta}, \wh{\varphi}$ and $\GC{\NS{\varphi}{\Delta}}$ respectively,  
	then there is a $\NIntCKcut$ derivation $\nder $ of $\GC{\NS{\eta}{\Delta}}$, such that 
	$\rk{\nder} = \max(\rk{\nder^1}, \rk{\nder^2}, \rk{\nder^3}, \w{\varphi}+1)$.
\end{restatable}

\begin{restatable}{lemma}{cutadmnes}
	\label{lemma:cut-adm}
	If $\nder^1$, $\nder^2$  are $\NIntCKcut$ derivations of $\GCD{\wh{\xi}}$ and $\GC{\bl{\xi}}$ such that $\rk{\nder^1} = \rk{\nder^2} = \w{\xi}$, then there is a $\NIntCKcutrep$ derivation $\nder$ of $\GC{\emptyset}$ s.t. $\rk{\nder} = \w{\xi}$.  Every rep-formula $\zeta$ in $\nder$ is s.t.  $\w{\zeta} < \w{\xi}$. 
\end{restatable}
\begin{proof}
	By induction on $m + n$, the sum of heights of $\nder^1$ and $\nder^2$. We here report one case, and refer to \Cref{sec:app:cut-el-nested} for more details. 
	Suppose that the cut formula is $\varphi \cb \psi$ is principal in the last rule applied in both $\nder^1$, $\nder^2$. 
	Then, for $\Delta$ and $\NS{\eta}{\Delta}$ input sequents, we have:
	\vspace{-0.5cm}
	$$
	\vlderivation{
		\vlin{\wh{\cb}}{}{ \GCD{ \wh{\varphi \cb \psi}, \NS{\eta}{\Delta} }  }{
			\vltr{\nder^1_1}{ \GCD{ \NS{\varphi}{\wh{\psi}} ,\NS{\eta}{\Delta} }	}{\vlhy{}}{\vlhy{\quad }}{\vlhy{}}
		}
	}	
	\qquad 
	\vlderivation{
		\vliiin{\bl{\cb}}{}{ \GC{ \bl{\varphi \cb \psi}, \NS{\eta}{\Delta} } }{
			\vltr{\nder^2_1}{	\bl{\varphi}, \wh{\eta} }{\vlhy{}}{\vlhy{\quad }}{\vlhy{}}
		}{
			\vltr{\nder^2_2}{	\bl{\eta}, \wh{\varphi} }{\vlhy{}}{\vlhy{\quad }}{\vlhy{}}
		}{
			\vltr{\nder^2_3}{  \GC{ \bl{\varphi \cb \psi}, \NS{\eta}{\bl{\psi}, \Delta} }  }{\vlhy{}}{\vlhy{\quad }}{\vlhy{}}
		}
	}
	$$
	We construct the following derivation $\nder$, where $\mathcal{Q} = 
		\GC{ \bl{\varphi \cb \psi}, \NS{\eta}{\bl{\psi}, \Delta} }  $:
			\vspace{-0.3cm}

	\begin{adjustbox}{max width = \textwidth}
		$
	\vlderivation{
		\vliq{\ctr,\meds}{}{ \GC{\NS{\eta}{\Delta}} }{
			\vliin{\cut}{}{ \GC{\NS{\eta}{\Delta}, \NS{\eta}{\Delta}}  }
					{
					\vliq{\wk}{}{ \GCD{  \NS{ \eta}{\wh{\psi}, \Delta }, \NS{\eta}{\Delta} } }{
						\vliiin{\rep}{}{ \GCD{  \NS{ \eta}{\wh{\psi} }, \NS{\eta}{\Delta} }   }{
							\vltr{\nder^2_1}{	\bl{\varphi}, \wh{\eta} }{\vlhy{}}{\vlhy{\quad }}{\vlhy{}}
						}{
							\vltr{\nder^2_2}{	\bl{\eta}, \wh{\varphi} }{\vlhy{}}{\vlhy{\quad }}{\vlhy{}}
						}{
							\vltr{\nder^1_1}{ \GCD{ \NS{\varphi}{\wh{\psi}} ,\NS{\eta}{\Delta} }	}{\vlhy{}}{\vlhy{\quad }}{\vlhy{}}
						}
					}
				}
				%
				{
				\vliq{\wk}{}{\GC{\NS{\eta}{ \bl{\psi}, \Delta }, \NS{\eta}{\Delta}}} {
					\vliid{ih}{}{\GC{\NS{\eta}{ \bl{\psi}, \Delta }} }{
						\vliq{\wk}{}{ \GCD{ \wh{\varphi \cb \psi}, \NS{\eta}{\bl{\psi}, \Delta} } }{
							\vlin{\wh{\cb}}{}{ \GCD{ \wh{\varphi \cb \psi}, \NS{\eta}{\Delta} }  }{
								\vltr{\nder^1_1}{ \GCD{ \NS{\varphi}{\wh\psi} ,\NS{\eta}{\Delta} }	}{\vlhy{}}{\vlhy{\quad }}{\vlhy{}}
							}
						}
					}
				{
					\vltr{\nder^2_3}{  
						\mathcal{Q}
					}{\vlhy{}}{\vlhy{\quad }}{\vlhy{}}
					}
				}
			}
		}
	}
	$
\end{adjustbox}\\

	\noindent The application of the inductive hypothesis is justified, since the sum of heights of the two subderivations above $ih$ is smaller than $m + n$.
	Next, we need to check that $\rk{\nder} \leq \w{\varphi \cb \psi} $. From \Cref{lemma:adm:rp}, we have that $\rk{\GCD{  \NS{ \eta}{\wh{\psi} }, \NS{\eta}{\Delta} } } = \max(\rk{\nder^2_1}, \rk{\nder^2_2}, \rk{\nder^1_1}, \w{\varphi} + 1)$.  By $ih$, $\rk{ \GC{\NS{\eta}{ \bl{\psi}, \Delta }}  } \leq \w{\varphi \cb \psi}$. 
	Then, $\rk{\nder} =   \max(\rk{\nder^2_1}, \rk{\nder^2_2}, \rk{\nder^1_1}, \w{\varphi} + 1,  \w{\varphi \cb \psi}, \w{\psi} + 1)$. By assumption, $\rk{\nder^i_j} \leq \w{\varphi \cb \psi}$, for $i \in \{1,2\}$ and $j\in\{1,2,3\}$. Moreover, $\w{\varphi} + 1 \leq \w{\varphi \cb \psi }$, and $\w{\psi} + 1 \leq \w{\varphi \cb \psi }$. Thus, $\rk{\nder} \leq \w{\varphi \cb \psi} $.
	Finally, the rep-formula $\varphi$ of the new application of $\rep$ in the derivation is such that $\w{\varphi}< \w{\varphi \cb \psi}$. 
\end{proof}

The following proof, reported in \Cref{sec:app:cut-el-nested}, consists in iterating  \Cref{lemma:cut-adm} and \Cref{lemma:adm:rp} in a certain order, until a $\cut$- and $\rep$-free derivation is obtained. 

\begin{restatable}[Cut-elimination]{theorem}{cutelnes}
	\label{thm:cut-el-nes}
	If $\nseq$ is derivable in $\NIntCKcut$, then $\nseq$ is derivable in $\NIntCK$. 
\end{restatable}
Thanks to the cut elimination result, we can now prove completeness of $\NIntCK$. 
For $A$ set of formulas, let $\bl{A}=\{ \bl{\gamma} \mid \gamma \in A \}$.

\begin{restatable}{theorem}{completenessnes}
	For  $A \cup \{\varphi\}$ finite set of formulas in $\lan$, if $\A \derivesIntCk \varphi$ then the nested sequent $\bl{A}, \wh{\varphi} $ is derivable in $\NIntCK$. 
\end{restatable}
\begin{proof}
	By showing that the axioms and rules of $\IntCK$ are derivable and admissible in $\NIntCKcut$, and then using \Cref{thm:cut-el-nes} to remove occurrences of $\cut$. Derivations of some axioms and rules can be found in \Cref{sec:app:cut-el-nested}. 
\end{proof}


\section{Basic constructive conditional logic}
\label{sec:basic}

We now move to the analysis of \const\ conditional logics.
We first present a Gentzen-style sequent calculus for $\CCKbox$
obtained as the single-succendent restriction of
(a two-sided formulation of) the sequent calculus for $\CKbox$ defined in \cite{pattinson:2011}.
We then consider sequent rules explicitly dealing with $\cd$,
and define on their basis logic $\CCK$, which is 
an extension of $\CCKbox$ in the full language $\lan$. Other than the sequent calculus, we shall introduce both an Hilbert-style axiom system and a class of models for $\CCK$.



Since we now deal with Gentzen-style sequents, we drop the polarities of formulas and use the standard two-sided notation. 
A \emph{sequent} is a pair $\G \seq \D$ of finite, possibly empty multisets of formulas. 
A sequent is called \emph{single-succedent} if $|\D| \leq 1$.
A sequent $\G\seq\D$ is interpreted 
via the \emph{formula interpretation} $\finseq$
as $\bigwedge\G \imp \bigvee\D$
if $\G \neq \emptyset$, and 
as $\bigvee\D$ if $\G = \emptyset$, where
$\bigvee\emptyset$ is interpreted as $\bot$.
We use $\G \sseq \D$ as an abbreviation for the two sequents $\G \seq \D$ and $\D \seq \G$.
The rules of calculus $\SCCKbox$ for $\CCKbox$ can be found on top of Fig.~\ref{fig:seq CCKbox}.

%
%


We shall prove soundness and completeness of $\SCCKbox$ as a consequence of soundness and completeness of an extension of $\SCCKbox$ with rules for the might operator $\cd$. 
To define such an extension, 
we start from a calculus $\SCK$ for classical $\CK$ with explicit rules for both $\cb$ and $\cd$.
Such a calculus can be obtained by extending the standard G3-style rules for $\CPL$ from~\cite{troelstra:2000}
with the following rules obtained by integrating $\cd$ into  the rules of \cite{pattinson:2011}
(cf. e.g. \cite{indr:2021book} for analogous rules for unary modal logics):

%
\begin{center}
\begin{small}
\ax{$\{\varphi \sseq \rho_i\}_{i \leq n}$} 
\ax{$\{\varphi \sseq \eta_j\}_{j \leq k}$}
\ax{$\sigma_1, ..., \sigma_n \seq \psi, \chi_1, ..., \chi_k$}
\rlab{($n,k \geq 0$)}
\tinf{$\G, \rho_1 \cb \sigma_1, ..., \rho_n \cb \sigma_n \seq \varphi \cb \psi, \eta_1 \cd \chi_1, ..., \eta_k \cd \chi_k, \D$}
\disp

\vspace{0.3cm}
\ax{$\{\varphi \sseq \rho_i\}_{i \leq n}$} 
\ax{$\{\varphi \sseq \eta_j\}_{j \leq k}$}
\ax{$\sigma_1, ..., \sigma_n, \psi \seq \chi_1, ..., \chi_k$}
\rlab{($n,k \geq 0$)}
\tinf{$\G, \rho_1 \cb \sigma_1, ..., \rho_n \cb \sigma_n, \varphi \cd \psi \seq \eta_1 \cd \chi_1, ..., \eta_k \cd \chi_k, \D$}
\disp
\end{small}
\end{center}

By restricting $\SCK$ to single-succedent sequents, we obtain the calculus $\SCCK$ defined by 
adding to the rules of $\SCCKbox$ the $\cd$ rules 
at the bottom of Fig.~\ref{fig:seq CCKbox}.
We now turn to investigating the structural properties of $\SCCKbox$ and $\SCCK$. 

\begin{figure}[t!]
	\ax{}
	\llab{\init}
	\uinf{ $\G, p \seq p$ }
	\disp
	\hfill
	\ax{}
	\llab{$\bot_{\mathsf{L}}$}
	\uinf{ $\G, \bot \seq \D$ }
	\disp
	\hfill
	\ax{$\G, \varphi, \psi \seq \D$}
	\llab{\lland}
	\uinf{$\G, \varphi \land \psi \seq \D$}
	\disp
	%
	\hfill
	\ax{$\G \seq \varphi$}
	\ax{$\G \seq \psi$}
	\llab{\rland}
	\binf{$\G \seq \varphi\land\psi$}
	\disp
	
	\vspace{0.2cm}
	\ax{$\G, \varphi \seq \D$}
	\ax{$\G, \psi \seq \D$}
	\llab{\llor}
	\binf{$\G, \varphi\lor\psi \seq \D$}
	\disp
	\hfill
	\ax{$\G \seq \varphi$}
	\llab{\rlorone}
	\uinf{$\G \seq \varphi\lor\psi$}
	\disp
	\hfill
	\ax{$\G \seq \psi$}
	\llab{\rlortwo}
	\uinf{$\G \seq \varphi\lor\psi$}
	\disp
	\hfill
	\ax{$\G, \varphi \seq \psi$}
	\llab{\rimp}
	\uinf{$\G \seq \varphi \imp \psi$}
	\disp
	
	\vspace{0.2cm}
	\ax{$\G, \varphi \imp \psi \seq \varphi$}
	\ax{$\G, \psi \seq \D$}
	\llab{\limp}
	\binf{$\G, \varphi \imp \psi \seq \D$}
	\disp
	\hfill
	\ax{$\{\varphi \sseq \rho_i\}_{i \leq n}$} 
	\ax{$\sigma_1, ..., \sigma_n \seq \psi$}
	\llab{\rulecb}
	\binf{$\G, \rho_1 \cb \sigma_1, ..., \rho_n \cb \sigma_n \seq \varphi \cb \psi$}
	\disp
	\vspace{0.2cm}
	
	\noindent\makebox[\linewidth]{\rule{\textwidth}{0.3pt}}
	\begin{center}
		\ax{$\{\varphi \sseq \rho_i\}_{i \leq n}$} 
		\ax{$\varphi \sseq \eta$}
		\ax{$\sigma_1, ..., \sigma_n, \psi \seq \vartheta$}
		\llab{\rulecd}
		\tinf{$\G, \rho_1 \cb \sigma_1, ..., \rho_n \cb \sigma_n, \varphi \cd \psi \seq \eta \cd \vartheta$}
		\disp
		\\[0.2cm]
		\ax{$\{\varphi \sseq \rho_i \}_{i \leq n}$} 
		\ax{$\sigma_1, ..., \sigma_n, \psi \seq$}
		\llab{\rulecbd}
		\binf{$\G, \rho_1 \cb \sigma_1, ..., \rho_n \cb \sigma_n, \varphi \cd \psi \seq \D$}
		\disp
	\end{center}
	\caption{\label{fig:seq CCKbox} 
		\textbf{Top}: Rules of $\SConstCKbox$, for $0 \leq |\D| \leq 1$ and $n \geq 0$. \textbf{Bottom}: Additional rules of $\SCCK$.}
\end{figure}

\begin{theorem}[Structural properties]\label{th:cut elim CCK}
The rules \lland, \rland, \llor, \rlorone, \rlortwo, \rimp\ are height-preserving invertible, 
and the rule \limp\ is height-preserving invertible with respect to the right premiss.
Moreover, the following weakening and contraction rules are height-preserving admissible in $\SCCK$,
and the following rule $\cut$ is admissible in $\SConstCK$.
\begin{center}
\ax{$\G \seq \D$}
\llab{$\lwk$}
\uinf{$\G, \varphi \seq \D$}
\disp
\
\ax{$\G \seq $}
\llab{$\rwk$}
\uinf{$\G \seq \varphi$}
\disp
\
\ax{$\G, \varphi, \varphi \seq \D$}
\llab{$\ctr$}
\uinf{$\G, \varphi \seq \D$}
\disp
\
\ax{$\G \seq \varphi$}
\ax{$\G', \varphi \seq \D$}
\llab{$\cut$}
\binf{$\G, \G' \seq \D$}
\disp
\end{center}
\end{theorem}
\begin{proof} 
The proof of hp-admissibility of $\lwk$, $\rwk$ and $\ctr$ and of the hp-invertibility of the propositional rules is
essentially the standard one for \textsf{G3ip} (cf. e.g.  \cite{troelstra:2000})
without major modification induced by addition of the conditional rules.
Admissibility of $\cut$ is proved by induction on lexicographically ordered pairs ($c$, $h$), 
where $c$ is the weight 
of the cut formula $\varphi$, and $h$ is the sum of the heights of the derivations of the two premisses of $\cut$.
As an example, consider the following:
\begin{center}
\begin{small}
\ax{$\varphi \seqder{1} \rho$ \ $\rho \seqder{2} \varphi$ \ $\varphi \seqder{3} \eta$ \ $\eta \seqder{4} \varphi$ \ $\sigma, \psi \seqder{5} \vartheta$}
\llab{\rulecd}
\uinf{$\G, \rho \cb \sigma, \varphi \cd \psi \seq \eta \cd \vartheta$}
\disp
\quad
\ax{$\eta \seqder{6} \xi$ \ $\xi \seqder{7} \eta$ \ $\chi, \vartheta \seqder{8}$}
\llab{\rulecbd}
\uinf{$\G', \xi \cb \chi, \eta \cd \vartheta \seq $}
\disp
\end{small}
\end{center}
with the conclusion $\G, \G', \rho \cb \sigma, \xi \cb \chi, \varphi \cd \psi \seq$ \  of $\cut$ derivable as follows,
where dotted lines represent applications of the induction hypothesis:
\vspace{-0.2cm}
\begin{center}
\begin{small}
\ax{$\varphi \seqder{1} \rho$ \ \ $\rho \seqder{2} \varphi$} 
\ax{$\varphi \seqder{3} \eta$ \ \ $\eta \seqder{6} \xi$}
\dottedLine
\uinf{$\varphi \seq \xi$}
\ax{$\xi \seqder{7} \eta$ \ \ $\eta \seqder{4} \varphi$}
\dottedLine
\uinf{$\xi \seq \varphi$}
\ax{$\sigma, \psi \seqder{5} \vartheta$ \ \ $\chi, \vartheta \seqder{8}$}
\dottedLine
\uinf{$\sigma, \chi, \psi \seq$}
\llab{\rulecbd}
\qinf{$\G, \G', \rho \cb \sigma, \xi \cb \chi, \varphi \cd \psi \seq$}
\disp
\end{small}
\end{center}\qed
\end{proof}

%
%


Next, we define an axiomatic system $\CCK$ equivalent to $\SCCK$, and prove  that $\CCK$ is a conservative extension of $\CCKbox$.

\begin{definition}[$\CCK$]
The logic $\CCK$ is defined in the language $\lan$ extending any axiomatisation of $\IPL$ (formulated in the language $\lan$), containing modus ponens,
with the axioms \axCMbox, \axCCbox, \axCNbox, \axCNdiam, \axCKdiam\ and the rules \RAbox, \RCbox, \RAdiam, \RCdiam\ from Fig.~\ref{fig:axioms}.
\end{definition}


\begin{theorem}[Syntactic equivalence]\label{th:synt compl CCK}
$\G \seq \D$ is derivable in $\SCCK$ if and only if $\finseq(\G \seq \D)$ is derivable in $\CCK$.
\end{theorem}

\begin{proof}
Both directions of this claim are proved by induction on the height of the derivations in the respective proof system.
 From left to right, we observe that for every rule 
$\varS_1, ..., \varS_k / \varS$ of $\SCCK$, the corresponding rule $\finseq(\varS_1), ..., \finseq(\varS_k) / \finseq(\varS)$ 
is derivable in $\CCK$.
For instance, 
for $n = 2, \G = \emptyset$, the rule corresponding to \rulecd\ has the following form:
$\varphi \biimp \rho_1, \varphi \biimp \rho_2, \varphi \biimp \eta, \sigma_1 \land \sigma_2 \land \psi \imp \vartheta/ (\rho_1 \cb \sigma_1) \land (\rho_2 \cb \sigma_2) \land (\varphi \cd \psi) \imp (\eta \cd \vartheta)$
To derive these rules we can make use of the rules \RMbox, \RMdiam\ and the axiom \axCW, which are easily derivable in $\CCK$
(see full derivations in Appendix \ref{app:der CCK}).
In the other direction, we can show that all axioms and rules of $\CCK$ are derivable, respectively admissible, 
in $\SCCK$, with modus ponens simulated as usual by $\cut$ (Appendix \ref{app:der SCCK}). \qed
\end{proof}


\begin{theorem}
If $\varphi \in \lanB$, then $\varphi$ is derivable in $\CCK$ if and only if it is derivable in $\CCKbox$.
\end{theorem}
\begin{proof}
Suppose that $\varphi\in\lanB$ is derivable in $\CCK$.
Then, $\seq \varphi$ is derivable in $\SCCK$.
Since $\SCCK$ is analytic, the derivation $\varDer$ of $\seq \varphi$ does not contain any formula with occurrences of $\cd$, 
meaning that $\varDer$ does not employ the rules \rulecd\ and \rulecbd.
Since all other rules of $\SCCK$ also belong to $\SCCKbox$, $\varDer$ is a derivation of $\seq \varphi$ in $\SCCKbox$, thus $\varphi$ is derivable in $\CCKbox$.
\end{proof}

The following is an immediate consequence of the previous results:

\begin{theorem}[Syntactic equivalence]
$\G \seq \D$ is derivable in $\SCCKbox$ if and only if $\finseq(\G \seq \D)$ is derivable in $\CCKbox$.
\end{theorem}

As already mentioned, $\ConstCKbox$ and $\IntCK$ naturally correspond to known \intu\ unary modal logics.
 Similarly, the modal correspondent of $\CCK$ is the propositional fragment of \wij's logic $\mathsf{CCDL}$,
often denoted $\WK$ \cite{das:2023}. This fragment is 
defined in a language 
$\lanm$,  containing $\Box, \diam$, by extending $\IPL$ with the rule $\varphi / \Box\varphi$ and the 
axioms 
$\Box(\varphi \imp \psi) \imp (\Box\varphi \imp \Box\psi)$,
$\Box(\varphi \imp \psi) \imp (\diam\varphi \imp \diam\psi)$,
$\neg\diam\bot$.
Let $(\cdot)^\tau$ be a translation $\lanm \longrightarrow \lan$ 
such that $p^\tau = p$, $\bot^\tau = \bot$, $(\rho \circ \psi)^\tau = \rho^\tau \circ \psi^\tau$ for $\circ \in \{\land, \lor, \imp\}$,
and $(\circ \psi)^\tau = \top \cc \psi^\tau$ for $\circ \in \{\Box, \diam\}$.

\begin{theorem}
For all $\varphi\in\lanm$, $\vd_{\WK} \varphi$ if and only if \ $\vd_{\CCK} \varphi^\tau$.
\end{theorem}
\begin{proof}
Considering $\SWK$, the propositional fragment of the sequent calculus for $\WK$ in \cite{wij:1990},
and transforming every derivation in $\SCCK$ into a derivation in $\SWK$ by replacing every application of \rulecb, \rulecd, \rulecbd\ with an application of the corresponding modal rule, and vice versa.
\end{proof}

\subsection{Semantics}
We now present a semantics for $\CCK$ defined in terms of basic Chellas models.

\begin{definition}[Constructive Chellas model]
We call \emph{\const\ Chellas model} (\ccm\ for short) any basic model (Def.~\ref{def:basic models})
where the satisfaction relation $\M, w \sat \varphi$ is inductively defined by the propositional clauses in Def.~\ref{def:sat}
and the \emph{global} interpretations $(G\cb)$ and $(G\cd)$ for $\cb$ and $\cd$.
\end{definition}



\ccm s and Weiss models can be shown equivalent for formulas of $\lanB$.
By induction 
on the construction of the formulas,
one can prove that \ccm s satisfy the hereditary property,
that is, 
for all \ccm s $\M$, $w, w'$ of $\M$, and $\varphi\in\lan$,
if  $w \Vd \varphi$ and $w \less w'$, then $w' \Vd \varphi$.
Moreover, one can easily show that  $\CCK$ is sound with respect to \ccm s,
by showing that all axioms and rules of $\CCK$ are valid, respectively validity-preserving, in every \ccm\
(c.f. Appendix \ref{app:sem}).


We now prove that $\CCK$ is complete with respect to \ccm s 
by a canonical model construction inspired from~\cite{wij:1990}.
In the following, $\vd$ only refers to $\vd_{\CCK}$.
A set $\A$ of formulas of $\lan$ is \emph{prime}
if it is consistent ($\A\not\vd\bot$),
deductively closed (if $\A\vd\varphi$, then $\varphi\in\A$),
and satisfies the disjunction property
(if $\varphi\lor\psi\in\A$, then $\varphi\in\A$ or $\psi\in\A$).
We denote by $\pset{\varphi}$ the set of all prime sets $\A$ such that $\varphi\in\A$. 
The following proof, reported in Appendix \ref{app:sem}, is a routine extension of the proof of \cite{segerbergMin} for $\IPL$ to the conditional language $\lan$.
%

%

\begin{restatable}{lemma}{lindenbaum}\label{lemma:lind}
If $\A\not\vd\varphi$, then there is a prime set $\B$ such that $\A\subseteq\B$ and $\varphi\notin\B$.
\end{restatable}

\begin{definition}[Canonical element]\label{def:can el}
A \emph{canonical element} (\ce) is a pair $(\A, \ceR)$, where 
$\A$ is a prime set, and $\ceR$ is a set of pairs $(\alpha, \beta)$ such that
\begin{itemize}
\item $\alpha = \pset{\varphi}$ for some $\varphi\in\lan$;
\item $\beta$ is a set of prime sets such that, for all $\psi\in \lan$:
\begin{itemize}
\item if $\varphi\cb\psi \in \A$, then for all $\B\in\beta$, $\psi\in\B$,
\item if $\varphi\cd\psi \in \A$, then there is $\B\in\beta$ such that $\psi\in\B$;
\end{itemize}
\item for all $(\alpha, \beta), (\alpha', \beta') \in \ceR$, $\alpha = \alpha'$ implies $\beta = \beta'$;
\item if $\varphi\cd\psi \in \A$, then there is $\beta$ such that $(\pset{\varphi}, \beta) \in \ceR$.
\end{itemize}
\end{definition}

We denote by $\ceset{\varphi}$ the set of all \ce s $(\A, \ceR)$ such that $\varphi\in\A$. 
For every set $\A \subseteq \lan$, we denote by $\fcbm{\A}$ the set $\{\psi \mid \varphi\cb\psi\in\A\}$.
The proof of the following lemma is in Appendix \ref{app:sem}. 

\begin{restatable}{lemma}{lemmacanel}\label{lemma:can el}
For every prime set $\A$:
\begin{enumerate}
\item There is a \ce\ $(\A, \ceR)$;
\item If $\varphi \cb \psi \notin \A$, then there is a \ce\ $(\A, \ceR)$ and a pair $(\pset{\varphi}, \beta) \in \ceR$ such that $\psi\notin\B$ for some $\B \in \beta$;
\item If $\varphi \cd \psi \notin \A$, then there is a \ce\ $(\A, \ceR)$ and a pair $(\pset{\varphi}, \beta) \in \ceR$ such that $\psi\notin\B$ for all $\B \in \beta$.
\end{enumerate}
\end{restatable}

\begin{definition}[Canonical model]\label{def:mod can}
The \emph{canonical model} for $\CCK$ is the tuple
$\Mc = \langle \Wc, \lessc, \Rc, \Vc \rangle$, where:
\begin{itemize}
\item $\Wc$ is the set of all canonical elements;
\item $(\A, \ceR) \lessc (\A', \ceR')$ iff $\A \subseteq \A'$;
\item $(\A, \ceR) \Rcof{X} (\B, \ceU)$ iff there are $\varphi$, $\beta$ s.t. $X = \ceset{\varphi}$, $(\pset{\varphi}, \beta)\in\ceR$ and $\B\in\beta$;
\item $p \in \Vc((\A, \ceR))$ iff $p \in \A$.
\end{itemize}
\end{definition}

It is easy to verify that the canonical model for $\CCK$ is a \ccm.

\begin{restatable}{lemma}{lemmamodcan}\label{lemma:mod can}
For all $\varphi\in\lan$ and all prime sets $\A$, $\varphi\in\A$ iff $(\A, \ceR) \Vd \varphi$ for all \ce s $(\A, \ceR)$.
\end{restatable}
\begin{proof}
Note that the claim is equivalent to $\truthset{\varphi} = \ceset{\varphi}$.
The proof is by induction on the construction of $\varphi$.
We only show the case $\varphi = \rho \cb \psi$ ($\varphi = \rho \cd \psi$ in Appendix \ref{app:sem}).
%
Assume $\rho \cb \psi \in \A$ and $(\A, \ceR) \lessc (\A', \ceR') \Rcof{\trset{\rho}} (\B, \ceU)$.
Then $\A \subseteq \A'$, which implies $\rho \cb \psi \in \A'$, and by \ih, $(\A', \ceR') \Rcof{\ceset{\rho}} (\B, \ceU)$.
By definition of $\Rcof{X}$, there is $\beta$ such that $(\pset{\rho}, \beta) \in \ceR'$ and $\B \in \beta$,
then by definition of \ce, $\psi\in\B$.
Thus by \ih, $(\B, \ceU)\Vd\psi$,
hence $(\A, \ceR) \Vd \rho \cb \psi$.
Now assume $\rho \cb \psi \notin \A$.
By Lemma~\ref{lemma:can el},
there are 
$(\A, \ceR)$, $\beta$, $\B$
such that 
$(\pset{\rho}, \beta)\in\ceR$, $\B\in\beta$ and $\psi\notin\B$.
Moreover, there exists a \ce\ $(\B, \ceU)$.
Then $(\A, \ceR) \Rcof{\ceset{\rho}} (\B, \ceU)$, 
and by \ih,
$(\A, \ceR) \Rcof{\trset{\rho}} (\B, \ceU)$ and $(\B, \ceU)\not\Vd\psi$.
Therefore $(\A, \ceR) \not\Vd \rho\cb\psi$. \qed
\end{proof}

\begin{theorem}[Completeness]
If $\varphi$ is valid in all \ccm s if and only if it is derivable in $\CCK$.
\end{theorem}
\begin{proof}
If $\not\vd \varphi$, then by Lemma \ref{lemma:lind} there is $\A$ prime such that $\varphi\notin\A$,
then by Lemma \ref{lemma:can el} there is a \ce\ $(\A, \ceR)$ that by Def. \ref{def:mod can} belongs to the canonical model for $\CCK$, hence by Lemma \ref{lemma:mod can}, $(\A, \ceR)\not\Vd \varphi$,
thus $\varphi$ is not valid in all \ccm s.
\end{proof}

\section{Beyond  basic constructive conditional logic}
\label{sec:ext}
Similarly 
to what happens in classical $\CK$, 
axioms and frame conditions can be added to $\CCKbox$, $\CCK$ and 
$\IntCK$. 
Several extensions of $\CCKbox$ are studied in \cite{ciardelli:2020} while, to the best of our knowledge, extensions of $\IntCK$ or $\CCK$ have not been investigated. 
In this section, we define sequent calculi, axiom systems and classes of models for $\CCK$ extended with \emph{identity} and \emph{conditional modus ponens}. While we focus on $\CCK$ for this section, by removing the $\cd$-rules from our proof systems, we obtain sequent calculi for the corresponding extensions of $\CCKbox$, defined in~\cite{ciardelli:2020}.  We also provide a sequent calculus and an axiom system for $\CCK$ extended with the \emph{conditional excluded middle}, but a semantic characterisation of the logic is left to future work. 

To define proof systems for extensions of $\CCK$, we consider the sequent calculi of \cite{pattinson:2011} for the extensions of $\CKbox$ with the axioms
\axIDbox, \axMPbox\ and \axCEMbox\ (Fig. \ref{fig:ax ext}).
We integrate the rules of these calculi with $\cd$, and reduce them to single-succedent sequents,
obtaining the constructive conditional rules and calculi in Fig. \ref{fig:rules ext}.  
%
We first consider extensions of $\CCK$ with \axIDbox\ and \axMPbox. 
Take $\SCCKstar \in \{ \SCCKID, \SCCKMP, \SCCKMPID\}$. 

\begin{figure}[t!]
\begin{tabular}{ll}
\axIDbox\ & $\varphi \cb \varphi$ \\

\axMPbox\ & $(\varphi \cb \psi) \imp (\varphi \imp \psi)$ \\

\axMPdiam\ & $\varphi \land \psi \imp (\varphi \cd \psi)$ \\

\axCEMbox\ & $(\varphi \cb \psi) \lor (\varphi \cb \neg\psi)$ \\
 
\axCEMdiam\ & $(\varphi \cd \psi) \land (\varphi \cd \chi) \imp (\varphi \cd \psi \land \chi)$ \\
\end{tabular}
\hfill
\begin{tabular}{lllll}
\emph{System name} \ \ && \emph{Axiomatisation} \\
\hline
$\CCKID$ & \quad & $\CCK$ + \axIDbox \\
$\CCKMP$ & \quad & $\CCK$ + \axMPbox, \axMPdiam \\
$\CCKMPID$ & \quad & $\CCK$ + \axMPbox, \axMPdiam, \axIDbox \\
$\CCKCEM$ & \quad & $\CCK$ + \axCEMdiam \\
\end{tabular}
\caption{\label{fig:ax ext}Conditional axioms and extensions of $\CCK$.}
\end{figure}

%
%


\begin{figure}[t!]
\centering
\ax{$\{\varphi \sseq \rho_i\}_{i \leq n}$ \quad $\sigma_1, ..., \sigma_n, \varphi \seq \psi$}
\llab{\rulecbid}
\uinf{$\G, \rho_1 \cb \sigma_1, ..., \rho_n \cb \sigma_n \seq \varphi \cb \psi$}
\disp
\hfill
\ax{$\{\varphi \sseq \rho_i \}_{i \leq n}$ \quad $\sigma_1, ..., \sigma_n, , \varphi, \psi \seq$}
\llab{\rulecbdid}
\uinf{$\G, \rho_1 \cb \sigma_1, ..., \rho_n \cb \sigma_n, \varphi \cd \psi \seq \D$}
\disp

\vspace{0.2cm}
\ax{$\{\varphi \sseq \rho_i\}_{i \leq n}$} 
\ax{$\varphi \sseq \eta$}
\ax{$\sigma_1, ..., \sigma_n, , \varphi, \psi \seq \vartheta$}
\llab{\rulecdid}
\tinf{$\G, \rho_1 \cb \sigma_1, ..., \rho_n \cb \sigma_n, \varphi \cd \psi \seq \eta \cd \vartheta$}
\disp

\vspace{0.3cm}
\ax{$\G, \varphi \cb \psi \seq \varphi$} 
\ax{$\G, \varphi \cb \psi, \psi \seq \D$}
\llab{\rulempbox}
\binf{$\G, \varphi \cb \psi \seq \D$}
\disp
\quad
\ax{$\G \seq \varphi$} 
\ax{$\G \seq \psi$}
\llab{\rulempdiam}
\binf{$\G \seq \varphi \cd \psi$}
\disp

\vspace{0.3cm}
\ax{$\{\varphi \sseq \rho_i\}_{i \leq n}$} 
\ax{$\{\varphi \sseq \xi_j\}_{j \leq k}$} 
\ax{$\varphi \sseq \eta$}
\ax{$\sigma_1, ..., \sigma_n, \chi_1, ..., \chi_k, \psi \seq \vartheta$}
\llab{\rulecdcem}
\qinf{$\G, \rho_1 \cb \sigma_1, ..., \rho_n \cb \sigma_n, \xi_1 \cd \chi_1, ..., \xi_k \cd \chi_k, \varphi \cd \psi \seq \eta \cd \vartheta$}
\disp

\vspace{0.2cm}
\ax{$\{\varphi \sseq \rho_i \}_{i \leq n}$} 
\ax{$\{\varphi \sseq \xi_j\}_{j \leq k}$} 
\ax{$\sigma_1, ..., \sigma_n, \chi_1, ..., \chi_k, \psi \seq$}
\llab{\rulecbdcem}
\tinf{$\G, \rho_1 \cb \sigma_1, ..., \rho_n \cb \sigma_n, \xi_1 \cd \chi_1, ..., \xi_k \cd \chi_k, \varphi \cd \psi \seq \D$}
\disp


\vspace{0.4cm}
\begin{tabular}{llll|lllll}
\emph{Calculus} \ \ && \emph{Conditional rules} &\quad&& \emph{Calculus} \ \ && \emph{Conditional rules} \\
\hline
$\SCCKID$ & \quad & \rulecbid, \rulecdid, \rulecbdid &&&
$\SCCKMPID$ & \quad & \rulecbid, \rulecdid, \rulecbdid, \rulempbox, \rulempdiam \TT \\
$\SCCKMP$ & \quad & \rulecb, \rulecd, \rulecbd, \rulempbox, \rulempdiam &&&
$\SCCKCEM$ & \quad & \rulecb, \rulecdcem, \rulecbdcem \\
\end{tabular}

\caption{\label{fig:rules ext}Conditional rules and calculi for extensions, where  $0 \leq |\D| \leq 1$ and $n, k \geq 0$.}
\end{figure}


%

The proof of the following theorem,  reported in Appendix \ref{app:cut elim SCCKstar}, extends or slightly modifies the proofs of Th. \ref{th:cut elim CCK} with the new conditional rules. 

\begin{restatable}{theorem}{thm:struct:SCCKstar}
	\label{th:cut elim SCCKstar}
The rules $\lwk$, $\rwk$, $\ctr$ are height-preserving admissible, and the rule $\cut$ is admissible in 
$\SCCKstar$.
\end{restatable}

For every calculus $\SCCKstar$ we define a corresponding axiomatic system $\CCKstar$, see Fig. \ref{fig:ax ext}, and prove their equivalence with the sequent calculus characterisation (Appendix \ref{app:der SCCKstar} and \ref{app:der CCKstar}).

\begin{restatable}{theorem}{thm:equiv:SCCKstar}\label{th:synt compl CCKstar} 
$\G \seq \D$ is derivable in $\SCCKstar$ if and only if $\finseq(\G \seq \D)$ is derivable in 
$\CCKstar$.
\end{restatable}



By extending the completeness proof for $\CCK$, we can show that 
these logics are characterised by exactly those classes of \ccm s that satisfy the frame properties 
of Chellas models for their classical counterparts (c.f. \cite{chellas:1975}).

\begin{theorem}\label{th:sem compl CCKID}
Formula $\varphi$ is derivable in $\CCKID$, resp. $\CCKMP$, resp. $\CCKMPID$ if and only if it is valid in all \ccm s satisfiying the property (id), resp. the property (mp), resp. both properties (id) and (mp) below:
\begin{center}
\begin{tabular}{lllll}
(id) & if $w \Rof{X} v$, then $v \in X$; & \ \ \  &
(mp) & if $w \in X$, then $w \Rof{X} w$. \\
\end{tabular}
\end{center}
\end{theorem}
\begin{proof}
	Soundness is easy.
	For completeness, in the case of $\CCKID$ we need to show that the canonical model satisfies (id):
	Suppose $(\A, \ceR) \Rcof{X} (\B, \ceU)$.
	By definition,  $X = \ceset{\varphi}$ for some $\varphi\in\lan$,
	and there is $\beta$ such that $(\pset{\varphi}, \beta)\in\ceR$ and $\B\in\beta$.
	Since $\varphi \cb \varphi \in A$, by Def.~\ref{def:can el}, $\varphi\in B$ for all $B \in \beta$,
	that is $B \subseteq \pset{\varphi}$, 
	therefore $(\B, \ceU) \in \ceset{\varphi}$.
	For $\CCKMP$, we need to slightly modify the construction of \ce s in the proof of Lemma \ref{lemma:can el}.
	First, call \emph{MP}-\ce\ any \ce\ $(\A, \ceR)$ such that if $(\alpha, \beta)\in\ceR$, $\alpha = \pset{\varphi}$ and $\varphi\in\A$, then $\A\in\beta$.
	Moreover, in item (i) of the proof, let $\beta$ be defined as before as a set containing a witness for every $\varphi \cd \vartheta \in \A$.
	We now define $\beta^* = \beta$ if $\varphi\notin\A$, and $\beta^* = \beta \cup \{A\}$ if $\varphi\in\A$,
	and define $\ceR$ as a set of pairs $(\pset{\varphi}, \beta^*)$.
	Then $(\pset{\varphi}, \beta^*)$ is still a \ce, in particular if $\varphi\cb\psi, \varphi \in \A$, then by \axMPbox, $\psi\in\A$.
	We define $\Wc$ in the canonical model for $\CCKMP$ as the set of all MP-ces. It is easy to verify that 
	the model 
	satisfies (mp).
	For $\CCKMPID$, we consider the same canonical model construction and show that it satisfies (id).\qed
\end{proof}

%
%

Let us consider \axCEMbox.
This case is particularly interesting, as there is no agreement in the literature about how
a constructive variant of it should, or could, be defined. 
According to Weiss \cite{weiss:2019}, 
intuitionistic analogues of these logics are not possible, 
 since $\varphi\lor\neg\varphi$ becomes derivable
in presence of \axMPbox\ and \axCEMbox.
Against this claim, Ciardelli and Liu \cite{ciardelli:2020} observe that by 
replacing \axCEMbox\ with $(\varphi \cb \psi \lor \chi) \imp (\varphi \cb \psi) \lor (\varphi \cb \chi)$,
which is classically but not intuitionistically equivalent to \axCEMbox,
one can define intuitionistic conditional systems that do not collapse into classical logic.
Our proof-theoretical approach suggests at least one alternative
constructive variant of the logics with \axCEMbox.
The classical sequent rule of \cite{pattinson:2011}
for $\CKbox$ + \axCEMbox\ has the following shape: 
\begin{center}
\begin{small}
\ax{$\{\varphi_1 \sseq \rho_i\}_{i \leq n}$} 
\ax{$\sigma_1, ..., \sigma_n \seq \psi_1, ..., \psi_m$}
\rlab{($n \geq 0, m \geq 1$)}
\binf{$\G, \rho_1 \cb \sigma_1, ..., \rho_n \cb \sigma_n \seq \varphi_1 \cb \psi_1, ..., \varphi_m \cb \psi_m, \D$}
\disp
\end{small}
\end{center}
The single-succedent restriction of this rule simply collapses into \rulecb. 
However, by integrating $\cd$ into the classical calculus we add $k \geq 1$
principal $\cd$-formulas in the premiss of conditional rules, 
which are preserved by their single-succedent restriction.
The resulting rules are displayed in Fig. \ref{fig:rules ext}. Using these rules, one can derive  axiom \axCEMdiam\ (Fig. \ref{fig:ax ext}),
which is the characterising axiom of $\CCKCEM$
(Cf. derivation of \axCEMdiam\ in Appendix \ref{app:der CEM}). The proofs of the following statements are in \Cref{app:sound:cem,app:cutel:cem}. 

\begin{theorem}\label{th:synt compl CCKCEM}
Weakening and contraction are height-preserving admissible, and $\cut$ is admissible in $\SCCKCEM$.
If   
$\G \seq \D$ is derivable in $\SCCKCEM$ if and only if $\finseq(\G \seq \D)$ is derivable in $\CCKCEM$.
\end{theorem}


\section{Conclusions} 
\label{sec:con}
We have investigated the proof theory of intuitionistic conditional logics, by defining a nested sequent calculus for $\IntCK$ and a sequent calculus for $\CCKbox$. The calculi are obtained by imposing a  single-succedent restriction to a nested  and a sequent calculus for classical $\CK$ respectively. Then, inspired by the sequent calculus rules, we have defined logic $\CCK$, a natural extension of $\CCKbox$ with the $\cd$ modality. We have introduced a class of models for the logic, and showed its soundness and completeness with respect to its axiomatization. Finally, we have defined extensions of $\CCK$ with \emph{identity}, \emph{conditional modus ponens} and \emph{conditional excluded middle}, and investigated their proof theory. 

In future work, we would like to employ our proof systems to determine decision procedures for intuitionistic conditional logics, thus giving the first complexity bounds for these logics. Moreover, we are interested in exploring extensions of $\CCK$ and $\IntCK$ with further axioms and frame conditions coming from classical conditional logics. As a first candidate, we plan to cover $\CCK$ with \emph{cautious monotonicity}, taking inspiration from the classical rules proposed in~\cite{schr:2010}. 
Regarding  $\IntCK$, we conjecture that the nested sequent calculi defined in ~\cite{alenda:2012,alenda:2016} for extensions of classical $\CK$ can be adapted to our nested sequent  framework. The rule for  \emph{conditional excluded middle} poses a problem though, as it does not seem to be derivable in the single-succedent setting of $\NIntCK$. Switching to a multi-succedent nested calculus as the one in~\cite{kuznets:2019} might be necessary to cover this case. 
Finally, we are interested in investigating the semantics and the proof theory of stronger conditional logics, such as preferential logics and Lewis' counterfactuals, in a constructive setting. 

 \bibliographystyle{splncs04}
 \bibliography{mybibliography}
%
%
%
%
%

\appendix


\section{Soundness and completeness of $\NIntCK$}

\begin{remark}
	\label{remark:contexts}
	Any \bcont $\GC{\,}$ has the following shape, where every $\Lambda_i$ is an input sequent: 
	$
	\Lambda_1, \NS{\eta_1}{\Lambda_1, \NS{\eta_2}{\Lambda_2 ,  \dots ,\NS{\eta_n}{\Lambda_n, \{ \,\} } \ldots}} 
	$. 
	If the hole in $\GC{\,}$ is filled with a nested sequent $\Sigma$, then $\GC{\Sigma}$ is a nested sequent.  Otherwise, if $\Sigma$ is an input sequent, then $\GC{\Sigma}$ is also an input sequent. 
	
	Any \wcont  $\GC{\,}$ is instead of the form $\GCp{ \LC{\,}, \Pi}$, where $\Pi$ is a nested sequent (thus containing one output formula), and $\GCp{\,}$, $\LC{\,}$ are \bcont. Since $\GC{\,}$  already contains an output formula, the hole in  $\GC{\,}$  can only be filled with an input sequent, yielding a nested sequent. Observe that $\GCp{\,}$, $\LC{\,}$  and $\Pi$ are uniquely determined by the position of $\{ \,\}$ in $\GC{\,}$. 
\end{remark}

\subsection{Soundness of $\NIntCK$}
\label{sec:app:soundness-nes}
We now turn to proving soundness of $\NIntCK$. We shall prove that every rule $\rr$ of the calculus preserves derivability in $\IntCK$, that is: if the formula interpretation of every premiss of $\rr$ is derivable in $\IntCK$, then the conclusion of $\rr$ is derivable in $\IntCK$. 
The proof structure is inspired by~\cite{strassburger:2013}. 
We start by recalling the derivability results that will be needed in the proof. 

\begin{lemma}
	\label{lemma:validities:prop}
	The following hold:
	\begin{enumerate}[$a)$]
		
		\item 
		\label{it:imp:mono:lemma:validities}
		If $\derivesIntCk \varphi \IMP \vartheta$, then $\derivesIntCk(\xi \IMP \varphi ) \IMP (\xi \IMP \vartheta)$; 
		
		\item If $\derivesIntCk (\varphi  \AND \psi) \IMP \vartheta$, then $\derivesIntCk \big((\xi \IMP \varphi ) \AND (\xi \IMP \psi) \big) \IMP (\xi\IMP \vartheta)$; 
		\label{it:imp:multi:lemma:validities}

		\item 
		\label{it:impand:mono:lemma:validities}
		If $\derivesIntCk\vartheta \IMP \varphi $, then $\derivesIntCk (\xi \AND \vartheta) \IMP (\xi \AND \varphi)$;
		
		\item 
		\label{it:impor:multi:validities}
		If $\derivesIntCk\vartheta \IMP ( \varphi \OR \psi)$, then $\derivesIntCk(\xi \AND \vartheta) \IMP \big((\xi \AND \varphi) \OR (\xi \AND \psi)\big)$;

		\item 
		\label{it:imporand:multi:validities}
		If $\derivesIntCk\vartheta \IMP ( \varphi \OR \psi)$, then $\derivesIntCk \big((\varphi \IMP \xi) \AND  (\psi \IMP \xi) \big) \IMP (\vartheta \IMP \xi)$;

		\item 
		\label{it:impand:lemma:validities}
		If $\derivesIntCk\varphi \IMP (\psi \IMP \vartheta)$, then $\derivesIntCk(\varphi  \AND \psi) \IMP   \vartheta$;
		
		\item 
		\label{it:impimp:lemma:validities}
		If $\derivesIntCk \varphi \IMP\psi$, then $\derivesIntCk (\psi \IMP \xi) \IMP (\varphi \IMP \xi)$;
		
		\item 
		\label{it:impred:lemma:validites}
		If $\derivesIntCk  \varphi \IMP (\psi \IMP \theta) $, then $\derivesIntCk  (\psi \IMP \varphi) \IMP (\psi \IMP \theta)  $;
		
		\item  
		\label{it:impanddis:validities}
		If $\derivesIntCk (\varphi \AND \psi) \IMP \theta$, then $\derivesIntCk (\varphi \AND \psi) \IMP (\varphi \AND \theta)$.
		
	\end{enumerate}
\end{lemma}
\begin{proof}
	The proof of \ref{it:imp:mono:lemma:validities} - \ref{it:imporand:multi:validities} can be found in~\cite{strassburger:2013}; cases \ref{it:impand:lemma:validities} - \ref{it:impanddis:validities} can be easily verified in intuitionistic propositional logic. 
\end{proof}

\begin{lemma}
	\label{lemma:validities:cond}
	The following hold:
	\begin{enumerate}[$i)$]
		\item 
		\label{it:cbmono:lemma:validities}
		If $\derivesIntCk \varphi  \IMP \vartheta$, then $\derivesIntCk (\xi \cb \varphi ) \IMP (\xi\cb \vartheta)$; 	
		
		\item 
		\label{it:cbmulti:lemma:validities}
		If $\derivesIntCk (\varphi  \AND \psi) \IMP \vartheta$, then $\derivesIntCk \big( (\xi \cb \varphi ) \AND (\xi \cb \psi)\big) \IMP (\xi\cb \vartheta)$; 	
		
		\item 
		\label{it:cdmono:lemma:validities}
		If $\derivesIntCk\vartheta \IMP  \varphi $, then $\derivesIntCk (\xi \cd \vartheta) \IMP (\xi \cd \varphi)$;

		\item 
		\label{it:cdmulti:lemma:validities}
		If $\derivesIntCk\vartheta \IMP ( \varphi \OR \psi)$, then $\derivesIntCk (\xi \cd \vartheta) \IMP \big( (\xi \cd \varphi) \OR (\xi \cd \psi)\big)$;

		\item 
		\label{it:t1:lemma:validities}
		$\derivesIntCk \big(\varphi \cb (\psi \IMP \chi) \big) \IMP \big( ( \varphi \cb \psi ) \IMP (\varphi \cb \chi) \big)$; 
		
		\item 
		\label{it:t2:lemma:validities}
		$ \derivesIntCk\big(\varphi \cb (\psi \IMP \chi) \big) \IMP \big( ( \varphi \cd \psi ) \IMP (\varphi \cd \chi) \big)$.
		
	\end{enumerate}	
\end{lemma}
\begin{proof}
	Most of the proofs can be found in Lemma 4 of~\cite{olk:2024}. 
	Specifically, \ref{it:cbmono:lemma:validities}  is (\hilbertaxiomstyle{RM}$\BOX$), 
	\ref{it:cdmulti:lemma:validities} is (\hilbertaxiomstyle{RM}$\DIA$), 
	\ref{it:t2:lemma:validities} is (\hilbertaxiomstyle{T1}) 
	and 
	\ref{it:t1:lemma:validities} is (\hilbertaxiomstyle{T2}) from  Lemma 4 of~\cite{olk:2024}. 
	Then, \ref{it:cbmulti:lemma:validities} follows from  \ref{it:cbmono:lemma:validities} and axiom \axCMbox, and 
	\ref{it:cdmulti:lemma:validities} is derived from  \ref{it:cdmono:lemma:validities} and axiom \axCCdiam. 
\end{proof}

The following two lemmas are needed to take output and input contexts into account. In both lemmas, the first item takes care of one-premiss 
rules, while the second item is needed for the branching propositional rules. 

\begin{lemma}
	\label{lemma:filling:output}
	Let $\Delta_1, \Delta_2, \Sigma$ be nested sequents and $\GC{\,}$ be a \bcont. Then:
	\begin{enumerate}
		\item 
		\label{itmono:lemma:filling:output}
		If $\derivesIntCk \fin{\Delta_1}  \IMP \fin{\Sigma}$, then $\derivesIntCk \fin{\GC{\Delta_1}} \IMP \fin{\GC{\Sigma}}$. 
		\item 
		\label{itmulti:lemma:filling:output}
		If $\derivesIntCk \big(  \fin{\Delta_1} \AND \fin{\Delta_2}\big) \IMP \fin{\Sigma}$, then $\derivesIntCk\big( \fin{\GC{\Delta_1}} \AND \fin{\GC{\Delta_2}}\big) \IMP \fin{\GC{\Sigma}}$.  
	\end{enumerate}
\end{lemma}

\begin{proof}
	By induction on the structure of $\GC{\,}$. If  $\GC{\,}$ is $\{\,\}$, both \ref{itmono:lemma:filling:output} and \ref{itmulti:lemma:filling:output} immediately follow. 
	Suppose $\GC{\,}$ is of the shape $\Lambda, \GCp{\,}$. To prove \ref{itmono:lemma:filling:output}, suppose that $\derivesIntCk\fin{\Delta_1}  \IMP \fin{\Sigma}$. By inductive hypothesis,  $\derivesIntCk\fin{\GCp{\Delta_1}}  \IMP \fin{\GCp{\Sigma}}$. 
	Then, by \ref{it:imp:mono:lemma:validities}, we have  
	$\derivesIntCk (\fin{\Lambda} \IMP  \fin{\GCp{\Delta_1}} ) \IMP (\fin{\Lambda}  \IMP \fin{\GCp{\Sigma}} )$. 
	Using \ref{it:impand:lemma:validities}, we obtain 
	that  $\derivesIntCk (\fin{\Lambda} \AND  \fin{\GCp{\Delta_1}} ) \IMP (\fin{\Lambda}  \AND\fin{\GCp{\Sigma}} )$, and we are done. Statement \ref{itmulti:lemma:filling:output} is proved similarly, using \ref{it:imp:multi:lemma:validities}. 
	Next, suppose $\GC{\,}$ is of the shape $\NS{\eta}{ \GCp{\,}}$.  To prove \ref{itmono:lemma:filling:output}, suppose that $\derivesIntCk \fin{\Delta_1}  \IMP \fin{\Sigma}$. By inductive hypothesis,  $\derivesIntCk \fin{\GCp{\Delta_1}}  \IMP \fin{\GCp{\Sigma}}$. 
	By \ref{it:cbmono:lemma:validities}, we have that 
	$\derivesIntCk \eta \cb \fin{\GCp{\Delta_1}}  \IMP \eta \cb \fin{\GCp{\Sigma}}$, and we are done. 
	Case \ref{itmono:lemma:filling:output} is proved similarly, using \cref{it:cdmulti:lemma:validities}. 
\end{proof}

\begin{lemma}
	\label{lemma:filling:input}
	Let $\Delta_1, \Delta_2, \Sigma$ be input sequents and $\GC{\,}$ be a \wcont. Then:
	\begin{enumerate}
		\item 
		\label{itmono:lemma:filling:input}
		If $\derivesIntCk \fin{\Sigma} \IMP \fin{\Delta_1} $, then $ \derivesIntCk\fin{\GC{\Delta_1}} \IMP \fin{\GC{\Sigma}}$. 
		\item 
		\label{itmulti:lemma:filling:input}
		If $\derivesIntCk \fin{\Sigma} \IMP(\fin{\Delta_1} \OR \fin{\Delta_2}) $, then $\derivesIntCk \big(\fin{\GC{\Delta_1}} \AND \fin{\GC{\Delta_2}} \big) \IMP \fin{\GC{\Sigma}}  $.  
	\end{enumerate}
\end{lemma}

\begin{proof}
	By \Cref{remark:contexts}, $\GC{\,}$ can be written as $\GCp{\LC{\,}, \Pi}$, where  $\GCp{\,}$, $\LC{\,}$ are \bcont and $\Pi$ is a nested sequent. 
	To prove \ref{itmono:lemma:filling:input}, we shall first show, by induction on the structure of $\LC{\,}$, that if $\derivesIntCk\fin{\Sigma} \IMP \fin{\Delta_1}$, then $\derivesIntCk\fin{\LC{\Sigma}} \IMP { \fin{\LC{\Delta_1}}}$. Assume that $\derivesIntCk\fin{\Sigma} \IMP \fin{\Delta_1}$. 
	If $\LC{\,} = \{ \,\}$, then the result immediately follows. 
	Next, take $\LC{\,} = \Lambda_1, \Lambda_2\{\,\}$. By inductive hypothesis, we have that $\derivesIntCk \fin{\Lambda_2\{\Sigma\}} \IMP \fin{\Lambda_2\{\Delta_1\}} $. Using  \ref{it:impand:mono:lemma:validities}, we have that $\derivesIntCk (\fin{\Lambda_1} \AND \fin{\Lambda_2\{\Sigma\}}) \IMP (\fin{\Lambda_1} ) \AND\fin{\Lambda_2\{\Delta_1\}} $, whence $\fin{ \Lambda_1 \AND \Lambda_2\{\Sigma\}} \IMP \fin{\Lambda_1 \AND \Lambda_2\{\Delta_1\}}  $. 
	Finally, suppose  $\LC{\,} = \NS{\eta}{\Lambda_1\{\,\}}$. By inductive hypothesis,  $\derivesIntCk \fin{ \Lambda_1\{\Sigma\}} \IMP \fin{\Lambda_1 \{\Delta_1\}} $. Using \ref{it:cdmono:lemma:validities}, we have that $\derivesIntCk  \eta \cd \fin{ \Lambda_1, \{\Sigma\}} \IMP \eta \cd \fin{ \Lambda_1, \{\Delta_1\}}$, i.e., $\derivesIntCk  \fin { \NS{\eta}{\Lambda_1\{\Sigma\}}}  \IMP \fin{\NS{\eta}{\Lambda_1\{\Delta_1\}}}$. This concludes the subproof. 
	
	Next, suppose that $\derivesIntCk\fin{\Sigma} \IMP \fin{\Delta_1}$. By what we have just proved, we have that  $\derivesIntCk\fin{\LC{\Sigma}} \IMP { \fin{\LC{\Delta_1}}}$.  Applying  \ref{it:impimp:lemma:validities}, we obtain $\derivesIntCk  (\fin{\LC{\Delta_1}}  \IMP \fin{\Pi})\IMP ( \fin{\LC {\Sigma}} \IMP \fin{\Pi})$, whence $\derivesIntCk  \fin{\LC{\Delta_1}  ,\Pi}\IMP \fin{\LC {\Sigma}, \Pi}$. Observe that both $\LC{\Delta_1}  ,\Pi$ and $\LC {\Sigma}, \Pi$ are nested sequents, and that $\GCp{\,}$ is a \bcont. 
	Thus, we can apply \ref{itmono:lemma:filling:output} of \Cref{lemma:filling:output} and conclude that $\derivesIntCk  \fin{\GCp{\LC{\Delta_1}  ,\Pi}}\IMP \fin{\GCp{\LC {\Sigma}, \Pi}}$, thus concluding the proof of \ref{itmono:lemma:filling:input}. 
	
	The proof of \ref{itmulti:lemma:filling:input} proceeds similarly. 	We first show, by induction on the structure of $\LC{\,}$, that if $\derivesIntCk\fin{\Sigma} \IMP (\fin{\Delta_1} \OR \fin{\Delta_2})$, then $\derivesIntCk\fin{\LC{\Sigma}} \IMP ( \fin{\LC{\Delta_1}} \OR \fin{\LC{\Delta_2}})$. Assume that $\derivesIntCk\fin{\Sigma} \IMP (\fin{\Delta_1} \OR \fin{\Delta_2})$. 
	If $\LC{\,} = \{ \,\}$, then the result immediately follows. 
	Next, take $\LC{\,} = \Lambda_1, \Lambda_2\{\,\}$. By inductive hypothesis, we have that $\derivesIntCk \fin{\Lambda_2\{\Sigma\}} \IMP (\fin{\Lambda_2\{\Delta_1\}} \OR \fin{\Lambda_2\{\Delta_2\}})$. 
	Using  \ref{it:impor:multi:validities}, we have that $\derivesIntCk (\fin{\Lambda_1} \AND \fin{\Lambda_2\{\Sigma\}}) \IMP \big( (\fin{\Lambda_1} ) \AND\fin{\Lambda_2\{\Delta_1\}} \OR (\fin{\Lambda_1} ) \AND\fin{\Lambda_2\{\Delta_2\}}  \big) $, whence $\fin{ \Lambda_1 \AND \Lambda_2\{\Sigma\}} \IMP (\fin{\Lambda_1 \AND \Lambda_2\{\Delta_1\}}\OR \fin{\Lambda_1 \AND \Lambda_2\{\Delta_1\}})  $. 
	Finally, suppose  $\LC{\,} = \NS{\eta}{\Lambda_1\{\,\}}$. By inductive hypothesis,  $\derivesIntCk \fin{ \Lambda_1\{\Sigma\}} \IMP (\fin{\Lambda_1 \{\Delta_1\}} \OR \fin{\Lambda_1 \{\Delta_2\}})$. 
	Using \ref{it:cdmulti:lemma:validities}, we have that $\derivesIntCk  \eta \cd \fin{ \Lambda_1, \{\Sigma\}} \IMP (\eta \cd \fin{ \Lambda_1, \{\Delta_1\}} \OR \eta \cd \fin{ \Lambda_1, \{\Delta_2\}} )$, i.e., $\derivesIntCk  \fin { \NS{\eta}{\Lambda_1\{\Sigma\}}}  \IMP (\fin{\NS{\eta}{\Lambda_1\{\Delta_1\}}} \OR \fin{\NS{\eta}{\Lambda_1\{\Delta_1\}}})$. This concludes the subproof. 
	
	Next, suppose that $\derivesIntCk\fin{\Sigma} \IMP (\fin{\Delta_1} \OR \fin{\Delta_2})$. By what we have just proved, we have that  $\derivesIntCk\fin{\LC{\Sigma}} \IMP ( \fin{\LC{\Delta_1}} \OR \fin{\LC{\Delta_2}})$.  
	Applying  \ref{it:imporand:multi:validities}, we obtain $\derivesIntCk  \big((\fin{\LC{\Delta_1}}  \IMP \fin{\Pi}) \AND (\fin{\LC{\Delta_2}}  \IMP \fin{\Pi}) \big) \IMP ( \fin{\LC {\Sigma}} \IMP \fin{\Pi})$, 
	whence $\derivesIntCk  \big( \fin{\LC{\Delta_1}  ,\Pi} \AND  \fin{\LC{\Delta_2}  ,\Pi}  \big)\IMP \fin{\LC {\Sigma}, \Pi}$. 
	We have that $\LC{\Delta_1} ,\Pi$ and  $\LC{\Delta_2} ,\Pi$ and $\LC {\Sigma}, \Pi$ are nested sequents, and that $\GCp{\,}$ is a \bcont. 
	Thus, we can apply \ref{itmulti:lemma:filling:output} of \Cref{lemma:filling:output} and conclude that $\derivesIntCk  \big(\fin{\GCp{\LC{\Delta_1}  ,\Pi}}  \AND \fin{\GCp{\LC{\Delta_2}  ,\Pi}} \big)\IMP \fin{\GCp{\LC {\Sigma}, \Pi}}$, thus concluding the proof of \ref{itmulti:lemma:filling:input}. 
\end{proof}


\begin{lemma}
	\label{lemma:sound:rules:nes}
	For any rule of $\NIntCK$ having premisses $\Gamma_k$, for $k \in \{\emptyset, 1,2,3\}$ and conclusion $\Gamma$, it holds that if $\derivesIntCk \fin{\Gamma_k}$, for every $k$, then $\derivesIntCk \fin{\Gamma}$. 
\end{lemma}
\begin{proof}
	The proof that the initial sequents are derivable, i.e., that $\derivesIntCk \GC{\bl{p}, \wh{p}}$ and $\vdash \fin{\GC{\bl{\bot}}}$, is similar to the one in~\cite{strassburger:2013}. 
	
	For the one- and two-premisses rules, it suffices to show that $\big(\bigAND_{k}\fin{\Gamma_k}\big) \IMP \fin{\Gamma}$. This is an easy task because, thanks to \Cref{lemma:filling:output,lemma:filling:input}, we can `ignore' the contexts. More details on the propositional and intuitionistic rules can be found in~\cite{strassburger:2013}. 
	
	By means of example, let us consider rule $\bl{\cd}$. We need to show that $\fin{\GC{\NS{\varphi}{\bl{\psi}}}} \IMP \fin{\GC{\bl{\varphi \cd\psi}}}$. Thanks to \Cref{lemma:filling:input}, it suffices to show that $\fin{{\NS{\varphi}{\bl{\psi}}}} \IMP \fin{\bl{\varphi \cd\psi}}$, that is, $(\varphi \cd \psi) \IMP (\varphi \IMP \psi) $, which is an instance of an axiom. 
	
	Next, let us consider the three-premisses rules. 
	For rule $\wh\cd$, by hypothesis we have that $\derivesIntCk\varphi \IMP\eta$ and $\derivesIntCk \eta \IMP \varphi$. 
	We need to show that the   formula 
	$\big(\eta \cb (\fin{\Delta} \IMP \psi) \big) \IMP \big(( \eta \cd \fin{\Delta} ) \IMP (\varphi \cd \psi) \big)$
	is derivable in $\IntCK$. Since $\derivesIntCk \varphi \IMPP \eta$, this is equivalent to $\big(\varphi \cb (\fin{\Delta} \IMP \psi) \big) \IMP \big(( \varphi \cd \fin{\Delta} ) \IMP (\varphi \cd \psi) \big)$, which is an instance of \ref{it:t2:lemma:validities}, and thus derivable.  
	We conclude the proof applying \Cref{lemma:filling:output}. 
	
	For rule $\bl\cb$, by hypothesis we have that $\derivesIntCk\varphi \IMP\eta$ and $\derivesIntCk \eta \IMP \varphi$. We make a further case distinction, as the formula interpretation changes depending on the position of the output formula. 
	
	If the output formula is in $\GC{\,}$, then we need to show that the  following formula  is derivable in $\IntCK$: 
	$\big( (\varphi \cb \psi) \AND (\eta \cd \fin{\Delta})\big) \IMP \big( \varphi \cb \psi \AND (\eta \cd (\psi \AND \fin{\Delta}))\big) $. 
	Since $\derivesIntCk \varphi \IMPP \eta$, this formula is equivalent to the following: 
	\begin{equation}
		\label{eq:whiteG}
		\big( (\varphi \cb \psi) \AND (\varphi \cd \fin{\Delta})\big)  \IMP \big( \varphi \cb \psi \AND (\varphi \cd (\psi \AND \fin{\Delta})) \big) 	
	\end{equation}
	Consider the following formula, which is an instance of axiom \axCW\xspace and thus derivable: $(\varphi \cb \psi \AND \varphi \cd \fin{\Delta} ) \IMP \varphi \cd (\psi \AND \fin{\Delta})$. By applying \ref{it:impanddis:validities}, we obtain that \cref{eq:whiteG} is derivable. Then, from \Cref{lemma:filling:input}, we have that $\derivesIntCk \fin{\GC{\varphi \cb \psi \AND (\varphi \cd (\psi \AND \fin{\Delta})) }} \IMP \fin{\GC{ (\varphi \cb \psi) \AND (\eta \cd \fin{\Delta})}}$. Since by hypothesis the formula interpretation of the premiss of $\bl{\cb}$ is derivable, a final application of \emph{modus ponens} using the latter derivable formula yields derivability of the formula interpretation of the conclusion of rule $\bl{\cb}$.

	Since $\derivesIntCk \psi \AND \fin{\Delta} \IMP \fin{\Delta}$, this latter formula can be simplified to the formula $
	\big( \varphi \cb \psi \AND (\eta \cd (\fin{\Delta}) \big) \IMP \big( \varphi \cb \psi \AND (\eta \cd \fin{\Delta})\big) $, which is an instance of the identity axiom $\chi \IMP \chi$ and thus derivable. We  conclude the proof by \Cref{lemma:filling:input}. 
	
	Next, suppose that the output formula is in $\Delta$. Then we need to check derivability of the formula: 
	$
	\big( (\varphi \cb \psi) \IMP (\eta \cb (\psi \IMP \fin{\Delta})) \big) \IMP \big( (\varphi \cb \psi) \IMP (\eta \cb \fin{\Delta})  \big)
	$. 
	Since $\derivesIntCk \eta \IMPP \varphi$, this is equivalent to checking derivability of the formula:
	\begin{equation}
		\label{eq:whiteD}
		\big( (\varphi \cb \psi) \IMP (\varphi \cb (\psi \IMP \fin{\Delta})) \big) \IMP \big( (\varphi \cb \psi) \IMP (\varphi \cb \fin{\Delta})  \big)
	\end{equation}
	Consider the following formula, which is an instance of \ref{it:t1:lemma:validities} and thus derivable: 
	$\big(\varphi \cb (\psi \IMP \fin{\Delta}))\big) \IMP \big( (\varphi \cb \psi) \IMP (\varphi \cb \fin{\Delta})  \big)$. By applying \ref{it:impred:lemma:validites} to this latter, we obtain derivability of \cref{eq:whiteD}.  
	We then conclude the proof applying \Cref{lemma:filling:input} and \MP. 
\end{proof}

Soundness of $\NIntCK$ now follows as a corollary of \Cref{lemma:sound:rules:nes}. 

\soundnessnested*

\subsection{Cut elimination for $\NIntCK$}
\label{sec:app:cut-el-nested}

\structuralnes*
\begin{proof}
	Hp- and rp-admissibility of $\wk$ and of $\necs$ follow by an easy induction on the height of the derivation of the premiss of $\wk$ and of $\necs$. 
	In particular, the cases for the intuitionistic rules are as in~\cite{strassburger:2013}, while the cases for conditional rules are similar to those in~\cite{alenda:2016}. In case the last rule applied in the derivation of $\necs$ is $\rep$ or $\cut$, it suffices to apply the inductive hypothesis to the premiss(es) of   $\rep$ or $\cut$, followed by an application of $\rep$ or $\cut$. 
	Hp- and rp-admissibility of $\meds$  is proved by induction on the height of the derivation of $\GC{\NS{\eta}{\Delta_2}}$, and it is similar to the same proof in~\cite{brunnler:2009,strassburger:2013}. For every case, we apply $\wk$ to the premiss of the last rule $\rg$ applied in the derivation of$\GC{\NS{\eta}{\Delta_2}}$, followed by an application of $\rg$.  
	
	Hp- and rp-invertibility of the rules mentioned in $b)$ is also proved by induction on the height of the derivation of the conclusion of each rule. The proofs, which are quite straightforward, are similar to the proofs in\cite{brunnler:2009,strassburger:2013} (for the modal rules) and to the proofs in~\cite{alenda:2016} (fro the conditional rules). 
	
	For hp- and rp-admissibility of $\ctr$ we reason by induction on the height of the derivation of $\GC{\bl{\varphi},\bl{\varphi}}$, and by distinguishing cases according to the last rule $\rg$ applied in the derivation of $\GC{\bl{\varphi},\bl{\varphi}}$. In particular, if $\bl{\varphi}$  is not principal in $\rg$, then we apply the inductive hypothesis,  followed by an application of $\rg$ (to this case belong the cases of $\cut$ and $\rep$). Otherwise, if $\bl{\varphi}$  is principal in $\rg$, then we apply invertibility to the premiss(es) of $\rg$, followed by $\R$ and, possibly, instances of $\ctr$ and $\meds$.  
\end{proof}

\replacement*
\begin{proof}
	By induction on the height $n$ of the derivation $\nder$ of $\GC{\NS{\varphi}{\Delta}}$. 
	If $n=0$, then $\GC{\NS{\varphi}{\Delta}}$ is an initial sequent, and the same holds for $\GC{\NS{\eta}{\Delta}}$. Let us consider the case $n+1$. If the last rule applied in the derivation of $\GC{\NS{\varphi}{\Delta}}$ is \emph{not} $\bl{\cb}$ or $\wh{\cd}$, then it suffices to apply the  inductive hypothesis to the premisses of the rule, followed by an application of the rule. Suppose that the last rule applied is $\bl{\cb}$. Then, the derivations of the three premises of $\rep$ are the following, where $\nder^3$ is the rightmost derivation: 
	\vspace{-0.4cm}
	$$
	\vlderivation{
		\vltr{\nder^1}{ \bl{\varphi}, \wh{\eta} }{\vlhy{\quad}}{\vlhy{}}{\vlhy{}}
	}
	\qquad
	\vlderivation{
		\vltr{\nder^2}{ \bl{\eta}, \wh{\varphi} }{\vlhy{}}{\vlhy{\quad}}{\vlhy{}}
	}
	\qquad
	\vlderivation{
		\vliiin{\bl{\cb}}{}{ \GC{ \bl{\chi \cb \vartheta}, \NS{\varphi}{\Delta} } }{ 
			\vltr{\nder^3_1}{	\bl{\chi}, \wh{\varphi} }{\vlhy{}}{\vlhy{\quad }}{\vlhy{}}
		}{
			\vltr{\nder^3_2}{	\bl{\varphi}, \wh{\chi} }{\vlhy{}}{\vlhy{\quad }}{\vlhy{}}
		}{
			\vltr{\nder^3_3}{ \GC{ \bl{\chi \cb \vartheta}, \NS{\varphi}{\bl{\vartheta}, \Delta} }	 }{\vlhy{}}{\vlhy{\quad }}{\vlhy{}}
		}
	}
	$$
	We construct a derivation $\nder$ of $\GC{ \bl{\chi \cb \vartheta}, \NS{\eta}{\Delta} }	$ as follows, where the double line denotes (possibly multiple) applications of the admissible rules, and the dotted line signifies an application of the inductive hypothesis ($ih$): 
	\vspace{-0.4cm}
	$$
	\vlderivation{
		\vliiin{\cb}{}{\GC{ \bl{\chi \cb \vartheta}, \NS{\eta}{\Delta} }	}{
			\vliin{\cut}{}{\bl{\chi}, \wh{\eta}}{
				\vltr{\nder^3_1}{	\bl{\chi}, \wh{\varphi} }{\vlhy{}}{\vlhy{\quad }}{\vlhy{}}
			}{
				\vliq{\wk}{}{ \bl{\chi}, \bl{\varphi}, \wh{\eta}}{
					\vltr{\nder^1}{	 \bl{\varphi}, \wh{\eta} }{\vlhy{}}{\vlhy{\quad }}{\vlhy{}}
				}
			}
		}{
			\vliin{\cut}{}{\bl{\eta}, \wh{\chi}}{
				\vltr{\nder^2}{	\bl{\eta}, \wh{\varphi} }{\vlhy{}}{\vlhy{\quad }}{\vlhy{}}
			}{
				\vliq{\wk}{}{ \bl{\eta}, \bl{\varphi}, \wh{\chi}}{
					\vltr{\nder^3_2}{  \bl{\varphi}, \wh{\chi}  }{\vlhy{}}{\vlhy{\quad }}{\vlhy{}}
				}
			}
		}{
			\vlid{ih}{}{\GC{ \bl{\chi \cb \vartheta}, \NS{\eta}{\bl{\vartheta}, \Delta} }}{	
				\vltr{\nder^3_3}{ \GC{ \bl{\chi \cb \vartheta}, \NS{\varphi}{\bl{\vartheta}, \Delta} }	 }{\vlhy{}}{\vlhy{\quad }}{\vlhy{}}}
		}
	}
	$$
	In the derivation the $\cut$ rule is used twice, with cut formula $\varphi$. Thus, we have that $\rk{\nder} = \max(\rk{\nder_1}, \rk{\nder_2}, \rk{\nder_3}, \w{\varphi}+1)$. 
	The case in which the last rule applied is $\wh{\cd}$ is similar. 
\end{proof}

\cutadmnes*
\begin{proof}
	By induction on $m+n$, the sum of heights of $\nder^1$ and $\nder^2$, and by distinguishing cases according to the last rules applied in the derivations $\nder^1$ and $\nder^2$. We only detail some cases.

	\begin{description}
		\item[(i)]  At least one of the premisses of $\cut$ is axiomatic. 
		Suppose that the left premiss of $\cut$ is an instance of \init. If it has the form $\GC{\bl{p}, \wh{p}, \wh{\chi}}$, then also $\GC{\bl{p}, \wh{p}}$ is an initial sequent. Otherwise, it has the form $\GC{\bl{q}, \wh{q}}$, for $\chi = q$. The right premiss of $\cut$ has the form $\GC{\bl{q}, \bl{q}}$. Using contraction, we obtain a derivation of $\GC{\bl{q}}$. The other cases are symmetric, or easy. 
		\item[(ii)] The cut formula is \emph{not} principal in at least one of the last rule applied in the derivations of the  left premiss of $\cut$.  In this case, we apply the inductive hypothesis to the premiss(es) of the rule $\rg$ where the cut formula is not principal, and then $\rg$ again.  By means of example, consider the following: 
		$$
		\vlderivation{
			\vltr{\nder_1^1}{\GCD{\wh \chi, \NS{\eta}{\Delta}}}{\vlhy{}}{\vlhy{\quad\,}}{\vlhy{}}
		}
		\qquad 
		\qquad
		\vlderivation{
		\vliiin{\wh{\cd}}{}{\GC{\bl{\chi}, \wh{\varphi \cd \psi}, \NS{\eta}{\Delta} } }{
			\vltr{\nder_2^1}{\bl{\varphi}, \wh{\eta}}{\vlhy{}}{\vlhy{\quad\,}}{\vlhy{}}
		}{
			\vltr{\nder_2^2}{\bl{\eta}, \wh{\varphi}}{\vlhy{}}{\vlhy{\quad\,}}{\vlhy{}}
		}{
			\vltr{\nder_2^3}{ \GC{\bl{\chi},\NS{\eta}{\Delta, \wh{\psi}} } }{\vlhy{}}{\vlhy{\quad\,}}{\vlhy{}}
		}
		}
		$$
		We construct the following derivation: 
		\vspace{-0.3cm}
		$$
		\vlderivation{
			\vliiin{\wh{\cd}}{}{\GC{ \wh{\varphi \cd \psi}, \NS{\eta}{\Delta}  }}{
				\vltr{\nder_2^1}{\bl{\varphi}, \wh{\eta}}{\vlhy{}}{\vlhy{\quad\,}}{\vlhy{}}
			}{
				\vltr{\nder_2^2}{\bl{\eta}, \wh{\varphi}}{\vlhy{}}{\vlhy{\quad\,}}{\vlhy{}}
			}{
			\vliid{\ih}{}{  
			\GCD{ \NS{\eta}{\Delta, \wh{\psi}}}
			}{
					\vltr{\nder_1^1}{\GCD{\wh \chi, \NS{\eta}{\Delta}}}{\vlhy{}}{\vlhy{\quad\,}}{\vlhy{}}
			}{
						\vltr{\nder_2^3}{ \GC{\bl{\chi},\NS{\eta}{\Delta, \wh{\psi}} } }{\vlhy{}}{\vlhy{\quad\,}}{\vlhy{}}
			}
			}
		}
		$$

		\item[(iii)] The cut formula is  principal in both last rule applied in the derivation of both premisses of $\cut$. Propositional and intuitionistic cases can be found in~\cite{strassburger:2013}. We have discussed the case $\bl{\cb}$ and $\wh{\cb}$ in the main text; the case $\bl{\cd}$ and $\wh{\cd}$ is symmetric.  
	\end{description}
\end{proof}

\cutelnes*
\begin{proof}
	Suppose that there is a derivation $\nder$ of  $\nseq$ in $\NIntCKcut$. 
	We show how to construct a derivation $\nder^*$ of $\nseq$ in $\NIntCK$, reasoning by induction on $\rk{\nder}$. If $\rk{\nder} = 0 $, set $\nder^* = \nder$. 
	If $\rk{\nder} > 0 $, select a topmost applications of $\cut$ in $\nder$ having  rank $n+1$ and conclusion $\Gamma$. Thus, the cut formula  $\xi$ of this occurrence of $\cut$ has weight $\w{\xi} = n$,  and the two subderivations of the premisses of $\cut$ have rank $n$. 
	Apply \Cref{lemma:cut-adm} to the premisses of  $\cut$. We obtain a derivation $\nder'$ of $\Gamma$ such that $\rk{\nder'} = n$. 
	The derivation $\nder'$ might contain occurrences of $\rep$, whose rep-formulas  $\zeta$ are such that $\w\zeta < \w\xi$. 
	Next, take the highest occurrence of $\rep$ in $\nder '$, and apply \Cref{lemma:adm:rp} to its premisses.  
	Iterate this procedure until all applications of $\rep$ have been `removed' from $\nder '$. The resulting derivation $\nder^\dagger$ is such that $\rk{\nder^\dagger} \leq n$. 
	Consider the `full' derivation $\nder$, which we have now modified by replacing its subderivation $\nder'$ with $\nder^\dagger$. Select a topmost cut of highest rank, and repeat the procedure. 
\end{proof}

\completenessnes*

\begin{proof}
	We show the derivations of some axioms and inference rules of $\IntCK$ in $\NIntCK$. 
	For reasons of space, we  omit writing the  premisses $\bl{\varphi}, \wh{\varphi}$, derivable by \Cref{prop:gen:init}, of rules $\bl{\cb}$ and $\wh{\cd}$ applied to sequents $\GC{\bl{\varphi \cb \psi}, \NS{\varphi}{\Delta}}$ or $\GC{\wh{\varphi \cd \psi}, \NS{\varphi}{\Delta}}$.
	
	\noindent Derivation of axiom \axCMdiam: 
	$$
	\vlderivation{
		\vlin{\wh{\IMP}}{}{\wh{(\varphi \cd \psi)\lor(\varphi\cd\chi) \imp (\varphi\cd \psi\lor\chi)}}{
			\vliin{\bl{\OR}}{}{ \bl{(\varphi \cd \psi)\lor(\varphi\cd\chi) }, \wh{ (\varphi\cd \psi\lor\chi)} }{
				\vlin{\bl{\cd}}{}{\bl{(\varphi \cd \psi) }, \wh{ (\varphi\cd \psi\lor\chi)}}{
					\vlin{\wh{\cd}}{}{\NS{\varphi}{\bl{\psi}}, \wh{ (\varphi\cd \psi\lor\chi)}}{
						\vlin{\wh{\OR}}{}{\NS{\varphi}{\bl{\psi}, \wh{\psi\lor\chi}}}{
							\vlin{\text{Prop.}~\ref{prop:gen:init}}{}{\NS{\varphi}{\bl{\psi}, \wh{\psi}} }{\vlhy{}}
						}
					}
				}
			}{
				\vlin{\bl{\cd}}{}{\bl{(\varphi \cd \chi) }, \wh{ (\varphi\cd \psi\lor\chi)}}{
					\vlin{\wh{\cd}}{}{\NS{\varphi}{\bl{\chi}}, \wh{ (\varphi\cd \psi\lor\chi)}}{
						\vlin{\wh{\OR}}{}{\NS{\varphi}{\bl{\chi}, \wh{\psi\lor\chi}}}{
							\vlin{\text{Prop.}~\ref{prop:gen:init}}{}{\NS{\varphi}{\bl{\chi}, \wh{\chi}} }{\vlhy{}}
						}
					}
				}
			}
		}
	}
	$$
	Derivation of axiom \axCFS:
	$$
	\vlderivation{
	\vlin{\wh{\IMP}}{}{ \wh{((\varphi \cd \psi)\imp(\varphi\cb\chi)) \imp (\varphi\cb (\psi\imp\chi))} }{
	\vlin{\wh{\cb}}{}{ \bl{(\varphi \cd \psi)\imp(\varphi\cb\chi)},  \wh{(\varphi\cb (\psi\imp\chi))}}{
	\vlin{\wh{\IMP}}{}{  \bl{(\varphi \cd \psi)\imp(\varphi\cb\chi)}, \NS{\varphi}{\wh{\psi\IMP \chi}} }{
	\vliin{\bl{\IMP}}{}{  \bl{(\varphi \cd \psi)\imp(\varphi\cb\chi)}, \NS{\varphi}{\bl{\psi}, \wh{\chi}} }{
	\vlin{\wh{\cd}}{}{ \bl{(\varphi \cd \psi)\imp(\varphi\cb\chi)}, \NS{\varphi}{\bl{\psi} }, \wh{\varphi \cd \psi} }{
	\vlin{\text{Prop.}~\ref{prop:gen:init}}{}{ \bl{(\varphi \cd \psi)\imp(\varphi\cb\chi)}, \NS{\varphi}{\bl{\psi}, \wh{\psi} } }{\vlhy{}}
	}
	}{
	\vlin{\bl{\cb}}{}{ \bl{\varphi\cb\chi},\NS{\varphi}{\bl{\psi}, \wh{\chi}}}{
	\vlin{\text{Prop.}~\ref{prop:gen:init}}{}{  \bl{\varphi\cb\chi},\NS{\varphi}{\bl{\chi}\bl{\psi}, \wh{\chi}} }{ \vlhy{}}
	}
	}
	}
	}
	}
	}
	$$
	
	\noindent Derivation of rule \RAdiam. By hypothesis, we have a derivation $\nder_1$ of 
	 $ \wh{ \varphi \IMP\rho}$ 
	and a derivation $\nder_2$ of 
	$ \wh{\rho \IMP\varphi}$. 
	We derive $\wh{(\varphi \cd \psi) \IMP (\rho \cd \psi)}$, the other direction of the implication being similar.  
	In the derivation, $\mathsf{inv.}$ denotes invertibility of $\wh{\IMP}$ (\Cref{lemma:str-prop}).
	$$
	\vlderivation{
		\vlin{\wh{\IMP}}{}{\wh{(\varphi \cd \psi) \IMP (\rho \cd \psi)}}{
		\vlin{\bl{\cd}}{}{ \bl{\varphi \cd \psi}, \wh{\rho \cd \psi} }{
		\vliiin{\wh{\cd}}{}{ \NS{\varphi}{\bl{\psi}}, \wh{\rho \cd \psi}}{
			\vlid{\mathsf{inv.}}{}{ \bl{\varphi}, \wh{\rho} }{
			\vltr{\nder_1}{  \wh{ \varphi \IMP\rho} }{\vlhy{}}{\vlhy{\quad\,}}{\vlhy{}}
		}
		}{
		\vlid{\mathsf{inv.}}{}{ \bl{\rho}, \wh{\varphi} }{
		\vltr{\nder_2}{  \wh{ \rho \IMP\varphi} }{\vlhy{}}{\vlhy{\quad\,}}{\vlhy{}}
		}
		}{
		\vlin{\text{Prop.}~\ref{prop:gen:init}}{}{\NS{\varphi}{\bl{\psi}\wh{\psi}} }{\vlhy{}}
		}
		}
		}
	}
	$$

\end{proof}

\section{Cut admissibility in sequent calculi $\SCCKstar$}\label{app:cut elim SCCKstar}

\subsection{Cut admissibility in $\SCCK$}

\begin{proposition}[Hp-admissibility of weakening]
The rules $\lwk$ and $\rwk$ are height-preserving admissible in $\SCCK$.
\end{proposition}
\begin{proof}
By induction on the height of the derivation of the premiss.
The proof extend the standard proof for \textsf{G3ip} with the new cases induced by the conditional rules.
Suppose that the last rule applied in the derivation of the premiss is a conditional rule.
Then the conclusion of $\lwk$ or $\rwk$ can be obtained by a different application of the same conditional rule. For example, if we have
\begin{center}
\begin{small}
\ax{$\varphi \seqder{1} \rho$} 
\ax{$\rho \seqder{2} \varphi$}
\ax{$\sigma, \psi \seqder{3}$}
\rlab{\rulecbd}
\tinf{$\G, \rho \cb \sigma, \varphi \cd \psi \seq$}
\disp
\end{small}
\end{center}
Then the consequences $\G, \chi, \rho \cb \sigma, \varphi \cd \psi \seq$ \ and $\G, \rho \cb \sigma, \varphi \cd \psi \seq \vartheta$ of respectively $\lwk$ and $\rwk$ can be obtained by
\begin{center}
\begin{small}
\ax{$\varphi \seqder{1} \rho$} 
\ax{$\rho \seqder{2} \varphi$}
\ax{$\sigma, \psi \seqder{3}$}
\rlab{\rulecbd}
\tinf{$\G, \chi, \rho \cb \sigma, \varphi \cd \psi \seq$}
\disp
\ \
\ax{$\varphi \seqder{1} \rho$} 
\ax{$\rho \seqder{2} \varphi$}
\ax{$\sigma, \psi \seqder{3}$}
\rlab{\rulecbd}
\tinf{$\G, \rho \cb \sigma, \varphi \cd \psi \seq \vartheta$}
\disp
\end{small}
\end{center}\qed
\end{proof}

\begin{proposition}[Hp-invertibility]
The rules \lland, \rland, \llor, \rlorone, \rlortwo, \rimp\ are height-preserving invertible, 
and the rule \limp\ is height-preserving invertible with respect to the right premiss.
\end{proposition}
\begin{proof}
By induction on the height of the derivation $\mathcal D$ of the conclusion of the rules of $\ctr$, extending the proof for \textsf{G3ip} with the cases involving the conditional rules. 
For instance, suppose we have:
\begin{center}
\begin{small}
\ax{$\seqder{} \sigma$}
\rlab{\rulecb}
\uinf{$\G, \varphi \imp \psi \seq\rho \cb \sigma$}
\disp
\end{small}
\end{center}
The sequent $\G, \varphi \imp \psi \seq\rho \cb \sigma$ is an instance of the conclusion of \limp.
We can derive the right premiss of \limp\ as follows:
\begin{center}
\begin{small}
\ax{$\seqder{} \sigma$}
\rlab{\rulecb}
\uinf{$\G, \psi \seq\rho \cb \sigma$}
\disp
\end{small}
\end{center}
\qed
\end{proof}

\begin{proposition}[Hp-admissibility of contraction]
The rule $\ctr$ is height-preserving admissible in $\SCCK$.
\end{proposition}
\begin{proof}
By induction on the height of the derivation $\mathcal D$ of the premiss of $\ctr$, extending the proof for \textsf{G3ip} with the conditional rules. Consider as examples the following cases.

\medskip
\noindent
(i) The contracted formula is not principal in the last rule applied in $\mathcal D$. For instance:
\begin{center}
\begin{small}
\ax{$\varphi \seqder{1} \rho$} 
\ax{$\rho \seqder{2} \varphi$}
\ax{$\sigma, \psi \seqder{3}$}
\rlab{\rulecbd}
\tinf{$\G, \chi, \chi, \rho \cb \sigma, \varphi \cd \psi \seq$}
\disp
\end{small}
\end{center}
We can obtain the consequence $\G, \chi, \rho \cb \sigma, \varphi \cd \psi \seq$ of $\ctr$ as follows:
\begin{center}
\begin{small}
\ax{$\varphi \seqder{1} \rho$} 
\ax{$\rho \seqder{2} \varphi$}
\ax{$\sigma, \psi \seqder{3}$}
\rlab{\rulecbd}
\tinf{$\G, \chi, \rho \cb \sigma, \varphi \cd \psi \seq$}
\disp
\end{small}
\end{center}

\medskip
\noindent
(ii) Only one occurrence of the contracted formula is principal in the last rule applied in $\mathcal D$. For instance:
\begin{center}
\begin{small}
\ax{$\varphi \seqder{1} \rho$} 
\ax{$\rho \seqder{2} \varphi$}
\ax{$\sigma, \psi \seqder{3}$}
\rlab{\rulecbd}
\tinf{$\G, \rho \cb \sigma, \rho \cb \sigma, \varphi \cd \psi \seq$}
\disp
\end{small}
\end{center}
We can obtain the consequence $\G, \rho \cb \sigma, \varphi \cd \psi \seq$ of $\ctr$ as follows:
\begin{center}
\begin{small}
\ax{$\varphi \seqder{1} \rho$} 
\ax{$\rho \seqder{2} \varphi$}
\ax{$\sigma, \psi \seqder{3}$}
\rlab{\rulecbd}
\tinf{$\G, \rho \cb \sigma, \varphi \cd \psi \seq$}
\disp
\end{small}
\end{center}

\medskip
\noindent
(iii) Both occurrences of the contracted formula are principal in the last rule applied in $\mathcal D$. For instance:
\begin{center}
\begin{small}
\ax{$\varphi \seqder{1} \rho$ \ \ \ 
$\rho \seqder{2} \varphi$ \ \ \ 
$\varphi \seqder{3} \rho$ \ \ \
$\rho \seqder{4} \varphi$ \ \ \ 
$\sigma, \sigma, \psi \seqder{5}$}
\rlab{\rulecbd}
\uinf{$\G, \rho \cb \sigma, \rho \cb \sigma, \varphi \cd \psi \seq$}
\disp
\end{small}
\end{center}
We can obtain the consequence $\G, \rho \cb \sigma, \varphi \cd \psi \seq$ of $\ctr$ as follows, where the dotted line represent an application of the inductive hypothesis:
\begin{center}
\begin{small}
\ax{$\varphi \seqder{1} \rho$}
\ax{$\rho \seqder{2} \varphi$}
\ax{$\sigma, \sigma, \psi \seqder{5}$}
\dottedLine
\rlab{$i.h.$}
\uinf{$\sigma, \psi \seq$}
\rlab{\rulecbd}
\tinf{$\G, \rho \cb \sigma, \varphi \cd \psi \seq$}
\disp
\end{small}
\end{center}\qed
\end{proof}

\begin{theorem}[Admissibility of cut]\label{appth:cut elim SCCK}
The rule $\cut$ is admissible in $\SCCK$.
\end{theorem}
\begin{proof}
By induction on lexicographically ordered pairs ($c$, $h$), 
where $c$ is the weight 
of the cut formula $\varphi$, and $h$ is the sum of the heights of the derivations of the two premisses of $\cut$.
We only show the inductive step of the proof, distinguishing some cases according to whether the cut formula is or not principal
in the last rules applied in the derivations of the premisses of $\cut$.
Moreover, we only consider the cases involving the conditional rules, the other cases are exactly as in the standard proof of cut elimination for the intuitionistic propositional calculus \textsf{G3ip} (see e.g. \cite{troelstra:2000}).
For the sake of readability, we show applications of conditional rules with only one or two principal $\cb$-formulas that however contain all information needed to construct the cases of rule applications with more principal $\cb$-formulas.

\bigskip
\noindent
(i) The cut formula is not principal in the last rule applied in the derivation of the left premiss of $\cut$.
Then, the only conditional rule that can be applied as last rule is \rulecbd.

\begin{center}
\begin{small}
\ax{$\varphi \seqder{1} \rho$ \ \ $\rho \seqder{2} \varphi$ \ \ $\sigma, \psi \seqder{3}$}
\rlab{\rulecbd}
\uinf{$\G, \rho \cb \sigma, \varphi \cd \psi \seq \chi$}
\disp
\qquad
$\G', \chi \seqder{4} \D$
\end{small}
\end{center}
The consequence $\G, \G', \rho \cb \sigma, \varphi \cd \psi \seq \D$ of $\cut$ can be derived as follows:

\begin{center}
\begin{small}
\ax{$\varphi \seqder{1} \rho$ \ \ $\rho \seqder{2} \varphi$ \ \ $\sigma, \psi \seqder{3}$}
\rlab{\rulecbd}
\uinf{$\G, \G', \rho \cb \sigma, \varphi \cd \psi \seq \D$}
\disp
\end{small}
\end{center}

\bigskip
\noindent
(ii) 
The cut formula is not principal in the last rule applied in the derivation of the right premiss of $\cut$.
Consider the cases where the last rule applied is a conditional rule.
In all there cases the consequence of $\cut$ ca be obtained by means of a different application of the same rule.
For instance, suppose we have:
\begin{center}
\begin{small}
$\G \seqder{1} \chi$ \qquad
\ax{$\varphi \seqder{2} \rho$ \ \ $\rho \seqder{3} \varphi$ \ \ $\sigma, \psi \seqder{4}$}
\rlab{\rulecbd}
\uinf{$\G', \chi, \rho \cb \sigma, \varphi \cd \psi \seq \D$}
\disp
\end{small}
\end{center}
The consequence $\G, \G', \rho \cb \sigma, \varphi \cd \psi \seq \D$ of cut can be obtained as follows:
\begin{center}
\begin{small}
\ax{$\varphi \seqder{2} \rho$ \ \ $\rho \seqder{3} \varphi$ \ \ $\sigma, \psi \seqder{4}$}
\rlab{\rulecbd}
\uinf{$\G, \G', \rho \cb \sigma, \varphi \cd \psi \seq \D$}
\disp
\end{small}
\end{center}

\bigskip
\noindent
(iii)
The cut formula is principal in the last rules applied in the derivations of both premisses of $\cut$.
As before, we only consider the possible cases involving conditional rules.

\bigskip
$\bullet$ The last rules applied in the derivations of the premisses of $\cut$ are, respectively, \rulecb\ and \rulecb:
\begin{center}
\begin{small}
\ax{$\varphi \seqder{1} \rho$ \ \ $\rho \seqder{2} \varphi$ \ \ $\sigma \seqder{3} \psi$}
\rlab{\rulecb}
\uinf{$\G, \rho \cb \sigma \seq \varphi \cb \psi$}
\disp
\quad
\ax{$\varphi \seqder{4} \xi$ \ \ $\xi \seqder{5} \varphi$ \ \  $\eta \seqder{6} \xi$ \ \ $\xi \seqder{7} \eta$ \ \ $\psi, \vartheta \seqder{8} \chi$}
\rlab{\rulecb}
\uinf{$\G', \varphi \cb \psi, \eta \cb \vartheta \seq \xi \cb \chi$}
\disp
\end{small}
\end{center}
The consequence $\G, \G', \rho \cb \sigma, \eta \cb \vartheta \seq \xi \cb \chi$  of $\cut$ can be derived as follows,
where dotted lines represent applications of the induction hypothesis:
%
\begin{center}
\begin{small}
\ax{$\xi \seqder{5} \varphi$ \ \ $\varphi \seqder{1} \rho$}
\dottedLine
\uinf{$\xi \seq \rho$}
\ax{$\rho \seqder{2} \varphi$ \ \ $\varphi \seqder{4} \xi$}
\dottedLine
\uinf{$\rho \seq \xi$}
\ax{$\eta \seqder{6} \xi$ \ $\xi \seqder{7} \eta$} 
\ax{$\sigma \seqder{3} \psi$ \ \ $\psi, \vartheta \seqder{8} \chi$}
\dottedLine
\uinf{$\sigma,\vartheta \seq \chi$}
\rlab{\rulecb}
\qinf{$\G, \G', \rho \cb \sigma, \eta \cb \vartheta \seq \xi \cb \chi$}
\disp
\end{small}
\end{center}

\bigskip
$\bullet$ The last rules applied in the derivations of the premisses of $\cut$ are, respectively, \rulecb\ and \rulecd:
\begin{center}
\begin{small}
\ax{$\varphi \seqder{1} \rho$ \ \ $\rho \seqder{2} \varphi$ \ \ $\sigma \seqder{3} \psi$}
\rlab{\rulecb}
\uinf{$\G, \rho \cb \sigma \seq \varphi \cb \psi$}
\disp
\quad
\ax{$\eta \seqder{4} \varphi$ \ \ $\varphi \seqder{5} \eta$ \  \ $\eta \seqder{6} \xi$ \ \ $\xi \seqder{7} \eta$ \ \ $\psi, \vartheta \seqder{8} \chi$}
\rlab{\rulecd}
\uinf{$\G', \varphi \cb \psi, \eta \cd \vartheta \seq \xi \cd \chi$}
\disp
\end{small}
\end{center}
The consequence $\G, \G', \rho \cb \sigma, \eta \cd \vartheta \seq \xi \cd \chi$  of $\cut$ can be derived as follows,
where dotted lines represent applications of the induction hypothesis:
%
\begin{center}
\begin{small}
\ax{$\eta \seqder{4} \varphi$ \ \ $\varphi \seqder{1} \rho$}
\dottedLine
\uinf{$\eta \seq \rho$}
\ax{$\rho \seqder{2} \varphi$ \ \ $\varphi \seqder{5} \eta$}
\dottedLine
\uinf{$\rho \seq \eta$}
\ax{$\eta \seqder{6} \xi$ \ $\xi \seqder{7} \eta$} 
\ax{$\sigma \seqder{3} \psi$ \ \ $\psi, \vartheta \seqder{8} \chi$}
\dottedLine
\uinf{$\sigma,\vartheta \seq \chi$}
\rlab{\rulecd}
\qinf{$\G, \G', \rho \cb \sigma, \eta \cd \vartheta \seq \xi \cd \chi$}
\disp
\end{small}
\end{center}

\bigskip
$\bullet$ The last rules applied in the derivations of the premisses of $\cut$ are, respectively, \rulecb\ and \rulecbd:
\begin{center}
\begin{small}
\ax{$\varphi \seqder{1} \rho$ \ \ $\rho \seqder{2} \varphi$ \ \ $\sigma \seqder{3} \psi$}
\rlab{\rulecb}
\uinf{$\G, \rho \cb \sigma \seq \varphi \cb \psi$}
\disp
\quad
\ax{$\eta \seqder{4} \varphi$ \ \ $\varphi \seqder{5} \eta$ \ \ $\psi, \vartheta \seqder{6}$}
\rlab{\rulecbd}
\uinf{$\G', \varphi \cb \psi, \eta \cd \vartheta \seq \D$}
\disp
\end{small}
\end{center}
The consequence $\G, \G', \rho \cb \sigma, \eta \cd \vartheta \seq \D$ of $\cut$ can be derived as follows,
where dotted lines represent applications of the induction hypothesis:
%
\begin{center}
\begin{small}
\ax{$\eta \seqder{4} \varphi$ \ \ $\varphi \seqder{1} \rho$}
\dottedLine
\uinf{$\eta \seq \rho$}
\ax{$\rho \seqder{2} \varphi$ \ \ $\varphi \seqder{5} \eta$}
\dottedLine
\uinf{$\rho \seq \eta$}
\ax{$\sigma \seqder{3} \psi$ \ \ $\psi, \vartheta \seqder{6}$}
\dottedLine
\uinf{$\sigma,\vartheta \seq$}
\rlab{\rulecbd}
\tinf{$\G, \G', \rho \cb \sigma, \eta \cd \vartheta \seq \D$}
\disp
\end{small}
\end{center}

\bigskip
$\bullet$ The last rules applied in the derivations of the premisses of $\cut$ are, respectively, \rulecd\ and \rulecd:
\begin{center}
\begin{small}
\ax{$\varphi \sseqder{1}{2} \rho$ \quad $\varphi \sseqder{3}{4} \eta$ \quad $\sigma, \psi \seqder{5} \vartheta$}
\rlab{\rulecd}
\uinf{$\G, \rho \cb \sigma, \varphi \cd \psi \seq \eta \cd \vartheta$}
\disp
\hfill
\ax{$\eta \sseqder{6}{7} \xi$ \quad  $\eta \sseqder{8}{9} \kappa$ \quad $\chi, \vartheta \seqder{10} \zeta$}
\rlab{\rulecd}
\uinf{$\G', \xi \cb \chi, \eta \cd \vartheta \seq \kappa \cd \zeta$}
\disp
\end{small}
\end{center}
The consequence $\G, \G', \rho \cb \sigma, \xi \cb \chi, \varphi \cd \psi \seq \kappa \cd \zeta$  of $\cut$ can be derived as follows,
where dotted lines represent applications of the induction hypothesis:
%
\begin{center}
\begin{tiny}
\ax{$\varphi \sseqder{1}{2} \rho$} 
\ax{$\varphi \seqder{3} \eta$ \ \ $\eta \seqder{6} \xi$}
\dottedLine
\uinf{$\varphi \seq \xi$}
\ax{$\xi \seqder{7} \eta$ \ \ $\eta \seqder{4} \varphi$}
\dottedLine
\uinf{$\xi \seq \varphi$}
\ax{$\eta \sseqder{8}{9} \kappa$}
\ax{$\sigma, \psi \seqder{5} \vartheta$ \ \ $\chi, \vartheta \seqder{10} \zeta$}
\dottedLine
\uinf{$\sigma, \chi, \psi \seq \zeta$}
\rlab{\rulecd}
\QuinaryInfC{$\G, \G', \rho \cb \sigma, \xi \cb \chi, \varphi \cd \psi \seq \kappa \cd \zeta$}
\disp
\end{tiny}
\end{center}

\bigskip
$\bullet$ The last rules applied in the derivations of the premisses of $\cut$ are, respectively, \rulecd\ and \rulecbd:
\begin{center}
\begin{small}
\ax{$\varphi \seqder{1} \rho$ \ \ $\rho \seqder{2} \varphi$ \ \ $\varphi \seqder{3} \eta$ \ \ $\eta \seqder{4} \varphi$ \ \ $\sigma, \psi \seqder{5} \vartheta$}
\rlab{\rulecd}
\uinf{$\G, \rho \cb \sigma, \varphi \cd \psi \seq \eta \cd \vartheta$}
\disp
\ \
\ax{$\eta \seqder{6} \xi$ \ \ $\xi \seqder{7} \eta$ \ \ $\chi, \vartheta \seqder{8}$}
\rlab{\rulecbd}
\uinf{$\G', \xi \cb \chi, \eta \cd \vartheta \seq \D$}
\disp
\end{small}
\end{center}
The consequence $\G, \G', \rho \cb \sigma, \xi \cb \chi, \varphi \cd \psi \seq \D$  of $\cut$ can be derived as follows,
where dotted lines represent applications of the induction hypothesis:
%
\begin{center}
\begin{small}
\ax{$\varphi \seqder{1} \rho$ \ \ $\rho \seqder{2} \varphi$} 
\ax{$\varphi \seqder{3} \eta$ \ \ $\eta \seqder{6} \xi$}
\dottedLine
\uinf{$\varphi \seq \xi$}
\ax{$\xi \seqder{7} \eta$ \ \ $\eta \seqder{4} \varphi$}
\dottedLine
\uinf{$\xi \seq \varphi$}
\ax{$\sigma, \psi \seqder{5} \vartheta$ \ \ $\chi, \vartheta \seqder{8}$}
\dottedLine
\uinf{$\sigma, \chi, \psi \seq$}
\rlab{\rulecbd}
\qinf{$\G, \G', \rho \cb \sigma, \xi \cb \chi, \varphi \cd \psi \seq \D$}
\disp
\end{small}
\end{center}
\qed
\end{proof}

\subsection{Cut admissibility in $\SCCKMP$}\label{app:cut elim ext}

\begin{proposition}
The following rule is admissible in $\SCCKMP$,
where $\Phi \cb \Psi = \varphi_1 \cb \psi_1, ..., \varphi_n \cb \psi_n$ and $n \geq 1$:
%
\begin{center}
\begin{small}
\ax{$\{\G, \Phi \cb \Psi \seq \varphi_i\}_{i \leq n}$} 
\ax{$\G, \Phi \cb \Psi, \psi_1, ..., \psi_n \seq \D$}
\llab{\rulempboxstar}
\binf{$\G, \Phi \cb \Psi  \seq \D$}
\disp.
\end{small}
\end{center}
\end{proposition}

\begin{proof}
The rule \rulempboxstar\ can be shown admissible as follows, where ($A)$ stays for the sequent $\G, \Phi \cb \Psi, \psi_1, ..., \psi_n \seq \D$.
First, for all $1 \leq i \leq n-1$, we obtain
\begin{center}
\begin{small}
\ax{$\G, \Phi \cb \Psi \seq \varphi_i$} 
\doubleLine
\llab{$\lwk$}
\uinf{$\G, \Phi \cb \Psi, \psi_{i+1}, ..., \psi_n \seq \varphi_i$ \ ($B_i$)} 
\disp
\end{small}
\end{center}
Then:
\begin{center}
\begin{small}
\ax{$\G, \Phi \cb \Psi \seq \varphi_n$} 
\ax{$\vdots$}
\ax{($B_2$)}
\ax{($B_1$)}
\ax{($A$)}
\rlab{\rulempbox}
\binf{$\G, \Phi \cb \Psi, \psi_2, ..., \psi_n \seq \D$}
\rlab{\rulempbox}
\binf{$\G, \Phi \cb \Psi, \psi_3, ..., \psi_n \seq \D$}
\noLine
\uinf{$\vdots$}
\binf{$\G, \Phi \cb \Psi, \psi_n \seq \D$}
\rlab{\rulempbox}
\binf{$\G, \Phi \cb \Psi \seq \D$}
\disp
\end{small}
\end{center}
\qed
\end{proof}

\begin{theorem}
The rule $\cut$ is admissible in $\SCCKMP$.
\end{theorem}

\begin{proof}
By extending the proof for $\SCCK$ with the cases involving the rules \rulempbox\ and \rulempdiam.
We only show that cases where the cut formula is principal in the last rules applied in the derivations of both premisses of $\cut$.

\bigskip
$\bullet$ The last rules applied in the derivations of the premisses of $\cut$ are, respectively, \rulecb\ and \rulempbox:
\begin{center}
\begin{small}

\ax{$\varphi \seqder{1} \rho_1$ \quad $\rho_1 \seqder{2} \varphi$ \quad ... \quad $\varphi \seqder{2n - 1} \rho_n$ \quad $\rho_n \seqder{2n} \varphi$ \quad $\sigma_1, ..., \sigma_n \seqder{2n + 1} \psi$}
\rlab{\rulecb}
\uinf{$\G, \rho_1 \cb \sigma_1, ..., \rho_n \cb \sigma_n \seq \varphi \cb \psi$}
\disp

\vspace{0.3cm}
\ax{$\G', \varphi \cb \psi \seqder{2n + 2} \varphi$}
\ax{$\G', \varphi \cb \psi, \psi \seqder{2n + 3} \D$}
\rlab{\rulempbox}
\binf{$\G', \varphi \cb \psi \seq \D$}
\disp
\end{small}
\end{center}

\noindent
The consequence $\G, \G', \rho_1 \cb \sigma_1, ..., \rho_n \cb \sigma_n \seq \D$ of $\cut$ can be derived as follows,
where $P \cb \Sigma = \rho_1 \cb \sigma_1, ..., \rho_n \cb \sigma_n$.
First, for all $1 \leq i \leq n$, we obtain:
\begin{center}
\begin{small}
\ax{$\G, P \cb \Sigma \seq \varphi \cb \psi$}
\ax{$\G', \varphi \cb \psi \seqder{2n + 2} \varphi$}
\dottedLine
\binf{$\G, \G', P \cb \Sigma \seq \varphi$}
\ax{$\varphi \seqder{2i - 1} \rho_i$}
\dottedLine
\binf{$\G, \G', P \cb \Sigma \seq \rho_i$ \ \ ($A_i$)}
\disp
\end{small}
\end{center}
Then:
\begin{center}
\begin{small}
\ax{$\sigma_1, ..., \sigma_n \seqder{2n + 1} \psi$}
\ax{$\G, P \cb \Sigma \seq \varphi \cb \psi$}
\ax{$\G', \varphi \cb \psi, \psi \seqder{2n + 3} \D$}
\dottedLine
\binf{$\G, \G', P \cb \Sigma, \psi \seq \D$}
\dottedLine
\binf{$\G, \G', P \cb \Sigma, \sigma_1, ..., \sigma_n \seq \D$ \ \ ($B$)}
\disp
\end{small}
\end{center}
Finally:
\begin{center}
\begin{small}
\ax{($A_1$) \quad ... \quad ($A_n$) \quad ($B$)}
\rlab{\rulempboxstar}
\uinf{$\G, \G', P \cb \Sigma \seq \D$}
\disp
\end{small}
\end{center}

\bigskip
$\bullet$ The last rules applied in the derivations of the premisses of $\cut$ are, respectively, \rulempdiam\ and \rulecd:
\begin{center}
\begin{small}
\ax{$\G \seqder{1} \varphi$} 
\ax{$\G \seqder{2} \psi$}
\rlab{\rulempdiam}
\binf{$\G \seq \varphi \cd \psi$}
\disp

\vspace{0.3cm}
\ax{$\{\varphi \sseqder{2i + 1}{2i + 2} \rho_i \}_{i \leq n}$}
\ax{$\varphi \sseqder{2n + 3}{2n + 4} \eta$}
\ax{$\sigma_1, ..., \sigma_n, \psi \seqder{2n + 5} \vartheta$}
\rlab{\rulecd}
\tinf{$\G', \rho_1 \cb \sigma_1, ..., \rho_n \cb \sigma_n, \varphi \cd \psi \seq \eta \cd \vartheta$}
\disp
\end{small}
\end{center}
The consequence $\G, \G', \rho_1 \cb \sigma_1, ..., \rho_n \cb \sigma_n \seq \eta \cd \vartheta$ of $\cut$ can be derived as follows,
where $P \cb \Sigma = \rho_1 \cb \sigma_1, ..., \rho_n \cb \sigma_n$.
\begin{center}
\begin{small}
\ax{$\G \seqder{1} \varphi$} 
\ax{$\varphi \seqder{2i + 1} \rho_i$}
\llab{$\big\{$}
\rlab{$\big\}_{i \leq n}$}
\dottedLine
\binf{$\G \seq \rho_i$} 
\doubleLine
\rlab{$\lwk$}
\uinf{$\G, P \cb \Sigma \seq \rho_i$} 
\ax{$\G \seqder{2} \psi$}
\ax{$\sigma_1, ..., \sigma_n, \psi \seqder{2n + 5} \vartheta$}
\dottedLine
\binf{$\G, \sigma_1, ..., \sigma_n \seq \vartheta$}
\doubleLine
\rlab{$\lwk$}
\uinf{$\G, P \cb \Sigma, \sigma_1, ..., \sigma_n \seq \vartheta$}
\rlab{\rulempboxstar}
\binf{$\G, P \cb \Sigma \seq \vartheta$ \ \ ($A$)} 
\disp

\vspace{0.3cm}
\ax{$\G \seqder{1} \varphi$} 
\ax{$\varphi \seqder{2n + 3}\eta$}
\dottedLine
\binf{$\G \seq \eta$}
\doubleLine
\rlab{$\lwk$}
\uinf{$\G, P \cb \Sigma \seq \eta$}
\ax{($A$)} 
\rlab{\rulempdiam}
\binf{$\G, P \cb \Sigma \seq \eta \cd \vartheta$}
\disp
\end{small}
\end{center}

\bigskip
$\bullet$ The last rules applied in the derivations of the premisses of $\cut$ are, respectively, \rulempdiam\ and \rulecbd:

\begin{center}
\begin{small}
\ax{$\G \seqder{1} \varphi$} 
\ax{$\G \seqder{2} \psi$}
\rlab{\rulempdiam}
\binf{$\G \seq \varphi \cd \psi$}
\disp
\hfill
\ax{$\{\varphi \sseqder{2i + 1}{2i + 2} \rho_i \}_{i \leq n}$}
\ax{$\sigma_1, ..., \sigma_n, \psi \seqder{2n + 3}$}
\rlab{\rulecbd}
\binf{$\G', \rho_1 \cb \sigma_1, ..., \rho_n \cb \sigma_n, \varphi \cd \psi \seq \D$}
\disp
\end{small}
\end{center}

\noindent
The consequence $\G, \G', \rho_1 \cb \sigma_1, ..., \rho_n \cb \sigma_n \seq \D$ of $\cut$ can be derived as follows,
where $P \cb \Sigma = \rho_1 \cb \sigma_1, ..., \rho_n \cb \sigma_n$.

\begin{center}
\begin{small}
\ax{$\G \seqder{1} \varphi$} 
\ax{$\varphi \seqder{2i + 1} \rho_i$}
\llab{$\big\{$}
\rlab{$\big\}_{i \leq n}$}
\dottedLine
\binf{$\G \seq \rho_i$} 
\doubleLine
\rlab{$\lwk$}
\uinf{$\G, P \cb \Sigma \seq \rho_i$} 
\ax{$\G \seqder{2} \psi$}
\ax{$\sigma_1, ..., \sigma_n, \psi \seqder{2n + 3}$}
\dottedLine
\binf{$\G, \sigma_1, ..., \sigma_n \seq$}
\doubleLine
\rlab{$\lwk$}
\uinf{$\G, P \cb \Sigma, \sigma_1, ..., \sigma_n \seq$}
\rlab{\rulempboxstar}
\binf{$\G, P \cb \Sigma \seq$}
\doubleLine
\rlab{$\rwk$}
\uinf{$\G, P \cb \Sigma \seq \D$}
\disp
\end{small}
\end{center}
\qed
\end{proof}

\subsection{Cut admissibility in $\SCCKID$}

\begin{theorem}
The rule $\cut$ is admissible in $\SCCKID$.
\end{theorem}
\begin{proof}
By modifying the proof for $\SCCK$, replacing \rulecb, \rulecd, \rulecbd\ with  \rulecbid, \rulecdid, \rulecbdid.
As before, we only show that cases where the cut formula is principal in the last rules applied in the derivations of both premisses of $\cut$.

\bigskip
$\bullet$ The last rules applied in the derivations of the premisses of $\cut$ are, respectively, \rulecbid\ and \rulecbid:
\begin{center}
\begin{small}
\ax{$\varphi \seqder{1} \rho$ \ $\rho \seqder{2} \varphi$ \ $\sigma, \varphi \seqder{3} \psi$}
\rlab{\rulecbid}
\uinf{$\G, \rho \cb \sigma \seq \varphi \cb \psi$}
\disp
\
\ax{$\varphi \seqder{4} \xi$ \ $\xi \seqder{5} \varphi$ \  $\eta \seqder{6} \xi$ \ $\xi \seqder{7} \eta$ \ $\psi, \vartheta, \xi \seqder{8} \chi$}
\rlab{\rulecbid}
\uinf{$\G', \varphi \cb \psi, \eta \cb \vartheta \seq \xi \cb \chi$}
\disp
\end{small}
\end{center}
The consequence $\G, \G', \rho \cb \sigma, \eta \cb \vartheta \seq \xi \cb \chi$  of $\cut$ can be derived as follows,
where dotted lines represent applications of the induction hypothesis:
%
\begin{center}
\begin{small}
\ax{$\xi \seqder{5} \varphi$}
\ax{$\sigma, \varphi \seqder{3} \psi$}
\ax{$\psi, \vartheta, \xi \seqder{8} \chi$}
\dottedLine
\binf{$\sigma, \varphi, \vartheta, \xi \seq \chi$}
\dottedLine
\binf{$\sigma, \xi, \vartheta, \xi \seq \chi$}
\rlab{$\ctr$}
\uinf{$\sigma, \vartheta, \xi \seq \chi$ \ ($A$)}
\disp

\vspace{0.3cm}
\ax{$\xi \seqder{5} \varphi$ \ \ $\varphi \seqder{1} \rho$}
\dottedLine
\uinf{$\xi \seq \rho$}
\ax{$\rho \seqder{2} \varphi$ \ \ $\varphi \seqder{4} \xi$}
\dottedLine
\uinf{$\rho \seq \xi$}
\ax{$\eta \seqder{6} \xi$ \ $\xi \seqder{7} \eta$} 
\ax{($A$)}
\rlab{\rulecbid}
\qinf{$\G, \G', \rho \cb \sigma, \eta \cb \vartheta \seq \xi \cb \chi$}
\disp
\end{small}
\end{center}

\bigskip
$\bullet$ The last rules applied in the derivations of the premisses of $\cut$ are, respectively, \rulecbid\ and \rulecdid:
\begin{center}
\begin{small}
\ax{$\varphi \seqder{1} \rho$ \ $\rho \seqder{2} \varphi$ \ $\sigma, \varphi \seqder{3} \psi$}
\rlab{\rulecbid}
\uinf{$\G, \rho \cb \sigma \seq \varphi \cb \psi$}
\disp
\
\ax{$\eta \seqder{4} \varphi$ \ $\varphi \seqder{5} \eta$ \  $\eta \seqder{6} \xi$ \ $\xi \seqder{7} \eta$ \ $\psi, \vartheta, \eta \seqder{8} \chi$}
\rlab{\rulecdid}
\uinf{$\G', \varphi \cb \psi, \eta \cd \vartheta \seq \xi \cd \chi$}
\disp
\end{small}
\end{center}
The consequence $\G, \G', \rho \cb \sigma, \eta \cd \vartheta \seq \xi \cd \chi$  of $\cut$ can be derived as follows,
where dotted lines represent applications of the induction hypothesis:

\begin{center}
\begin{small}
\ax{$\eta \seqder{4} \varphi$}
\ax{$\sigma, \varphi \seqder{3} \psi$}
\ax{$\psi, \vartheta, \eta \seqder{8} \chi$}
\dottedLine
\binf{$\sigma, \varphi, \vartheta, \eta \seq \chi$}
\dottedLine
\binf{$\sigma, \eta, \vartheta, \eta \seq \chi$}
\rlab{$\ctr$}
\uinf{$\sigma, \eta, \vartheta \seq \chi$ \ ($A$)}
\disp

\vspace{0.3cm}
\ax{$\eta \seqder{4} \varphi$ \ \ $\varphi \seqder{1} \rho$}
\dottedLine
\uinf{$\eta \seq \rho$}
\ax{$\rho \seqder{2} \varphi$ \ \ $\varphi \seqder{5} \eta$}
\dottedLine
\uinf{$\rho \seq \eta$}
\ax{$\eta \seqder{6} \xi$ \ $\xi \seqder{7} \eta$} 
\ax{($A$)}
\rlab{\rulecdid}
\qinf{$\G, \G', \rho \cb \sigma, \eta \cd \vartheta \seq \xi \cd \chi$}
\disp
\end{small}
\end{center}

$\bullet$ The last rules applied in the derivations of the premisses of $\cut$ are, respectively, \rulecbid\ and \rulecbdid:

\begin{center}
\begin{small}
\ax{$\varphi \seqder{1} \rho$ \ \ $\rho \seqder{2} \varphi$ \ \ $\sigma, \varphi \seqder{3} \psi$}
\rlab{\rulecbid}
\uinf{$\G, \rho \cb \sigma \seq \varphi \cb \psi$}
\disp
\quad
\ax{$\eta \seqder{4} \varphi$ \ \ $\varphi \seqder{5} \eta$ \ \ $\psi, \vartheta, \eta \seqder{6}$}
\rlab{\rulecbdid}
\uinf{$\G', \varphi \cb \psi, \eta \cd \vartheta \seq \D$}
\disp
\end{small}
\end{center}

The consequence $\G, \G', \rho \cb \sigma, \eta \cd \vartheta \seq \D$ of $\cut$ can be derived as follows,
where dotted lines represent applications of the induction hypothesis:

\begin{center}
\begin{small}
\ax{$\eta \seqder{4} \varphi$ \ \ $\varphi \seqder{1} \rho$}
\dottedLine
\uinf{$\eta \seq \rho$}
\ax{$\rho \seqder{2} \varphi$ \ \ $\varphi \seqder{5} \eta$}
\dottedLine
\uinf{$\rho \seq \eta$}
\ax{$\eta \seqder{4} \varphi$}
\ax{$\sigma, \varphi \seqder{3} \psi$}
\ax{$\psi, \vartheta, \eta \seqder{6}$}
\dottedLine
\binf{$\sigma, \varphi, \vartheta, \eta \seq$}
\dottedLine
\binf{$\sigma, \eta, \vartheta, \eta \seq$}
\rlab{$\ctr$}
\uinf{$\sigma, \eta, \vartheta \seq$}
\rlab{\rulecbdid}
\tinf{$\G, \G', \rho \cb \sigma, \eta \cd \vartheta \seq \xi \cd \chi$}
\disp
\end{small}
\end{center}

$\bullet$ The remaining cases \rulecdid\ - \rulecdid\ and \rulecdid\ - \rulecbdid\ are similar. 
\qed
\end{proof}

\subsection{Cut admissibility in $\SCCKCEM$}
\label{app:cutel:cem}
\begin{theorem}
The rule $\cut$ is admissible in $\SCCKID$.
\end{theorem}
\begin{proof}
The proof is almost identical to that of $\SCCK$, with every occurrence of a $\cd$-formula in the left-hand-side of a sequent 
replaced with $k \geq 1$ $\cd$-formulas.
As an example, suppose that the cut formula is principal in the last rules applied in the derivations of both premisses of $\cut$, and the last rule applied are, respectively, \rulecdcem\ and \rulecbdcem:
\begin{center}
\begin{small}
\ax{$\varphi _1\seqder{1} \rho$ \ \ $\rho \seqder{2} \varphi$ \ \ $\varphi \seqder{3} \eta_1$ \ \ $\eta_1 \seqder{4} \varphi$ \ \ $\sigma, \psi \seqder{5} \vartheta_1$}
\rlab{\rulecdcem}
\uinf{$\G, \rho \cb \sigma, \varphi \cd \psi, \seq \eta_1 \cd \vartheta_1$}
\disp

\vspace{0.3cm}
\ax{$\eta_1 \seqder{6} \xi$ \ \ $\xi \seqder{7} \eta_1$ \ \ $\eta_1 \seqder{8} \eta_2$ \ \ $\eta_2 \seqder{9} \eta_1$ \ \ $\chi, \vartheta_1, \vartheta_2 \seqder{10}$}
\rlab{\rulecbdcem}
\uinf{$\G', \xi \cb \chi, \eta_1 \cd \vartheta_1, \eta_2 \cd \vartheta_2 \seq \D$}
\disp
\end{small}
\end{center}
The consequence $\G, \G', \rho \cb \sigma, \xi \cb \chi, \varphi \cd \psi, \eta_2 \cd \vartheta_2 \seq \D$  of $\cut$ can be derived as follows,
where dotted lines represent applications of the induction hypothesis:
%
\begin{center}
\begin{small}
\ax{$\varphi \seqder{3} \eta_1$ \ \ $\eta_1 \seqder{6} \xi$}
\dottedLine
\uinf{$\varphi \seq \xi$ \ ($A$)}
\disp
\hfill
\ax{$\xi \seqder{7} \eta_1$ \ \ $\eta_1 \seqder{4} \varphi$}
\dottedLine
\uinf{$\xi \seq \varphi$ \ ($B$)}
\disp
\hfill
\ax{$\eta_2 \seqder{9} \eta_1$ \ \ $\eta_1 \seqder{4} \varphi$}
\dottedLine
\uinf{$\eta_2 \seq \varphi$ \ ($C$)}
\disp
\hfill
\ax{$\varphi \seqder{3} \eta_1$ \ \ $\eta_1 \seqder{8} \eta_2$}
\dottedLine
\uinf{$\varphi \seq \eta_2$ \ ($D$)}
\disp

\vspace{0.3cm}
\ax{$\varphi \seqder{1} \rho$ \ \ \ $\rho \seqder{2} \varphi$
\ \ \ ($A$) \ \ \ ($B$) \ \ \ ($C$) \ \ \ ($D$)}
\ax{$\sigma, \psi \seqder{5} \vartheta_1$ \ \ $\chi, \vartheta_1, \vartheta_2 \seqder{10}$}
\dottedLine
\uinf{$\sigma, \chi, \psi, \vartheta_2 \seq$}
\rlab{\rulecbdcem}
\binf{$\G, \G', \rho \cb \sigma, \xi \cb \chi, \varphi \cd \psi, \eta_2 \cd \vartheta_2 \seq \D$}
\disp
\end{small}
\end{center}
\qed
\end{proof}

\section{Derivations in sequent calculi}\label{app:der SCCKstar}
\subsection{Derivations in $\SCCK$}\label{app:der SCCK}

1. The sequent $\G, \varphi \seq \varphi$ is derivable for every $\varphi\in\lan$. 
By induction on the construction of $\varphi$.
For example, if $\varphi = \rho\cd\sigma$:
\medskip
\begin{center}
\ax{}
\dottedLine
\llab{$i.h.$}
\uinf{$\rho \seq \rho$}
\ax{}
\dottedLine
\llab{$i.h.$}
\uinf{$\rho \seq \rho$}
\ax{}
\dottedLine
\llab{$i.h.$}
\uinf{$\sigma \seq \sigma$}
\rlab{\rulecd}
\tinf{$\G, \rho \cd \sigma \seq \rho \cd \sigma$}
\disp
\end{center}

\bigskip
\noindent
2. Derivation of \axCMbox:
\begin{center}
\ax{$\varphi \sseq \varphi$}
\ax{$\psi \land \chi \seq \psi$}
\llab{\rulecb}
\binf{$\varphi \cb \psi \land \chi \seq \varphi \cb \psi$}
\ax{$\varphi \sseq \varphi$}
\ax{$\psi \land \chi \seq \chi$}
\rlab{\rulecb}
\binf{$\varphi \cb \psi \land \chi \seq \varphi \cb \chi$}
\rlab{\rland}
\binf{$\varphi \cb \psi \land \chi \seq (\varphi \cb \psi) \land (\varphi \cb \chi)$}
\rlab{\rimp}
\uinf{$\seq (\varphi \cb \psi \land \chi) \imp (\varphi \cb \psi) \land (\varphi \cb \chi)$}
\disp
\end{center}

\bigskip
\noindent
3. Derivation of \axCCbox:
\begin{center}
\ax{$\varphi \sseq \varphi$}
\ax{$\psi, \chi \seq \psi \land \chi$}
\rlab{\rulecb}
\binf{$\varphi \cb \psi, \varphi \cb \chi \seq \varphi \cb \psi \land \chi$}
\rlab{\lland}
\uinf{$(\varphi \cb \psi) \land (\varphi \cb \chi) \seq \varphi \cb \psi \land \chi$}
\rlab{\rimp}
\uinf{$\seq (\varphi \cb \psi) \land (\varphi \cb \chi) \imp (\varphi \cb \psi \land \chi)$}
\disp
\end{center}


\bigskip
\noindent
4. Derivation of \axCNbox:
\begin{center}
\ax{$\seq \top$}
\rlab{\rulecb}
\uinf{$\seq \varphi \cb \top$}
\disp
\end{center}

\bigskip
\noindent
5. Derivation of \axCNdiam:
\begin{center}
\ax{$\bot \seq$}
\rlab{\rulecd}
\uinf{$\varphi \cd \bot \seq \bot$}
\rlab{\rimp}
\uinf{$\seq (\varphi \cd \bot) \imp \bot$}
\disp
\end{center}

\bigskip
\noindent
6. Derivation of \axCKdiam:
\begin{center}
\ax{$\varphi \sseq \varphi$}
\ax{$\psi \imp \chi, \psi \seq \chi$}
\rlab{\rulecd}
\binf{$\varphi \cb (\psi \imp \chi), \varphi \cd \psi \seq \varphi \cd \chi$}
\rlab{\rimp}
\uinf{$\varphi \cb (\psi \imp \chi) \seq (\varphi \cd \psi) \imp (\varphi \cd \chi)$}
\rlab{\rimp}
\uinf{$\seq (\varphi \cb (\psi \imp \chi)) \imp ((\varphi \cd \psi) \imp (\varphi \cd \chi))$}
\disp
\end{center}

\bigskip
\noindent
7. Admissibility of \RAbox\ (similar proofs for \RAdiam, \RCbox, \RCdiam):
\begin{center}\begin{small}
\ax{$\seq (\varphi \imp \rho) \land (\rho \imp \varphi)$}
\dottedLine
\llab{$inv$\rland}
\uinf{$\seq \varphi \imp \rho$}
\dottedLine
\llab{$inv$\rimp}
\uinf{$\varphi \seq \rho$}
\ax{$\seq (\varphi \imp \rho) \land (\rho \imp \varphi)$}
\dottedLine
\llab{$inv$\rland}
\uinf{$\seq \rho \imp \varphi$}
\dottedLine
\llab{$inv$\rimp}
\uinf{$\rho \seq \varphi$}
\ax{$\psi \seq \psi$}
\rlab{\rulecb}
\tinf{$\varphi \cb \psi \seq \rho \cb \psi$}
\rlab{\rimp}
\uinf{$\seq (\varphi \cb \psi) \imp (\rho \cb \psi)$}
\disp
\end{small}
\end{center}

\subsection{Derivations in $\SCCKID$}

Derivation of \axIDbox:
\begin{center}
\ax{$\varphi \seq \varphi$}
\rlab{\rulecbid}
\uinf{$\seq \varphi \cb \varphi$}
\disp
\end{center}

\subsection{Derivations in $\SCCKMP$}

1. Derivation of \axMPbox:
\begin{center}
\ax{$\varphi \cb \psi, \varphi \seq \varphi$}
\ax{$\varphi \cb \psi, 	\psi \seq \psi$}
\rlab{\rulempbox}
\binf{$\varphi \cb \psi, \varphi \seq \psi$}
\rlab{\rimp}
\uinf{$\varphi \cb \psi \seq \varphi \imp \psi$}
\rlab{\rimp}
\uinf{$\seq (\varphi \cb \psi) \imp (\varphi \imp \psi)$}
\disp
\end{center}

\bigskip
\noindent
2. Derivation of \axMPdiam:
\begin{center}
\ax{$\varphi \seq \varphi$}
\ax{$\psi \seq \psi$}
\rlab{\rulempdiam}
\binf{$\varphi, \psi \seq \varphi \cd \psi$}
\rlab{\lland}
\uinf{$\varphi \land \psi \seq \varphi \cd \psi$}
\rlab{\rimp}
\uinf{$\seq \varphi \land \psi \imp (\varphi \cd \psi)$}
\disp
\end{center}

\subsection{Derivations in $\SCCKCEM$}\label{app:der CEM}

Derivation of \axCEMdiam:
\begin{center}
\ax{$\varphi \sseq \varphi$}
\ax{$\varphi \sseq \varphi$}
\ax{$\psi, \chi \seq \psi \land \chi$}
\rlab{\rulecdcem}
\tinf{$\varphi \cd \psi, \varphi \cd \chi \seq \varphi \cd \psi \land \chi$}
\rlab{\lland}
\uinf{$(\varphi \cd \psi) \land (\varphi \cd \chi) \seq \varphi \cd \psi \land \chi$}
\rlab{\rimp}
\uinf{$\seq (\varphi \cd \psi) \land (\varphi \cd \chi) \imp (\varphi \cd \psi \land \chi)$}
\disp
\end{center}

\section{Derivations in axiomatic systems}\label{app:der CCKstar}
\subsection{Derivations in $\CCK$}\label{app:der CCK}

\noindent
1. Derivation of \RMbox:

\begin{center}
\begin{tabular}{lr}
$\psi \imp \vartheta$ \\
$\psi \biimp \psi \land \vartheta$ & (1) \\
$(\varphi \cb \psi) \imp (\varphi \cb \psi \land \vartheta)$ \ \ & (2, \RCbox) \\
$(\varphi \cb \psi) \imp (\varphi \cb \vartheta)$ & (3, \axCMbox) \\
\end{tabular} 
\end{center}

\bigskip
\noindent
2. Derivation of \RMdiam:

\begin{center}
\begin{tabular}{lr}
$\psi \imp \vartheta$ \\
$\psi \biimp \psi \land \vartheta$ & (1) \\
$(\varphi \cb \psi) \imp (\varphi \cb \psi \land \vartheta)$ \ \ & (2, \RCbox) \\
$(\varphi \cb \psi) \imp (\varphi \cb \vartheta)$ & (3, \axCMbox) \\
\end{tabular} 
\end{center}

\bigskip
\noindent
3. Derivation of \axCW:

\begin{center}
\begin{tabular}{lr}
$\chi \imp (\psi \imp \psi \land \chi)$ \\
$(\varphi \cb \chi) \imp (\varphi \cb \psi \imp \psi \land \chi)$ & (1, \RMbox) \\
$(\varphi \cb \psi \imp \psi \land \chi) \imp (\varphi \cd \psi \imp \varphi \cd \psi \land \chi)$ \ \ & (\axCKdiam) \\
$(\varphi \cd \psi)\land(\varphi\cb\chi) \imp (\varphi\cd \psi\land\chi)$ & (2,3) \\
\end{tabular}
\end{center}

\bigskip
\noindent
4. Derivation of the rule
\begin{center}
\ax{$\varphi \biimp \rho_1$}
\ax{$\varphi \biimp \rho_2$}
\ax{$\varphi \biimp \eta$}
\ax{$\sigma_1 \land \sigma_2 \land \psi \imp \vartheta$}
\rlab{(\rulecd)}
\qinf{$(\rho_1 \cb \sigma_1) \land (\rho_2 \cb \sigma_2) \land (\varphi \cd \psi) \imp (\eta \cd \vartheta)$}
\disp:
\end{center}

\begin{center}
\begin{small}
\begin{tabular}{lr}
$\varphi \biimp \rho_1$ \\

$(\rho_1 \cb \sigma_1) \imp (\varphi \cb \sigma_1)$ & (1, \RAbox) \\

$\varphi \biimp \rho_2$ \\

$(\rho_2 \cb \sigma_2) \imp (\varphi \cb \sigma_2)$ & (3, \RAbox) \\

$(\rho_1 \cb \sigma_1) \land (\rho_2 \cb \sigma_2) \land (\varphi \cd \psi) \imp (\varphi \cb \sigma_1) \land (\varphi \cb \sigma_2) \land (\varphi \cd \psi)$ \ \ & (2, 4)\\

$(\varphi \cb \sigma_1) \land (\varphi \cb \sigma_2) \land (\varphi \cd \psi) \imp (\varphi \cd \sigma_1 \land \sigma_2 \land \psi)$ & (\axCCbox, \axCW)\\

$\sigma_1 \land \sigma_2 \land \psi \imp \vartheta$ \\

$(\varphi \cd \sigma_1 \land \sigma_2 \land \psi) \imp (\varphi \cd \vartheta)$ & (7, \RMdiam) \\

$\varphi \biimp \eta$ \\

$(\varphi \cd \vartheta) \imp (\eta \cd \vartheta)$ & (9, \RAdiam) \\

$(\rho_1 \cb \sigma_1) \land (\rho_2 \cb \sigma_2) \land (\varphi \cd \psi) \imp (\eta \cd \vartheta)$ & (5, 6, 8, 10)
\end{tabular}
\end{small}
\end{center}

\bigskip
\noindent
5. Derivations of Hilbert rules corresponding to \rulecb\ and \rulecbd\ are similar to \rulecd.

\subsection{Derivations in $\CCKID$}

Derivation of the rule
\begin{center}
\ax{$\varphi \biimp \rho_1$}
\ax{$\varphi \biimp \rho_2$}
\ax{$\varphi \biimp \eta$}
\ax{$\sigma_1 \land \sigma_2 \land \varphi \land \psi \imp \vartheta$}
\rlab{(\rulecdid)}
\qinf{$(\rho_1 \cb \sigma_1) \land (\rho_2 \cb \sigma_2) \land (\varphi \cd \psi) \imp (\eta \cd \vartheta)$}
\disp:
\end{center}

\begin{center}
\begin{small}
\begin{tabular}{lr}
$\varphi \biimp \rho_1$ \\

$(\rho_1 \cb \sigma_1) \imp (\varphi \cb \sigma_1)$ & (1, \RAbox) \\

$\varphi \biimp \rho_2$ \\

$(\rho_2 \cb \sigma_2) \imp (\varphi \cb \sigma_2)$ & (3, \RAbox) \\

$(\rho_1 \cb \sigma_1) \land (\rho_2 \cb \sigma_2) \land (\varphi \cd \psi) \imp (\varphi \cb \sigma_1) \land (\varphi \cb \sigma_2) \land (\varphi \cd \psi)$ \ \ & (2, 4)\\

$(\varphi \cb \sigma_1) \land (\varphi \cb \sigma_2) \land (\varphi \cd \psi) \imp (\varphi \cb \sigma_1) \land (\varphi \cb \sigma_2) \land (\varphi \cb \varphi) \land (\varphi \cd \psi)$ & (\axIDbox)\\

$(\varphi \cb \sigma_1) \land (\varphi \cb \sigma_2) \land (\varphi \cb \varphi) \land (\varphi \cd \psi) \imp (\varphi \cd \sigma_1 \land \sigma_2 \land \varphi \land \psi)$ & (\axCCbox, \axCW)\\

$\sigma_1 \land \sigma_2 \land \varphi \land \psi \imp \vartheta$ \\

$(\varphi \cd \sigma_1 \land \sigma_2 \land \varphi \land \psi) \imp (\varphi \cd \vartheta)$ & (8, \RMdiam) \\

$\varphi \biimp \eta$ \\

$(\varphi \cd \vartheta) \imp (\eta \cd \vartheta)$ & (10, \RAdiam) \\

$(\rho_1 \cb \sigma_1) \land (\rho_2 \cb \sigma_2) \land (\varphi \cd \psi) \imp (\eta \cd \vartheta)$ & (5, 7, 9, 11)
\end{tabular}
\end{small}
\end{center}

\subsection{Derivations in $\CCKMP$}

1. Derivation of the rule
\begin{center}
\ax{$(\varphi \cb \psi) \imp \varphi$}
\ax{$(\varphi \cb \psi) \land \psi \imp \chi$}
\rlab{(\rulempbox)}
\binf{$(\varphi \cb \psi) \imp \chi$}
\disp:
\end{center}

\begin{center}
\begin{tabular}{lr}
$(\varphi \cb \psi) \imp \varphi$ \\

$(\varphi \cb \psi) \imp \varphi \land (\varphi \imp \psi)$ \ \ & (1, \axIDbox) \\

$(\varphi \cb \psi) \imp \psi$ & (2) \\

$(\varphi \cb \psi) \land \psi \imp \chi$ \\

$(\varphi \cb \psi) \imp \chi$ & (3, 4) \\
\end{tabular}
\end{center}

\bigskip
\noindent
2. Derivation of the rule
\begin{center}
\ax{$\xi \imp \varphi$}
\ax{$\xi \imp \psi$}
\rlab{(\rulempdiam)}
\binf{$\xi \imp (\varphi \cd \psi)$}
\disp:
\end{center}

\begin{center}
\begin{tabular}{lr}
$\xi \imp \varphi$ \\

$\xi \imp \psi$ \\

$\xi \imp \varphi \land \psi$ & (1, 2) \\

$\xi \imp (\varphi \cd \psi)$ \ \ & (3, \axMPdiam) \\ 
\end{tabular}
\end{center}

\subsection{Derivations in $\CCKCEM$}
\label{app:sound:cem}
Derivation of the rule
\begin{center}
\begin{small}
\ax{$\varphi_1 \biimp \rho_1$}
\ax{$\varphi_1 \biimp \rho_2$}
\ax{$\varphi_1 \biimp \varphi_2$}
\ax{$\varphi_1 \biimp \eta$}
\ax{$\sigma_1 \land \sigma_2 \land \psi_1 \land \psi_2 \imp \vartheta$}
\rlab{(\rulecdcem)}
\QuinaryInfC{$(\rho_1 \cb \sigma_1) \land (\rho_2 \cb \sigma_2) \land (\varphi_1 \cd \psi_1) \land (\varphi_2 \cd \psi_2) \imp (\eta \cd \vartheta)$}
\disp:
\end{small}
\end{center}

\begin{center}
\begin{small}
\begin{tabular}{lr}
$\varphi_1 \biimp \rho_1$ \\

$(\rho_1 \cb \sigma_1) \imp (\varphi_1 \cb \sigma_1)$ & (1, \RAbox) \\

$\varphi_1 \biimp \rho_2$ \\

$(\rho_2 \cb \sigma_2) \imp (\varphi_1 \cb \sigma_2)$ & (3, \RAbox) \\

$\varphi_1 \biimp \varphi_2$ \\

$(\varphi_2 \cd \psi_2) \imp (\varphi_1 \cd \psi_2)$ & (5, \RAdiam) \\

$(\rho_1 \cb \sigma_1) \land (\rho_2 \cb \sigma_2) \land (\varphi_1 \cd \psi_1) \land (\varphi_2 \cd \psi_2) \imp$ \\

\hfill $(\varphi_1 \cb \sigma_1) \land (\varphi_1 \cb \sigma_2) \land (\varphi_1 \cd \psi_1) \land (\varphi_1 \cd \psi_2)$ \ \ & (2, 4, 6)\\

$(\varphi_1 \cb \sigma_1) \land (\varphi_1 \cb \sigma_2) \land (\varphi_1 \cd \psi_1) \land (\varphi_1 \cd \psi_2) \imp$ \\ 
\hfill $(\varphi_1 \cd \sigma_1 \land \sigma_2 \land \psi_1 \land \psi_2)$ \ \ & (\axCCbox, \axCW, \axCEMdiam)\\

$\sigma_1 \land \sigma_2 \land \psi_1 \land \psi_2 \imp \vartheta$ \\

$(\varphi_1 \cd \sigma_1 \land \sigma_2 \land \psi_1 \land \psi_2) \imp (\varphi_1 \cd \vartheta)$ & (9, \RMdiam) \\

$\varphi_1 \biimp \eta$ \\

$(\varphi_1 \cd \vartheta) \imp (\eta \cd \vartheta)$ & (11, \RAdiam) \\

$(\rho_1 \cb \sigma_1) \land (\rho_2 \cb \sigma_2) \land (\varphi_1 \cd \psi_1) \land (\varphi_2 \cd \psi_2) \imp (\eta \cd \vartheta)$ \qquad \qquad  & (7, 8, 10, 12)
\end{tabular}
\end{small}
\end{center}

\section{Semantic completeness of $\ConstCK$}\label{app:sem}

\begin{proposition}[Soundness]
For $\varphi\in\lan$, if $\varphi$ is derivable in $\CCK$, then $\varphi$ is valid in all \ccm s.
\end{proposition}
\begin{proof}
By showing that all axioms and rules of $\CCK$ are valid, respectively validity-preserving, in every \ccm.
Validity of intuitionistic axioms is obvious because \ccm s are a particular case of intuitionistic Kripke models.
We show the following examples.

(\RAdiam) Suppose $\M \models \varphi \biimp \rho$, that is $\trset{\varphi} = \trset{\rho}$. 
If $w \Vd \varphi \cd \psi$, then for all $w' \more w$, there is $v$ such that $w' \Rof{\trset{\varphi}} v$ and $v \Vd \psi$.
Then $w' \Rof{\trset{\rho}} v$, therefore  $w \Vd \rho \cd \psi$.

(\RCdiam) Suppose $\M \models \psi \biimp \chi$, that is for all $z\in\W$, $z \Vd \psi$ iff $z \Vd \chi$. 
If $w \Vd \varphi \cd \psi$, then for all $w' \more w$, there is $v$ such that $w' \Rof{\trset{\varphi}} v$ and $v \Vd \psi$.
Then $v \Vd \chi$, therefore  $w \Vd \varphi \cd \chi$.

(\axCKdiam) Suppose (i) $w \Vd \varphi \cb (\psi \imp \chi)$ and (ii) $w \Vd \varphi \cd \psi$.
By (ii), for all $w' \more w$, there is $v$ such that $w' \Rof{\trset{\varphi}} v$ and $v \Vd \psi$.
Then by (i), $v \Vd \chi$, therefore $w \Vd \varphi \cd \psi$.

\end{proof}

\lindenbaum*
\begin{proof}
The proof extends the one of \cite{segerbergMin} for $\IPL$ to the conditional language $\lan$
without essential modifications.
Consider an enumeration $\psi_0, \psi_1, \psi_2, ...$ of all formulas of $\lan$.
The set $\B$ is 
constructed as follows:
$\A_0 = \A$; 
$\A_{n+1}=\A_n\cup\{\psi_n\}$ if $\A_n\not\vd \psi_n \imp \varphi$,
and $\A_{n+1}=\A_n$ otherwise;
$\B = \bigcup_{n \in \mathbb{N}}\A_n$.
It follows that $\B\not\vd \varphi$,
otherwise there are $\chi_1,.., \chi_k\in\B$ such that
$\vd \chi_1 \land  ... \land \chi_k\imp \varphi$,
which means that there is $n\in \mathbb{N}$ such that $\chi_1,.., \chi_k\in\A_n$,
hence $\A_n\vd \varphi$,
against the construction hypothesis.
Then, $\varphi\notin \B$ (given  the validity of $\varphi \imp \varphi$ in $\IPL$).
Moreover, $\B$ is prime:
it is consistent because $\B\not\vd \varphi$.
It is deductively closed: suppose $\B\vd \chi$ and $\chi\notin\B$.
Then $\chi=\psi_n$ for some $n\in \mathbb{N}$, and, by construction, $\A_n \vd \chi\imp \varphi$.
Since $\A_n\subseteq\B$, it follows
$\B\vd \chi\imp \varphi$, hence $\B\vd  \varphi$,
giving a contradiction.
Finally, $\B$ satisfies the disjunction property:
suppose $\chi \lor \xi\in\B$, $\chi\notin\B$ and $\xi\notin\B$.
Then $\chi = \psi_i$ and $\xi = \psi_j$ for some $i,j \in \mathbb{N}$.
Let $n = max\{i,j\}$.
Then by construction, $\A_n \vd \chi \imp \varphi$ and $\A_n \vd \xi \imp \varphi$,
thus  $\A_n \vd \chi \lor \xi \imp \varphi$.
It follows $\B\vd \chi \lor \xi \imp \varphi$, and since $\chi \lor \xi \in \B$,
we have $\B\vd \varphi$,
therefore $\varphi \in \B$, giving a contradiction. 
\end{proof}

\lemmacanel*
\begin{proof}
First, observe that $\vd \varphi \biimp \rho$ 
is an equivalence relation $eq \subseteq \lan^2$.

(1.) 
Let $\A$ be a prime set.
We define $\ceR$ as a set of pairs $(\pset{\varphi}, \beta)$, where 
$\varphi$ is the canonical representative of some equivalence class induced by $eq$
such that  $\varphi\cd\chi \in \A$ for some $\chi$, 
and the corresponding set $\beta$ is constructed as follows:
If $\varphi\cd\chi \in \A$, then $\fcbm{\A} \cup \{\chi\} \not\vd \bot$,
otherwise there are $\varphi\cb\sigma_1, ..., \varphi\cb\sigma_n \in \A$ such that $\varphi\cb\sigma_1, ..., \varphi\cb\sigma_n, \varphi\cd\chi \vd \bot$,
against the consistency of $\A$.
Then by Lemma~\ref{lemma:lind}, there is $\B$ prime such that $\fcbm{\A} \cup \{\chi\} \subseteq \B$.
We call $\B$ \emph{witness} for $\varphi\cd\chi$,
and define $\beta$ as a set containing a witness for every $\varphi\cd\vartheta\in\A$.
One can easily verify that $(\A, \ceR)$ is a \ce.
In particular, if $\eta\cd\vartheta\in\A$, then $\vd \eta\biimp \varphi$, for some canonical representative $\varphi$.
Thus $\pset{\eta} = \pset{\varphi}$, and by rule \RAdiam\
and closure under derivation, 
$\varphi \cd \vartheta \in \A$,
hence there is a corresponding pair $(\pset{\varphi}, \beta) = (\pset{\eta}, \beta) \in \ceR$.

(2.) Consider the same construction as before, noticing that $\fcbm{\A} \not\vd \psi$,
otherwise $\A \vd \varphi\cb\psi$, hence $\varphi\cb\psi\in\A$.

(3.) We construct pairs $(\pset{\varphi}, \beta)$ similarly to case 1,
noticing that if $\varphi\cd\chi \in \A$, then $\fcbm{\A} \cup \{\chi\} \not\vd \psi$,
otherwise $\A \vd \varphi\cd\psi$, hence $\varphi\cd\psi\in\A$.
Thus by Lemma~\ref{lemma:lind}, there is $\B$ prime such that $\fcbm{\A} \cup \{\chi\} \subseteq \B$ and $\psi\notin\B$.
Then the proof proceeds as in 1.\qed
\end{proof}

\lemmamodcan*
\begin{proof}
Note that the claim is equivalent to $\truthset{\varphi} = \ceset{\varphi}$.
We prove the lemma by induction on the construction of $\varphi$,
showing only the cases $\varphi = \rho \cb \psi$, $\varphi = \rho \cd \psi$.

($\varphi = \rho \cb \psi$)
Assume $\rho \cb \psi \in \A$ and $(\A, \ceR) \lessc (\A', \ceR') \Rcof{\trset{\rho}} (\B, \ceU)$.
Then $\A \subseteq \A'$, which implies $\rho \cb \psi \in \A'$, and by \ih, $(\A', \ceR') \Rcof{\ceset{\rho}} (\B, \ceU)$.
By definition of $\Rcof{X}$, there is $\beta$ such that $(\pset{\rho}, \beta) \in \ceR'$ and $\B \in \beta$,
then by definition of \ce, $\psi\in\B$.
Thus by \ih, $(\B, \ceU)\Vd\psi$,
hence $(\A, \ceR) \Vd \rho \cb \psi$.
Now assume $\rho \cb \psi \notin \A$.
By Lemma~\ref{lemma:can el},
there are 
$(\A, \ceR)$, $\beta$, $\B$
such that 
$(\pset{\rho}, \beta)\in\ceR$, $\B\in\beta$ and $\psi\notin\B$.
Moroever, there exists a \ce\ $(\B, \ceU)$.
Then $(\A, \ceR) \Rcof{\ceset{\rho}} (\B, \ceU)$, 
and by \ih,
$(\A, \ceR) \Rcof{\trset{\rho}} (\B, \ceU)$ and $(\B, \ceU)\not\Vd\psi$.
Therefore $(\A, \ceR) \not\Vd \rho\cb\psi$.

($\varphi = \rho \cd \psi$)
Assume $\rho \cd \psi \in \A$ and $(\A, \ceR) \lessc (\A', \ceR')$.
Then $\rho \cd \psi \in \A'$.
By definition of \ce, there are $(\pset{\rho}, \beta) \in \ceR$ and $\B \in \beta$ such that $\psi\in\B$.
Moreover, by Lemma~\ref{lemma:can el}, there is a \ce\ $(\B, \ceU)$.
Hence, by definition, $(\A', \ceR') \Rcof{\ceset{\rho}} (\B, \ceU)$, and by \ih,
$(\A', \ceR') \Rcof{\trset{\rho}} (\B, \ceU)$ and $(\B, \ceU) \Vd \psi$.
Therefore $(\A, \ceR) \Vd \rho \cd \psi$.
Now assume $\rho \cd \psi \notin \A$.
By Lemma~\ref{lemma:can el},
there are 
$(\A, \ceR)$, $\beta$
such that 
$(\pset{\rho}, \beta)\in\ceR$ and for all $\B\in\beta$, $\psi\notin\B$.
Thus, for all $(\B, \ceU)$ such that $(\A, \ceR) \Rcof{\ceset{\rho}} (\B, \ceU)$, 
$\psi\notin\B$,
hence by \ih,
for all $(\B, \ceU)$ such that $(\A, \ceR) \Rcof{\trset{\rho}} (\B, \ceU)$, $(\B, \ceU) \not\Vd \psi$.
Therefore $(\A, \ceR) \not \Vd \rho\cd\psi$.\qed
\end{proof}

\end{document}